\documentclass[twoside, 11pt]{article}
\usepackage{jmlr2e}

\usepackage{lastpage}
\jmlrheading{21}{2020}{1-\pageref{LastPage}}{8/18; Revised
12/19}{1/20}{18-577}{Xin Tong, Lucy Xia, Jiacheng Wang and Yang Feng}
\ShortHeadings{Neyman-Pearson classification: parametrics and sample size requirement}{Tong, Xia, Wang and Feng}

\usepackage{diagbox}

\usepackage{graphicx,enumerate}
\usepackage{natbib}
\usepackage{delarray}
\usepackage{hyperref, subcaption}
\usepackage{rotating}
\usepackage{url}
\usepackage{colortbl,color}
\usepackage{anysize}
\usepackage{multirow,amssymb,amsmath, amsfonts}

\usepackage{mathrsfs}

\newtheorem{assumption}{Assumption}

\newcommand{\p}{{\rm I}\kern-0.18em{\rm P}}
\newcommand{\1}{{\rm 1}\kern-0.24em{\rm I}}
\newcommand{\E}{{\rm I}\kern-0.18em{\rm E}}
\newcommand{\uderbar}[1]{\underset{\raise0.3em\hbox{$\smash{\scriptscriptstyle-}$}}{#1}}
\def\boxit#1{\vbox{\hrule\hbox{\vrule\kern6pt\vbox{\kern6pt#1\kern6pt}\kern6pt\vrule}\hrule}}

\DeclareMathOperator*{\argmin}{arg\,min}

\begin{document}

\title{Neyman-Pearson classification: parametrics and sample size requirement}
\author{\name Xin Tong \email xint@marshall.usc.edu \\
       \addr Department of Data Sciences and Operations \\
       Marshall Business School\\
       University of Southern California
 \AND       
 \name Lucy Xia  \email lucyxia@ust.hk  \\
        \addr  Department of ISOM\\
School of Business and Management\\
Hong Kong University of Science and Technology
\AND 
\name Jiacheng Wang \email jiachengwang@galton.uchicago.edu\\
       \addr Department of Statistics\\
University of Chicago
\AND 
  \name Yang Feng  \email yang.feng@nyu.edu \\
       \addr Department of Biostatistics\\
       School of Global Public Health\\
       New York University
}

\editor{Xiaotong Shen}

\maketitle

\begin{abstract}
The Neyman-Pearson (NP) paradigm in binary classification seeks classifiers that achieve a minimal type II error while enforcing the prioritized type I error controlled under some user-specified level $\alpha$. This paradigm serves naturally in applications such as severe disease diagnosis and spam detection, where people have clear priorities among the two error types. Recently, \cite{tong2016np} proposed a nonparametric umbrella algorithm that adapts all scoring-type classification methods (e.g., logistic regression, support vector machines, random forest) to respect the given type I error (i.e., conditional probability of classifying a class $0$ observation as class $1$ under the 0-1 coding) upper bound $\alpha$ with high probability, without specific  distributional assumptions on the features and the responses. Universal the umbrella algorithm is, it demands an explicit minimum sample size requirement on class $0$, which is often the more scarce class, such as in rare disease diagnosis applications. In this work, we employ the parametric linear discriminant analysis (LDA) model and propose a new parametric thresholding algorithm, which does not need the minimum sample size requirements on class $0$ observations and thus is suitable for small sample applications such as rare disease diagnosis. Leveraging both the existing nonparametric and the newly proposed parametric thresholding rules, we propose four LDA-based NP classifiers, for both low- and high-dimensional settings. On the theoretical front, we prove NP oracle inequalities for one proposed classifier, where the rate for excess type II error benefits from the explicit parametric model assumption. Furthermore, as NP classifiers involve a sample splitting step of class $0$ observations,  we construct a new adaptive sample splitting scheme that can be applied universally to NP classifiers, and this adaptive strategy reduces the type II error of these classifiers.

\end{abstract}


\section{Introduction}

Classification aims to predict discrete outcomes (i.e., class labels) for new observations, using algorithms trained on labeled data. It is one of the most studied machine learning problems with  applications  including automatic disease diagnosis, email spam filters, and image classification. Binary classification, where the outcomes belong to one of two classes and the class labels are usually coded as $\{0, 1\}$ (or $\{-1, 1\}$ or $\{1, 2\}$), is the most common type.  Most binary classifiers are constructed to minimize the overall classification error (i.e., \textit{risk}), which is a weighted sum of type I and type II errors.  Here, \textit{type I error} is defined as the conditional probability of misclassifying a class $0$ observation as class $1$, and \textit{type II error} is the conditional probability of misclassifying a class $1$ observation as class $0$ \footnote{In verbal discussion, with a slight abuse of language, we also refer to the action of assigning a class $0$ observation to class $1$ as type I error, and that of assigning a class $1$ observation to class $0$ as type II error.}. In the following, we refer to this paradigm as the \textit{classical classification paradigm}. Along this line, numerous methods have been proposed, including linear discriminant analysis (LDA) in both low and high dimensions \citep{Guo.Hastie.ea.2005, Cai.Liu.2011, Shao.Wang.ea.2011, Witten.Tibshirani.2012, Fan.Feng.ea.2011,mai2012direct},  logistic regression, support vector machine (SVM) \citep{vapnik2013nature}, random forest \citep{Breiman.2001},   among others. 

In contrast, the \textit{Neyman-Pearson (NP) classification paradigm} \citep{CanHowHus02,ScoNow05,Rigollet.Tong.2011, Tong.2013, Zhao2016neyman, tong2016survey} was developed to seek a classifier that minimizes the type II error while maintaining the type I error below a user-specified level $\alpha$, usually a small value (e.g., $5\%$). We call this target classifier the NP oracle classifier. The NP paradigm is appropriate in applications such as cancer diagnosis, where a type I error (i.e., misdiagnosing a cancer patient to be healthy) has more severe consequences than a type II error (i.e., misdiagnosing a healthy patient as with cancer). The latter incurs  extra medical costs and patients' anxiety but will not result in the tragic loss of life, so it is appropriate to have type I error control as the priority.  Cost-sensitive learning, which assigns different costs as weights of type I and type II errors \citep{Elkan01, ZadLanAbe03} is a popular  alternative paradigm to address asymmetric errors. This approach has merits and many practical values. However, when there is no consensus to assign costs to errors, or in applications such as medical diagnosis, where it is morally unacceptable to do a cost and benefit analysis, the NP paradigm is a more natural choice.  Previous  NP classification literature use both empirical risk minimization (ERM) \citep{CanHowHus02,CasChe03,Sco05,ScoNow05,HanCheSun08,Rigollet.Tong.2011} and plug-in approaches \citep{Tong.2013,Zhao2016neyman}, and its genetic application is suggested in \cite{Li.Tong.2016}. More recently, \cite{tong2016np} took a different route, and proposed an NP umbrella algorithm that adapts scoring-type classification algorithms (e.g., logistic regression, support vector machines, random forest, etc.) to the NP paradigm, by setting nonparametric order statistics based thresholds on the classification scores.  To implement the NP paradigm, it is very tempting to simply tune the empirical type I error to (no more than) $\alpha$.  Nevertheless, as argued extensively in \cite{tong2016np}, doing so would not lead to classifiers whose type I errors are bounded from above by $\alpha$ with high probability. 

Universal the NP umbrella algorithm is, it demands an explicit minimum sample size requirement on class $0$, which is usually the more scarce class. While this requirement is not too stringent, it does cause a problem when the smaller class $0$ has an insufficient sample size.  For instance, a commonly-used lung cancer diagnosis example \citep{Gordon:2002wr} in the high-dimensional statistics literature has $181$ subjects and $12,533$ features. These subjects carry two different types of lung cancer, namely, adenocarcinoma ($150$ subjects) and mesothelioma ($31$ subjects). If we were to treat mesothelioma as the class $0$, and would like to control the type I error under $\alpha = .05$ with high probability $1-\delta_0 = .95$, then the sample size requirement of the NP umbrella algorithm for the left-out class $0$ observations is $59$ (or $\log \delta_0 / \log (1-\alpha)$ to be exact), which is almost twice as much as the total available class $0$ sample size $31$.    
 In this work, we employ the linear discriminant analysis (LDA) model and propose a new thresholding algorithm, which does not have the minimum sample size requirements on class $0$ observations and thus applies to small sample applications such as rare disease diagnosis.

In total, we develop four NP classifiers in this paper: combinations of two different ways to create scoring functions in low- and high-dimensional settings under the LDA model, and two thresholding methods that include the nonparametric order statistics based NP umbrella algorithm and the newly proposed parametric thresholding rule. We will denote these classifiers by \verb+NP-LDA+, \verb+NP-sLDA+, \verb+pNP-LDA+ and \verb+pNP-sLDA+, where \verb+p+ means ``parametric thresholding" and \verb+s+ means ``sparse". Extensive numerical experiments will suggest the recommended application domains of these methods.

On the theoretical front, NP oracle inequalities, a core theoretical criterion to evaluate classifiers under the NP paradigm,  were only established for nonparametric plug-in classifiers \citep{Tong.2013, Zhao2016neyman}. The current work is the first to establish NP oracle inequalities under parametric models.  Concretely, we will show that \verb+NP-sLDA+ satisfies the NP oracle inequalities.  Another major contribution of this work is that  we design an adaptive sample splitting scheme that can be applied universally to existing NP classifiers. This adaptive strategy enhances the power (i.e., reduces type II error) and therefore raises the practicality of the NP algorithms.

The rest of this paper is organized as follows.  Section $2$ introduces the notations and model setup. Section $3$ constructs the new parametric thresholding rule based on the LDA model and introduces four new LDA based NP classifiers. Section $4$ formulates theoretical conditions and derives NP oracle inequalities for \verb+NP-sLDA+. 
Section $5$ describes the new data-adaptive sample splitting scheme.    Numerical results are presented in Section $6$, followed by a short discussion in Section $7$.  Longer proofs, additional numerical results, and other supporting intermediate results are relegated to the appendices.

\section{Notations and model setup}
A few standard notations are introduced 
to facilitate our discussion.  
Let $(X, Y)$ be a random pair where $X \in \mathcal{X} \subset \mathbb{R}^d$ is a $d$-dimensional vector of features, and $Y \in \{0,1\}$ indicates $X$'s class label.  Denote respectively by $\p$ and $\E$ generic probability distribution and expectation. 
A \emph{classifier} $\phi : \mathcal{X} \to \{0,1\}$ is a data-dependent mapping from $\mathcal{X}$ to $\{0,1\}$ that assigns $X$ to one of the classes.  
The overall classification error of $\phi$ is
$
R(\phi)=\E \1\{\phi(X)\neq Y\} = \p\left\{\phi(X)\neq Y\right\}
$, where $\1(\cdot)$ denotes the indicator function. 
By the law of total probability, $R(\phi)$ can be decomposed into a weighted average of type~I error $R_0(\phi)=\p \left\{\phi(X)\neq Y|Y=0\right\}$ and type~II error $R_1(\phi)=\p \left\{\phi(X)\neq Y|Y=1\right\}$ as
\begin{equation}\label{EQ:risk break down}
R(\phi) \,=\, \pi_0 R_0(\phi)+\pi_1 R_1(\phi)\,,
\end{equation}
where $\pi_0 = \p(Y=0)$ and $\pi_1 = \p(Y=1)$.
While the classical paradigm minimizes $R(\cdot)$, the Neyman-Pearson (NP) paradigm seeks to minimize $R_1$ while controlling $R_0$ under a user-specified level $\alpha$.
The (level-$\alpha$) \emph{NP oracle classifier} is thus
\begin{equation}
\label{eq::goal}
\phi^*_{\alpha} \,\in\, \argmin_{R_0(\phi)\leq\alpha}R_1(\phi)\,,
\end{equation}
where the \emph{significance level} $\alpha$ reflects the level of conservativeness towards type I error.

In this paper, we assume that  $(X|Y=0)$ and $(X|Y=1)$ follow multivariate Gaussian distributions with a common covariance matrix. That is, their probability density functions $f_0$ and $f_1$ are  
$$f_0 \sim \mathcal{N}(\mu^0, \Sigma) \text{ and }f_1 \sim \mathcal{N}(\mu^1, \Sigma)\,,$$ 
where the mean vectors $\mu^0, \mu^1 \in \mathbb{R}^d$ ($\mu^0 \neq \mu^1$) and the common positive definite covariance matrix $\Sigma \in \mathbb{R}^{d\times d}$. This model is frequently referred to as the linear discriminant analysis (LDA) model. Despite its simplicity, the LDA model has been proved to be effective in many applications and benchmark datasets. Moreover, in the last ten years, several papers \citep{Shao.Wang.ea.2011, Cai.Liu.2011, Fan.Feng.ea.2011, Witten.Tibshirani.2012,mai2012direct} have developed LDA based algorithms  under high-dimensional settings where the dimensionality of features  is comparable to or larger than the sample size.    %Therefore, we choose the LDA model as our first attempt under parametric settings and conduct theoretical analysis on NP classifiers based on this simple but powerful model. %\textcolor{red}{Add some literature on LDA here}

It is well known that the Bayes classifier (i.e., oracle classifier) of the classical paradigm is $\phi^*(x) = \1(\eta(x) > 1/2)$, where $\eta(x) = \E(Y|X=x) = \p(Y=1|X=x)$ is the regression function.  Since 
$$
\eta(x) = \frac{\pi_1\cdot f_1(x)/f_0(x)}{\pi_1\cdot f_1(x)/f_0(x) + \pi_0}\,,
$$  
the oracle classifier can be written alternatively as $\1(f_1(x)/f_0(x) > \pi_0 / \pi_1)$.  When $f_0$
 and $f_1$ follow the LDA model, the oracle classifier of the classical paradigm is
\begin{equation}\label{eqn:classical_oracle}
\phi^*(x) =\1\left\{(x-\mu_a)^{\top}\Sigma^{-1}\mu_d+\log {\frac{\pi_1}{\pi_0}} > 0\right\} = \1\left\{ (\Sigma^{-1}\mu_d)^{\top} x > \mu_a^{\top} \Sigma^{-1}\mu_d - \log \frac{\pi_1}{\pi_0}\right\}\,,
\end{equation}
where $\mu_a = \frac{1}{2}(\mu^0 + \mu^1)$, $\mu_d = \mu^1 - \mu^0$, and $(\cdot)^{\top}$ denotes the transpose of a vector.  In contrast, motivated by the famous \textit{Neyman-Pearson Lemma}  in hypothesis testing (attached in the Appendix \ref{NP lemma} for readers' convenience), the NP oracle classifier  is 
\begin{equation}\label{eqn:np_original}
\phi^*_{\alpha} (x) = \1\left\{\frac{f_1(x)}{f_0(x)} >  C_{\alpha}\right\}\,,
\end{equation}
for some threshold $C_{\alpha}$ such that $P_0\{f_1(X)/f_0(X)>C_{\alpha}\}\leq\alpha$ and $P_0\{f_1(X)/f_0(X)\geq C_{\alpha}\}\geq\alpha$, where $P_0$ is the conditional probability distribution of $X$ given $Y=0$ ($P_1$ is defined similarly).    

Under the LDA model assumption, the NP oracle classifier is  $\phi^*_{\alpha}(x)=\1( (\Sigma^{-1}\mu_d)^{\top}x >  C_{\alpha}^{**})$, where $C_{\alpha}^{**} = \log C_{\alpha} + \mu_a^{\top} \Sigma^{-1}\mu_d$. Denote by $\beta^{\text{Bayes}} = \Sigma^{-1}\mu_d$ and   $s^*(x)  = (\Sigma^{-1}\mu_d)^{\top} x = (\beta^{\text{Bayes}})^{\top} x$, then the NP oracle classifier \eqref{eqn:np_original} can be written as 
\begin{equation}\label{eqn:np_oracle}
\phi^*_{\alpha}(x)  = \1(s^*(x) > C_{\alpha}^{**})\,.
\end{equation}
In fact, leveraging further the LDA model assumption, $C^{**}_{\alpha}$ can be written out explicitly.  Note that when $X^j\sim \mathcal{N}(\mu^j, \Sigma)$, $(\Sigma^{-1}\mu_d)^{\top}X^j\sim \mathcal{N}(\mu_d^{\top} \Sigma^{-1} \mu^j, \mu_d^{\top}\Sigma^{-1}\mu_d)$, where $j = 0, 1$. Let $Z = \left\{(\Sigma^{-1}\mu_d)^{\top}X^0 - (\Sigma^{-1}\mu_d)^{\top}\mu^0\right\}\left(\mu_d^{\top}\Sigma^{-1}\mu_d\right)^{-1/2}$, then $Z\sim \mathcal{N}(0, 1)$.  As $\Sigma^{-1}$ is positive definite and $\mu^0\neq \mu^1$,  we have $\mu_d^{\top} \Sigma^{-1} \mu^0 < \mu_d^{\top} \Sigma^{-1}\mu^1$, and thus the level-$\alpha$ NP oracle is 

\begin{equation}\label{eqn:np_oracle_explicit}
\phi^*_{\alpha}(x) = \1(s^*(x) > C_{\alpha}^{**}) = \1\left((\Sigma^{-1}\mu_d)^{\top}x > \sqrt{\mu_d^{\top}\Sigma^{-1}\mu_d} \Phi^{-1}(1-\alpha) +  \mu_d^{\top} \Sigma^{-1} \mu^0\right)\,,
\end{equation} 
where $\Phi(\cdot)$ is the CDF of the univariate standard normal distribution $\mathcal{N}(0, 1)$.
This oracle holds independent of the feature dimensionality. In reality, we cannot expect that the type I error bound holds almost surely; instead, we can only hope that a classifier $\hat \phi_{\alpha}$ trained on a finite sample have $R_0(\hat \phi_{\alpha}) \leq \alpha$ with high probability.     We will construct multiple versions of  LDA-based $\hat\phi_{\alpha}$ to suit different application domains in the next section.

Other mathematical notations are introduced as follows.   For a general $m_1 \times m_2$ matrix $M$,  $\|M\|_{\infty} = \max_{i = 1, \ldots, m_1} \sum_{j=1}^{m_2} |M_{ij}|$, and $\|M\|$ denotes the operator norm.  For a vector $b$, $\|b\|_{\infty} = \max_j |b_j|$, $|b|_{\text{min}} = \min_{j} |b_j|$, and $\|b\|$ denotes the $L_2$ norm.  Let $A = \{j:  (\Sigma^{-1}\mu_d)_j \neq 0\}$, and $\mu^1_{A}$ be a sub-vector of $\mu^1$ of length $s := \text{cardinality}(A)$ that consists of the coordinates of $\mu^1$ in $A$ (similarly for $\mu^0_A$).  Up to permutation, the $\Sigma$ matrix can be written as 
$$
\Sigma 
=
\begin{bmatrix}
    \Sigma_{AA} &     \Sigma_{AA^c} \\
    \Sigma_{A^cA} &     \Sigma_{A^cA^c}
\end{bmatrix}\,.
$$
We use $\lambda_{\max}(\cdot)$ and $\lambda_{\min}(\cdot)$ to denote maximum and minimum eigenvalues of a matrix, respectively.  

%\textcolor{red}{Explain the P notations in this section.}

\section{Constructing LDA-based NP classifiers  $\hat\phi_{\alpha}(\cdot)$}

In this section, we construct two estimates of $s^*(\cdot)$ in Sections \ref{sec::s estimate 1} and \ref{sec::s estimate 2},  and two estimates of $C_{\alpha}^{**}$ in Sections \ref{sec::c alpha estimate 1} and \ref{sec::c alpha estimate 2}.   Hence in total, we present in Section \ref{sec::four classifiers} four LDA-based NP classifiers   $\hat \phi_{\alpha}(\cdot)$ for different application domains.  
We assume the following sampling scheme for the rest of this section and in the theoretical analysis.  
Let $\mathcal{S}_0 = \{X^0_{1}, \ldots, X^0_{n_0}\}$ be an i.i.d. class $0$ sample of size  $n_0$, $\mathcal{S}_0' = \{X^0_{n_0+1}, \ldots, X^0_{n_0 + n_0'}\}$ be an i.i.d. class $0$ sample of size  $n_0'$ and $\mathcal{S}_1 = \{X^1_{1}, \ldots, X^1_{n_1}\}$ be an i.i.d. class  $1$ sample of size  $n_1$. The example sizes $n_0$, $n_0'$ and $n_1$ are considered as fixed numbers. Moreover, we assume that samples $\mathcal{S}_0$, $\mathcal{S}'_0$ and $\mathcal{S}_1$ are independent of each other.  

\subsection{Estimating $s^*(\cdot)$ in low-dimensional settings}\label{sec::s estimate 1}

To estimate $s^*(x) = (\Sigma^{-1}\mu_d)^{\top}(x)$, we divide the situation into low- and high-dimensional feature dimensionality $d$.  In low-dimensional settings, that is, when feature dimensionality $d$ is small compared to the sample sizes, we use $\mathcal{S}_0$ and $\mathcal{S}_1$ to get the sample means $\hat \mu ^0$ and $\hat \mu^1$ that estimate $\mu^0$ and $\mu^1$ respectively, and get the pooled sample covariance matrix $\widehat \Sigma$ that estimates $\Sigma$. Precisely, we have 
\begin{align}\label{eqn:s hat low d}
\widehat \Sigma &= \frac{1}{n_0+n_1-2}\left(\sum_{i=1}^{n_0}(X^0_i - \hat \mu^0)(X^0_i - \hat \mu^0)^{\top} + \sum_{i=1}^{n_1}(X^1_i - \hat \mu^1)(X^1_i - \hat \mu^1)^{\top}\right)\,,\\
\hat \mu^0 &= \frac{1}{n_0}\left(X^0_1 + \ldots +X^0_{n_0}\right)\,,  \\
\hat \mu^1 &= \frac{1}{n_1}\left(X^1_1 + \ldots +X^1_{n_1}\right)\,.  
\end{align}
%\textcolor{red}{double check if the equation labels are used later}
Let $\hat \mu_d = \hat \mu^1 - \hat \mu^0$, then we can estimate $s^*(\cdot)$ by
$
\hat s(x) = (\widehat \Sigma^{-1} \hat \mu_d)^{\top}x.
$

\subsection{Estimating $s^*(\cdot)$ in high-dimensional settings}\label{sec::s estimate 2}

In high-dimensional settings where $d$ is larger than the sample sizes, $\widehat \Sigma$ in \eqref{eqn:s hat low d} is not invertible, and we need to resort to more sophisticated methods to estimate $s^*(\cdot)$.  
First, we note that although the decision thresholds are different, the NP oracle $\phi^*_{\alpha}$ in \eqref{eqn:np_oracle} and the classical oracle $\phi^*$ in \eqref{eqn:classical_oracle} both project an observation $x$ to the $\beta^{\text{Bayes}}  = \Sigma^{-1}\mu_d$ direction. Hence one can leverage existing works on sparse LDA methods under the classical paradigm to find a $\beta^{\text{Bayes}}$ estimate, using samples $\mathcal{S}_0$ and $\mathcal{S}_1$. In particular, 
we adopt $\hat{\beta}^{\text{lasso}}$, the lassoed (sparse) discriminant analysis (\verb+sLDA+) direction  in \cite{mai2012direct}, which is defined by
\begin{equation}\label{eqn:slda_optimization}
(\hat\beta^{\text{lasso}}, \hat{\beta}^{\lambda}_0) = \argmin_{(\beta, \beta_0)}\left\{n^{-1} \sum_{i=1}^n (y_i - \beta_0 - x_i^{\top} \beta)^2 + \lambda  \sum_{j=1}^d |\beta_j|\right\}\,,
\end{equation}
where $n = n_0 + n_1$, $y_i = -n/n_0$ if the $i$th observation is from class $0$, and  $y_i = n/n_1$ if the $i$th observation is from class $1$. Then we estimate $s^*(\cdot)$ by
$$
\hat s(x) = (\hat\beta^{\text{lasso}})^{\top}x\,.
$$

Note that although the optimization program \eqref{eqn:slda_optimization} is the same as in \cite{mai2012direct},   our sampling scheme is different from that in  \cite{mai2012direct}, where they assumed i.i.d. observations from the joint distribution of $(X, Y)$. As a consequence, when analyzing theoretical properties for $\hat\beta^{\text{lasso}}$ in \eqref{eqn:slda_optimization}, it is necessary to  establish results that are counterparts to those in \cite{mai2012direct}.

%
%To estimate the threshold $C^{**}_{\alpha}$, we use the left-out class $0$ sample $\mathcal{S}_0' = \{x^0_{n_0 + 1}, \cdots, x^0_{n_0 + n_0'}\}$,  leveraging  the next proposition adapted from \cite{tong2016np}.  

\subsection{Estimating  $C_{\alpha}^{**}$ via the nonparametric NP umbrella algorithm} \label{sec::c alpha estimate 1}
\cite{tong2016np} provides a nonparametric order statistics based method, the NP umbrella algorithm, to estimate the threshold $C^{**}_{\alpha}$.  This algorithm leverages the following proposition.  
\begin{proposition}
\label{prop}
	Suppose that we use $\mathcal{S}_0$ and $\mathcal{S}_1$ to train a base algorithm (e.g., \texttt{sLDA}) , and obtain a scoring function $f$ (e.g., an estimate of $s^*$). 	Applying $f$ to $\mathcal{S}_0'$, we denote the resulting classification scores as $T_1, \ldots, T_{n'_0}$, which are real-valued random variables. Then, denote by $T_{(k)}$ the $k$-th order statistic (i.e., $T_{(1)} \leq \ldots \leq T_{(n'_0)}$). For a new observation $X$, if we denote its classification score $f(X)$ as $T$, we can construct classifiers $\hat\phi_k (X) = \1(T > T_{(k)})$, $k\in\{1,\ldots, n'_0\}$. Then, the population type I error of $\hat\phi_k$, denoted by $R_0(\hat\phi_k)$, is a function of $T_{(k)}$ and hence a random variable, and it holds that
	\begin{equation}
	\p\left[ R_0(\hat\phi_k) > \alpha \right] \le \sum_{j=k}^{n'_0} {{n'_0}\choose{j}} (1-\alpha)^j \alpha^{n_0'-j}\,.
	\label{eq3}
	\end{equation}
	That is, the probability that the type I error of $\hat\phi_k$ exceeds $\alpha$ is under a constant that only depends on $k$, $\alpha$ and $n_0'$. We call this probability the violation rate of $\hat\phi_k$ and denote its upper bound by $v(k) = \sum_{j=k}^{n'_0} {{n'_0}\choose{j}} (1-\alpha)^j \alpha^{n_0'-j}$. When $T_i$'s are continuous, this bound is tight.  
\end{proposition}

Proposition \ref{prop} is the key step towards the NP umbrella algorithm proposed in \cite{tong2016np}, which applies to all scoring-type classification methods (base algorithms), including logistic regression, support vector machines, random forest, etc. In the theoretical analysis part of this paper, we always assume the continuity of scoring functions. Under this mild assumption, $v(k)$ is the violation rate of type I error for $\hat \phi_k$. An essential step of the proof of the proposition hinges on the symmetry property of permutation.    It is obvious that $v(k)$ decreases as $k$ increases.    To choose from $\hat \phi_1, \ldots, \hat \phi_{n_0'}$ such that a classifier achieves minimal type II error with type I error violation rate less than or equal to a user's specified  $\delta_0$, the right order is 
\begin{equation}\label{eqn:kstar}
k^* = \min\left\{k\in\{1, \ldots, n'_0\} :   v(k) \leq \delta_0\right\}\,.
\end{equation}

In our current setting, $\hat s(\cdot)$, constructed as in Section \ref{sec::s estimate 1} or Section \ref{sec::s estimate 2},  plays the role of the scoring function. Let $\hat s_{(k^*)}$ be the $k^*$-th order statistic among the set  $\left\{\hat s(X^0_{n_0+1}), \ldots, \hat s(X^0_{n_0+n'_0})\right\}$, then the NP umbrella algorithm sets 

$$
\widehat C_{\alpha}  = \hat s_{(k^*)}\,.
$$

%To construct an NP classifier, one not only needs to specify a type I error upper bound $\alpha$, but also has to specify an upper bound $\delta_0$ on type I error violation rate. %The $\delta_0$ choice is usually positive, as one does not expect a reasonable classifier trained on finite sample  to have type I error be bounded by a small constant almost surely.    
To achieve $\p\left[ R_0(\hat\phi_k) > \alpha \right] \leq \delta_0$ for some $\hat \phi_k$, we need to control the violation rate under $\delta_0$ at least in the extreme case when $k=n'_0$; that is,  it is necessary to ensure $v(n'_0) = (1-\alpha)^{n'_0}  \le \delta_0$. Clearly, if the $(n'_0)$-th order statistic cannot guarantee the violation rate control, other order statistics certainly cannot. Therefore, for a given $\alpha$ and $\delta_0$, there exists a  minimum left-out class $0$ sample size requirement $$n'_0\geq \log \delta_0 / \log (1-\alpha)\,,$$ for the type I error violation rate control. Note that the control on type I error violation rate does not demand any sample size requirements on $\mathcal{S}_0$ and $\mathcal{S}_1$.  But these two parts have an impact on the quality of scoring functions, and hence on the type II error performance.

%\textcolor{red}{the section titles should probably be changed in view that we a change in the abstract and title, in particular the title of section 4}\textcolor{blue}{We can say that the efforts in this work is on type II error...}

\subsection{Estimating $C^{**}_{\alpha}$ by leveraging the parametric assumption}\label{sec::c alpha estimate 2}
We explicitly leverage the LDA model assumption to create an estimate of $C_{\alpha}^{**}$.  For simplicity, let $\hat A$ be either  $\widehat \Sigma^{-1} \hat\mu_d$ in Section \ref{sec::s estimate 1}
   or $\hat\beta^{\text{lasso}}$  in Section \ref{sec::s estimate 2}, corresponding to the low- and high-dimensional settings, respectively. First, we consider $\hat A$ as fixed (i.e., fix $\mathcal{S}_0$ and $\mathcal{S}_1$).  When $X^0\sim \mathcal{N}(\mu^0, \Sigma)$, we have
$$
\hat A^{\top} X^0 \sim \mathcal{N}(\hat A^{\top} \mu^0, \hat A^{\top} \Sigma \hat A)\,.
$$
Let 
$$
\tilde \phi_{\alpha}(x) = \1\left(\hat A^{\top}x > \sqrt{\hat A^{\top} \Sigma \hat A} \Phi^{-1}(1-\alpha) + \hat A^T \mu^0 \right)\,.
$$
 For every fixed $\mathcal{S}_0$ and $\mathcal{S}_1$, $\tilde \phi_{\alpha}$ is the NP oracle  at level $\alpha$.     Note that $\tilde \phi_{\alpha}$ is not an accessible classifier because $\Sigma$ and $\mu^0$ are unknown.  Plugging in the estimates of these parameters is not a good idea because this will not achieve a high probability control of the type I error under $\alpha$.  We plan to construct a statistic $\widehat C^p_{\alpha}$ (the super index $p$ stands for ``parametric thresholding") such that with high probability, 
$$
\widehat C^p_{\alpha} \geq \sqrt{\hat A^{\top} \Sigma \hat A} \Phi^{-1}(1-\alpha) + \hat A^{\top} \mu^0\,.
$$
Let
$$
\hat \phi_{\alpha}(x) = \1\left(\hat A^{\top}x > \widehat C^p_{\alpha} \right)\,.
$$
%Can we claim that $R_0(\hat \phi_{\alpha})\leq \alpha$ with  probability at least $1-\delta$?  
Then $R_0(\hat \phi_{\alpha})\leq R_0(\tilde \phi_{\alpha}) = \alpha$ with high probability.  Naturally, we want such a $\widehat C^p_{\alpha}$ as small as possible, so that the resulting classifier has good type II error performance. Towards this end, we build tight sample-based upper bounds for $\hat A^{\top} \mu^0$ (Lemma \ref{lem:ii}) and $\hat A^{\top} \Sigma \hat A$ (Lemma \ref{lem:iii}), and then combine these bounds to get $\widehat C^p_{
\alpha}$ (Proposition \ref{prop:c_p}).  

%$$
%\p\left(R_0(\hat \phi_{\alpha})|\mathcal{S}^0, \mathcal{S}^1\right)
%$$

\subsubsection{Upper bound for $\hat A^{\top}\mu^0$}

\begin{lemma}\label{lem:ii}
Let $W_1 = \hat A^{\top}X^0_{n_0+1}$, $W_2 = \hat A^{\top}X^0_{n_0+2}, \ldots$, and  $W_{n_0'} = \hat A^{\top}X^0_{n_0+n_0'}$.  Denote by $F_{t_{(n_0'-1)}}(\cdot)$ the cumulative distribution function of the $t$-distribution with degrees of freedom $(n_0'-1)$. Then with probability at least $1-\delta_0$, it holds that  	
\begin{equation}\label{eqn:bound first part}
 \hat A^{\top}\mu^0 \leq \bar W - F_{t_{(n_0'-1)}}^{-1}(\delta_0) \frac{S}{\sqrt{n_0'}}\,,
 \end{equation}
where $\bar W = (W_1 + \ldots + W_{n_0'})/n_0'$ and $S = \left\{\left[(W_1-\bar W)^2+\ldots+(W_{n'_0} - \bar W)^2\right]/(n_0'-1)\right\}^{1/2}$.

\end{lemma}
\begin{proof}
We invoke the following classic result.  Suppose $W_1, \ldots, W_m$ are i.i.d. from one-dimensional $\mathcal{N}(\mu, \sigma^2)$. Let $\bar W = (W_1+ \ldots + W_m)/m$ and $S = \left\{\sum_{j=1}^m(W_j-\bar W)^2/(m-1)\right\}^{1/2}$, then we have
$$
\frac{\bar W - \mu}{S/\sqrt{m}}\sim t_{(m-1)} \text{ }( t \text{ distribution with df}=m-1)\,.
$$

Take $W = \hat A^{\top} X$ and $m = n_0'$, then for fixed $\hat A$,  $W_1 = \hat A^{\top}X^0_{n_0+1}, \ldots, W_{n_0'} = \hat A^{\top}X^0_{n_0+n_0'}$  are  i.i.d. from one-dimensional normally distributed variable with mean $\hat A^{\top}\mu^0$. Then it follows that 
$$
\frac{\bar W - \hat A^{\top}\mu^0}{S/\sqrt{n_0'}}\sim t_{(n_0'-1)}\,.
$$ 
Thus with probability $1-\delta_0$ (randomness comes from $\mathcal{S}'_0$ while keeping $\mathcal{S}_0$ and $\mathcal{S}_1$ fixed),   
\begin{equation}
\frac{\bar W - \hat A^{\top}\mu^0}{S/\sqrt{n_0'}} \geq F_{t_{(n_0'-1)}}^{-1}(\delta_0)\, \text{ } \iff \text{ } \hat A^{\top}\mu^0 \leq \bar W - F_{t_{(n_0'-1)}}^{-1}(\delta_0) \frac{S}{\sqrt{n_0'}}\,.
\end{equation}
Since the above inequality is true for every realization of $\mathcal{S}_0$ and $\mathcal{S}_1$, and $\mathcal{S}'_0$ is independent of $\mathcal{S}_0$ and $\mathcal{S}_1$, the above upper bound for $\hat A^{\top} \mu^0$ holds with probability $1-\delta_0$ regarding all sampling randomness.     
\end{proof}

\subsubsection{Upper bound for $\hat A^{\top} \Sigma \hat A$}

\begin{lemma}\label{lem:iii}
Let $n = n_0 + n_1$. Suppose there exist some positive constants $c$ and $C$ such that $(n-2)^{1/C} < d < (n-2)$ and $d /(n-2) \leq 1 -c$. 	Then for any positive $\epsilon$ and $D$, there exists an $N(\epsilon, D)$ such that for all $n\geq N(\epsilon, D)$, it holds with probability at least $1 - (n-2)^{-D}$ that, 
\begin{equation}\label{eqn:max_eigen_2}
\hat A^{\top}\Sigma \hat A \leq \frac{1}{\left(1 - \sqrt{\frac{d}{n-2}}\right)^2- \frac{(n-2)^{\epsilon}}{\sqrt{n-2}d^{\frac{1}{6}}}}\lambda_{\max}(\widehat\Sigma)\hat A^{\top} \hat A\,.
\end{equation}
\end{lemma}

\begin{proof}
An obvious upper bound for $\hat A^{\top} \Sigma \hat A$ is $\lambda_{\max}(\Sigma) \hat A^{\top}\hat A$, but $\lambda_{\max}(\Sigma)$ is not accessible from samples. Instead, $\lambda_{\max}(\widehat\Sigma)$ is accessible. In the following, we explore the relations between $\lambda_{\max}(\Sigma)$ and $\lambda_{\max}(\widehat\Sigma)$ to derive a sample-based upper bound for $\hat A^{\top}\Sigma \hat A$.  
	
	Let $Z_1, \ldots, Z_{n_0}$ be i.i.d. from $d$-dimensional Gaussian $\mathcal{N}(\mu^0, I_{d\times d})$,  $Z_{n_0+1}, \ldots, Z_{n_0+n_1}$ be i.i.d. from $\mathcal{N}(\mu^1, I_{d\times d})$, and $n = n_0 + n_1$. Further assume that all the $Z_i$'s, $i=1, \ldots, n$, are independent of each other. %\textcolor{red}{need to update the next index....}
Let 
$$
\widehat U = \frac{1}{n-2}\left(\sum_{i=1}^{n_0}(Z_i - \bar Z_0)(Z_i - \bar Z_0)^{\top} + \sum_{i=n_0+1}^{n}(Z_i - \bar Z_1)(Z_i - \bar Z_1)^{\top}\right)\,,
$$
where $\bar Z_0 = (Z_1 + \ldots + Z_{n_0})/n_0$ and $\bar Z_1 = (Z_{(n_0+1)} + \ldots + Z_{n})/n_1$.  Define $G_i = Z_i - \mu^0$ for $i = 1, \ldots, n_0$ and $G_i = Z_i - \mu^1$ for $i = n_0+1, \ldots, n$. Then $G_i\sim \mathcal{N}(0, I_{d\times d})$ for all $i = 1, \ldots, n$.   Let
$$
\widehat V = \frac{1}{n-2} \left(\sum_{i=1}^{n_0} (G_i-\bar G_0)(G_i - \bar G_0)^{\top}+ \sum_{i=n_0+1}^{n} (G_{i}-\bar G_1)(G_{i} - \bar G_1)^{\top} \right)\,,
$$  
where $\bar G_0 = (G_1+ \ldots + G_{n_0})/{n_0}$ and $\bar G_1 = (G_{n_0+1}+ \ldots + G_{n})/{n_1}$.  Clearly $\widehat U = \widehat V$. 
Let 
$$\textbf{P} = \text{diag}\left(I_{n_0 \times n_0} - \frac{1}{n_0}\textbf{e}_0 \textbf{e}_0^{\top}, I_{n_1 \times n_1} - \frac{1}{n_1}\textbf{e}_1 \textbf{e}_1^{\top}\right)\,,$$ where the entries of $\textbf{e}_0$ and $\textbf{e}_1$ are all equal to $1$, and the lengths are $n_0$ and $n_1$ respectively. Then $\textbf{P}$ is a projection matrix with rank $n-2$.  Note that a projection matrix has eigenvalues all equal to $1$ and $0$, with the number of $1$'s equals to its rank.  Hence, we can decompose $\textbf{P}$ by $\textbf{P} = \textbf{U}\textbf{U}^{\top}$, where $\textbf{U}$ is an $n \times (n-2)$ orthogonal matrix with $\textbf{U}^{\top}\textbf{U} = I_{(n-2) \times (n-2)}$.  

Let $\textbf{G} = (G_1, \ldots, G_{n})$. $\textbf{G}$ is a $d \times n$ matrix and its columns are i.i.d. $d$-dimensional standard multivariate Gaussian. Let $\tilde{\textbf{G}} = \textbf{G}\textbf{U}$, then $\tilde{\textbf{G}}$ is a $d \times (n-2)$ matrix in which the columns are i.i.d. $d$-dimensional standard multivariate Gaussian.  Therefore, we have
$$
\widehat U = \widehat V = \frac{1}{n-2} \textbf{G}\textbf{P}\textbf{G}^{\top}=\frac{1}{n-2}\textbf{G}\textbf{U}\textbf{U}^{\top}\textbf{G}^{\top} = \frac{1}{n-2}\tilde{\textbf{G}}\tilde{\textbf{G}}^{\top}\,.
$$ 
By \cite{Bloemendal.et.al.2015}, we have the following concentration result on minimum eigenvalues: suppose there exist some positive constants $c$ and $C$ such that $(n-2)^{1/C} < d < (n-2)$ and $d/(n-2) \leq 1 -c$.  Then for any positive $\epsilon$ and $D$, there exists an $N(\epsilon, D)$ such that for all $n\geq N(\epsilon, D)$, we have
\begin{equation}\label{eqn:min_eigen}
\p\left(\left|\lambda_{\text{min}}\left(\frac{1}{n-2}\tilde{\textbf{G}}\tilde{\textbf{G}}^{\top}\right) - \left(1 - \sqrt{\frac{d}{n-2}}\right)^2\right| > \frac{(n-2)^{\epsilon}}{\sqrt{n-2}d^{\frac{1}{6}}}\right) \leq (n-2)^{-D}\,.
\end{equation}

This result implies that for $n\geq N(\epsilon, D)$, we have with probability at least $1 - (n-2)^{-D}$, 
$$
I_{d\times d}\leq \frac{1}{\left(1 - \sqrt{\frac{d}{n-2}}\right)^2- \frac{(n-2)^{\epsilon}}{\sqrt{n-2}d^{\frac{1}{6}}} }\left(\frac{1}{n-2}\tilde{\textbf{G}}\tilde{\textbf{G}}^{\top}\right)\,,
$$
where the inequality means ``$A\leq B$ iff $B-A$ is positive semi-definite". This further implies

$$
\widehat \Sigma \overset{d}= \Sigma^{1/2} \widehat U \Sigma^{1/2} =   \Sigma^{1/2} \left(\frac{1}{n-2}\tilde{\textbf{G}}\tilde{\textbf{G}}^{\top}\right)\Sigma^{1/2} \geq \left(\left(1 - \sqrt{\frac{d}{n-2}}\right)^2- \frac{(n-2)^{\epsilon}}{\sqrt{n-2}d^{\frac{1}{6}}}\right)  \Sigma\,,
$$ 
where the notation ``$\overset{d}=$" means equal in distribution.  
Therefore, for $n\geq N(\epsilon, D)$, we have with probability at least $1 - (n-2)^{-D}$, 
\begin{equation}\label{eqn:max_eigen}
\lambda_{\text{max}}(\Sigma)\leq \frac{1}{\left(1 - \sqrt{\frac{d}{n-2}}\right)^2- \frac{(n-2)^{\epsilon}}{\sqrt{n-2}d^{\frac{1}{6}}}}\lambda_{\text{max}}(\widehat\Sigma)\,.
\end{equation}
Hence we can bound $\hat A^{\top}\Sigma \hat A$ by  
\begin{equation}\label{eqn: chain}
\hat A^{\top}\Sigma \hat A \leq \lambda_{\text{max}}(\Sigma) \hat A^{\top} \hat A \leq \frac{1}{\left(1 - \sqrt{\frac{d}{n-2}}\right)^2- \frac{(n-2)^{\epsilon}}{\sqrt{n-2}d^{\frac{1}{6}}}}\lambda_{\text{max}}(\widehat\Sigma)\hat A^{\top} \hat A\,.
\end{equation}

\end{proof}

Lemma \ref{lem:iii}  can be improved for large $d$ and sparse $\hat A$.  To bound $\hat A^{\top}\Sigma \hat A = (\hat\beta^{\text{lasso}})^{\top}\Sigma \hat\beta^{\text{lasso}}$ using results in \cite{Bloemendal.et.al.2015}, we need the condition $d < n-2$. Even when $d$ is moderate, the bound on the right-hand side of \eqref{eqn:max_eigen_2} can be loose. A remedy exists when $\hat\beta^{\text{lasso}}$ is sparse, e.g., $\|\hat\beta^{\text{lasso}}\|_0 = s \ll d$. Concretely, we can replace $\hat\beta^{\text{lasso}}$ in the first inequality of \eqref{eqn: chain}  by its $s$-dimensional sub-vector that consists of the nonzero coordinates, and $\Sigma$ by its $s\times s$ sub-matrix that corresponds to the nonzero elements of $\hat\beta^{\text{lasso}}$.  Then in the multiplicative factor on the right-hand side of \eqref{eqn: chain}, we replace $d$ by a much smaller $s$ and replace the maximum eigenvalue of the $d \times d$ pooled sample covariance matrix by that of the $s\times s$ pooled sample covariance matrix. To summarize,  we achieve a much tighter high probability bound of $\hat A^{\top}\Sigma \hat A$ by exploring the sparsity of $\hat A$. This variant also allows us to handle the situation when $d$ is bigger than the sample sizes. In numerical implementation, when we use $\hat\beta^{\text{lasso}}$ for $\hat A$, this improved bound is what we always use, although for notational simplicity we still denote the threshold of scoring function as $\widehat C^p_{\alpha}$.

\subsubsection{Combining bounds for $\hat A^{\top}\mu^0$ and $\hat A^{\top}\Sigma \hat A$}
The arguments at the beginning of Section \ref{sec::c alpha estimate 2} together with the derived upper bounds for $\hat A^{\top}\mu^0$ and $\hat A^{\top}\Sigma \hat A$ imply the following proposition.  
\begin{proposition}\label{prop:c_p}
Under definitions and conditions in Lemmas \ref{lem:ii} and \ref{lem:iii},  if we set 
\begin{equation}\label{eqn:parametric threshold}
\widehat C^p_{\alpha} = \left(\frac{1}{\left(1 - \sqrt{\frac{d}{n_0+n_1-2}}\right)^2- \frac{(n_0+n_1-2)^{\epsilon}}{\sqrt{n_0+n_1-2}d^{\frac{1}{6}}}}\lambda_{\max}(\widehat\Sigma)\hat A^{\top} \hat A\right)^{\frac{1}{2}}\cdot \Phi^{-1}(1-\alpha) + \bar W - F_{t_{(n_0'-1)}}^{-1}(\delta_0) \frac{S}{\sqrt{n_0'}}\,,
\end{equation}
then $\hat \phi_{\alpha}(x) = \1( \hat A^{\top}x > \widehat C^p_{\alpha})$ satisfies
\begin{equation}\label{eqn:high prob}
\p(R_0(\hat \phi_{\alpha}) \leq \alpha ) > 1 - \delta_0 - (n-2)^{-D}\,.
\end{equation}
\end{proposition}

 In many applications,  $n_1$ is large, so the sample size requirement of \eqref{eqn:min_eigen} (and hence of \eqref{eqn:max_eigen} and \eqref{eqn:high prob}) on $n$, although impossible to check, can be comfortably assumed. Furthermore, simulation studies in Appendix  \ref{sec::eigen:bound:simu} show that the inequality \eqref{eqn:max_eigen} is often satisfied with probability very close to $1$ even when $n$ is moderate, with the choice of $\epsilon = 1e-3$.  Hence the right-hand side of \eqref{eqn:high prob} is often almost $1-\delta_0$ in practice.

\subsection{The resulting four LDA-based classifiers $\hat \phi_{\alpha}$}\label{sec::four classifiers}
  
Having obtained two estimates for $s^*$ and two  for $C^{**}_{\alpha}$, we construct four classifiers in total: \verb+NP-LDA+: $\1\left((\widehat\Sigma^{-1} \hat \mu_d)^{\top}x > \widehat C_{\alpha}\right)$, \verb+NP-sLDA+: $\1\left((\hat \beta^{\text{lasso}})^{\top}x > \widehat C_{\alpha}\right)$, \verb+pNP-LDA+: $\1\left((\widehat\Sigma^{-1} \hat \mu_d)^{\top}x > \widehat C^p_{\alpha}\right)$, \verb+pNP-sLDA+: $\1\left((\hat \beta^{\text{lasso}})^{\top}x > \widehat C^p_{\alpha}\right)$.  We will suggest the application domains of these four classifiers in the numerical analysis section.

\section{Theoretical analysis}
In the theoretical analysis, we focus on the \verb+NP-sLDA+ classifier, which we denote by $\hat \phi_{k^*}(x) = \1\left((\hat\beta^{\text{lasso}})^{\top}x > \widehat C_{\alpha}\right)$ in this section. We will establish NP oracle inequalities for $\hat \phi_{k^*}$.  The \textit{NP oracle inequalities} were formulated for classifiers under the NP paradigm to reckon the spirit of oracle inequalities in the classical paradigm. They require two properties to hold simultaneously with high probability: i). type I error $R_0(\hat \phi_{k^*})$ is  bounded from above by $\alpha$,  and ii). excess type II error, that is $R_1(\hat \phi_{k^*}) - R_1(\phi^*_{\alpha})$, diminishes as sample sizes increase.  \textit{By construction of the order $k^*$ in} \verb+NP-sLDA+ $\hat \phi_{k^*}$, \textit{the first property is already fulfilled,  so in the following we bound the excess type II error}.  

In the NP classification literature, nonparametric plug-in NP classifiers  constructed in \cite{Tong.2013} and \cite{Zhao2016neyman} were shown to satisfy the NP oracle inequalities.   Both  papers assume bounded feature support $[-1, 1]^d$. Under this assumption, uniform deviation bounds between $f_1 / f_0$ and its nonparametric estimate $\hat f_1 / \hat f_0$  were derived, and such uniform deviation bounds were crucial in bounding the excess type II error. However, as canonical parametric models in classification (such as LDA and QDA)  have unbounded feature support, the development of NP theory under parametric settings cannot bypass the challenges arisen from the unboundedness of feature support.  To address these challenges, we follow a conditioning-on-a-high-probability-set strategy  and formulate conditional marginal assumption and conditional detection condition. For the LDA model, we elaborate on these high-level conditions in terms of specific parameters. In fact, the same conditioning-on-a-high-probability-set strategy can work under the nonparametric model assumptions, such as with a mild finite moment condition. Thus,  we can also relax the bounded feature support assumption for the nonparametric methods. Before presenting the new assumptions and main theorem, we need a few technical lemmas to make the ``conditioning" work.  

%However, one cannot expect similar  results to hold for the feature support $\mathbb{R}^d$ of the Gaussian distributions, driving necessity for innovation in establishing NP oracle inequalities for NP-sLDA.    

\subsection{A few technical lemmas}

%The type I error violation rate $\delta_0$ is guaranteed by the construction of the order $k^*$. %To have good type II error performance, we also need $R_0(\hat \phi_{k^*})$ be close to $R_0(\phi^*)$.  

With kernel density estimates $\hat f_1$, $\hat f_0$, and an estimate of the threshold level $\widetilde C_{\alpha}$ based on VC inequality, \cite{Tong.2013} constructed a plug-in classifier $\1\{\hat f_1(x)/\hat f_0(x)\geq \widetilde C_{\alpha}\}$ that is of limited practical value unless the feature dimension is small and sample size is large. \cite{Zhao2016neyman} analyzed high-dimensional Naive Bayes models under the NP paradigm and innovated the threshold estimate by invoking order statistics with an explicit analytic formula for the chosen order. We denote that order by $k'$. The order $k^*$ derived in \cite{tong2016np} is a refinement of the order statistics approach to estimate the threshold.  However, although the order $k^*$ is optimal, it does not take an explicit formula and thus is not helpful in bounding the excess type II error.  Interestingly, efforts to approximate $k^*$ analytically for type II error control  leads to $k'$, and so $k'$ will be employed as a bridge in establishing NP oracle inequalities for $\hat \phi_{k^*}$.

To derive an upper bound for the excess type II error, it is essential to bound the deviation between  type I error of $\hat \phi_{k^*}$ and that of the NP oracle $\phi_{
\alpha}^*$.
To achieve this, we first quote the next proposition from \cite{Zhao2016neyman} and then derive a corollary.

%\textcolor{red}{We need not only type I error to be bounded, we actually need it to be close to the oracle type I error}

\begin{proposition}
\label{prop::kmin}
Given $\delta_0\in(0, 1)$, suppose $n'_0 \geq 4/(\alpha \delta_0)$, let the order $k'$ be defined as follows
\begin{equation}
\label{eq::kmin}
k' \,=\, \left\lceil (n'_0+1)A_{\alpha,\delta_0}(n'_0) \right\rceil,
\end{equation}
where $\lceil z\rceil$ denotes the smallest integer larger than or equal to $z$, and  
\begin{equation*}
\label{eq::A}
A_{\alpha,\delta_0}(n'_0) = \frac{1+2\delta_0(n'_0+2)(1-\alpha) + \sqrt{1+4\delta_0(1-\alpha)\alpha(n'_0+2)}}
{2\left\{\delta_0(n'_0+2)+1 \right\}}\,.
\end{equation*}
Then we have 
$$
\p\left(R_0(\hat \phi_{k'}) > \alpha \right)\leq \delta_0\,.
$$
In other words, the type I error of classifier $\hat \phi_{k'}$ ($\hat \phi_k$ was defined in Proposition \ref{prop}) is bounded from above by $\alpha$ with probability at least $1-\delta_0$\,.
\end{proposition}

\begin{corollary}\label{cor:kstarAndkprime}
Under continuity assumption of the classification scores $T_i$'s (which we always assume in this paper), the order $k^*$ is smaller than or equal to the order $k'$.  
\end{corollary}
\begin{proof}
Under the continuity assumption of $T_i$'s, $v(k)$ is the exact violation rate of classifier $\hat \phi_k$. By construction, both $v(k')$ and $v(k^*)$ are smaller than or equal to $\delta_0$. 
Since $k^*$ is the smallest $k$ that satisfies $v(k)
\leq \delta_0$, we have $k^*\leq k'$.  
\end{proof}
%In the following, we always assume that the mild continuity assumption of the classification scores is satisfied. 

\begin{lemma}
\label{prop::R0}

Let  $\alpha, \delta_0\in(0,1)$ and $n_0'  \geq 4/(\alpha\delta_0)$. For any $\delta_0' \in (0,1)$,
the distance between $R_0( \hat{\phi}_{k'})$  and $R_0(\phi^*_{\alpha})$ can be bounded as 
\begin{eqnarray*}
\p\{
|R_0( \hat{\phi}_{k'} ) 
- R_0( \phi^*_{\alpha}) | > \xi_{\alpha, \delta_0,n_0'}(\delta_0') \}
\,\leq\, \delta_0'\,,
\end{eqnarray*}
where 
\begin{eqnarray*}
\label{eq::xi}
\xi_{\alpha, \delta_0,n_0'}(\delta_0') = \sqrt{\frac{k'(n_0'+1-k')}{(n_0'+2)(n_0'+1)^2\delta_0'}} + A_{\alpha,\delta_0}(n_0') - (1-\alpha) + \frac{1}{n_0'+1}\,,
\end{eqnarray*}
in which $k'$ and  $A_{\alpha,\delta_0}(n_0')$ are the same as in Proposition \ref{prop::kmin}. 
Moreover, if $n_0' \geq \max(\delta_0^{-2}, \delta_0^{'-2})$, we have 
$
\xi_{\alpha, \delta_0,n_0'}(\delta_0') \leq  ({5}/{2}){n_0'^{-1/4}}.
$
\end{lemma}
%\textcolor{red}{We might not need the last result; it is for the detection condition within a certain range. On the other hand, this further restriction on $n'_0$ gives a clean rate result.  }

Lemma \ref{prop::R0} is borrowed from \cite{Zhao2016neyman}, so its proof is omitted.  Based on Lemma \ref{prop::R0} and Corollary  \ref{cor:kstarAndkprime}, we can derive the following result whose proof is in the Appendix.

\begin{lemma}\label{lem:2}
Under the same assumptions as in Lemma \ref{prop::R0}, the distance between $R_0( \hat{\phi}_{k^*})$  and $R_0(\phi^*_{\alpha})$ can be bounded as 
\begin{eqnarray*}
\p\{
|R_0( \hat{\phi}_{k^*} ) 
- R_0( \phi^*_{\alpha}) | > \xi_{\alpha, \delta_0,n_0'}(\delta_0') \}
\,\leq\, \delta_0 +  \delta_0'\,.
\end{eqnarray*}
\end{lemma}
Lemma \ref{lem:2} serves as an intermediate  step towards the final ``conditional" version to be elaborated in Lemma  \ref{lem:4}. Moving towards Lemma  \ref{lem:4}, we construct a set $\mathcal{C}\in \mathbb{R}^d$, such that $\mathcal{C}^c$ is ``small". We also show that the uniform deviation between $\hat s$ and $s^*$ on $\mathcal{C}$ is controllable (Lemma \ref{lem:large_prob_set}).  To achieve that, we digress to introduce some more notations.  
Suppose the lassoed linear discriminant analysis (\verb+sLDA+) finds the set $A$,  which is the support of the Bayes rule direction $\beta^{\text{Bayes}}$, we have $\hat{\beta}^{\text{lasso}}_{A^c}=0$ and $\hat{\beta}^{\text{lasso}}_{A} = \hat{\beta}_A$, where $\hat{\beta}_A$ is defined by
$$
(\hat \beta_{A}, \tilde\beta_0) = \argmin_{(\beta, \beta_0)}\left\{n^{-1} \sum_{i=1}^n (y_i - \beta_0 - \sum_{j\in A}x_{ij}\beta_j)^2 + \sum_{j\in A} \lambda  |\beta_j|\right\}\,.
$$
The quantity $\hat\beta_A$  is only for theoretical analysis, as the definition assumes knowledge of the true support set $A$.  The next proposition is a counterpart of Theorem 1 in \cite{mai2012direct}, but due to different sampling schemes, it differs from that theorem and a proof is attached in the Appendix.  

%\textcolor{red}{need to re-derive  the results because of different sampling assumptions.  }

\begin{proposition}\label{thm:1}

Assume $\kappa := \| \Sigma_{A^c A} (\Sigma_{AA})^{-1}\|_{\infty} < 1$ and choose $\lambda$ in the optimization program   \eqref{eqn:slda_optimization} such that $\lambda < \min\{|\beta^*|_{\min}/(2\varphi), \Delta\}$, where $\beta^* = (\Sigma_{AA})^{-1}(\mu^1_{A}-\mu^0_{A})$, $\varphi = \|(\Sigma_{AA})^{-1}\|_{\infty}$ and $\Delta = \|\mu^1_A - \mu^0_A\|_{\infty}$,    then it holds that 

\begin{itemize}
\item[1.] With probability at least $1-\delta^*_1$, $\hat {\beta}_A^{\text{lasso}} = \hat{\beta}_A$ and $\hat{\beta}_{A^c}^{\text{lasso}}=0$, where
$$
\delta^*_1 = \sum_{l=0}^1 2d \exp\left(-c_2 n_l \frac{\lambda^2 (1 - \kappa - 2 \varepsilon \varphi)^2}{16(1 + \kappa)^2}\right) + f(d, s, n_0, n_1, (\kappa+1)\varepsilon \varphi (1 - \varphi \varepsilon)^{-1})\,,
$$
in which  $\varepsilon$ is any positive constant less than $\min [ \varepsilon_0, \lambda (1-\kappa) (4\varphi)^{-1}(\lambda/2 + (1+\kappa)\Delta)^{-1}]$ and $\varepsilon_0$ is some positive constant, and in which
$$
f(d, s, n_0, n_1, \varepsilon) = (d+s) s \exp\left( -\frac{c_1  \varepsilon^2n^2}{4s^2n_0} \right)  + (d+s) s \exp\left( -\frac{c_1 \varepsilon^2n^2 }{4s^2n_1} \right)\,,
$$ 
for some constants $c_1$ and $c_2$.  
\item[2.] With probability at least $1-\delta^*_2$, none of the elements of $\hat \beta_{A}$ is zero, where
$$
\delta^*_2 = \sum_{l=0}^1 2s\exp(-n_l \varepsilon^2 c_2)  + \sum_{l=0}^1 2s^2 \exp\left(- \frac{c_1 \varepsilon^2 n^2}{4 n_l s^2}  \right)\,.
$$
in which $\varepsilon$ is any positive constant less than $\min[\varepsilon_0, \xi(3+\xi)^{-1}/\varphi, \Delta\xi (6 + 2\xi)^{-1}]$, where $\xi = |\beta^*|_{\min}/ (\Delta\varphi)$\,.

\item[3.] For any positive $\varepsilon$ satisfying $\varepsilon < \min\{\varepsilon_0, \lambda (2\varphi \Delta)^{-1}, \lambda\}$, we have
$$
\p\left(\|\hat{\beta}_{A} - \beta^*\|_{\infty}\leq 4\varphi \lambda \right)\geq 1 - \delta^*_2\,.
$$

\end{itemize}

\end{proposition}

Aided by Proposition \ref{thm:1},  the next lemma constructs $\mathcal{C}$,  a high probability set under both $P_0$ and $P_1$. Moreover, a high probability bound is derived for the uniform deviation between $\hat s$ and $s^*$ on this set.  

%\textcolor{red}{The next Lemma seems a little out of place here.}

\begin{lemma}\label{lem:large_prob_set}
Suppose $\max\{tr(\Sigma_{AA}), tr(\Sigma_{AA}^2), \|\Sigma_{AA}\|, \|\mu^0_A\|^2, \|\mu^1_A\|^2\}\leq c_0 s$ for some constant $c_0$, where $s = \text{cardinality}(A)$.  
For $\delta_3 = \exp\{- (n_0 \wedge n_1)^{1/2}\}$, there exists some constant  $c_1' > 0$, such that $\mathcal{C} = \{X\in \mathbb{R}^d: \|X_A\| \leq c_1' s^{1/2}(n_0 \wedge n_1)^{1/4} \}$ satisfies $P_0(X\in\mathcal{C})\geq 1-\delta_3$ and $P_1(X\in\mathcal{C})\geq 1-\delta_3$. Moreover, let $\|\hat s - s^* \|_{\infty, \mathcal{C}} := \max_{x\in \mathcal{C}}|\hat s(x) - s^*(x)|$. Then for    $\delta_1 \geq \delta_1^*$ and $\delta_2 \geq \delta_2^*$, where $\delta^*_1$ and $\delta^*_2$ are defined as in Proposition \ref{thm:1},  it holds that with probability at least  $1- \delta_1 - \delta_2$, $$\|\hat s - s^* \|_{\infty, \mathcal{C}}\leq 4c_1' \varphi \lambda s (n_0\wedge n_1)^{1/4}\,.$$
\end{lemma}

The set $\mathcal{C}$ was constructed with two opposing missions in mind. On one hand, we want to restrict the feature space $\mathbb{R}^d$ to  $\mathcal{C}$ so that the restricted uniform deviation of $\hat s$ from $s^*$ is controlled.  On the other hand, we also want  $\mathcal{C}$ to be sufficiently large, so that $P_0(X\in\mathcal{C}^c)$ and $P_1(X\in\mathcal{C}^c)$ diminish as sample sizes increase.  
%\adc{Should we change to $P_0(X\in \mathcal{C}^c)$? }
The next lemma is implied by Lemma \ref{lem:2} and Lemma \ref{lem:large_prob_set}. %, and it will be used in deriving a bound for excess type II error.  

\begin{lemma}\label{lem:4}
Let $\mathcal{C}$ be defined as in Lemma \ref{lem:large_prob_set}.  Then under the same conditions as in Lemma \ref{prop::R0}, the distance between $R_0( \hat{\phi}_{k^*}|\mathcal{C}) :=  P_0(\hat s(X) \geq \widehat C_{\alpha}| X\in\mathcal{C})$  and $R_0(\phi^*_{\alpha}|\mathcal{C}):=  P_0(s^*(X) \geq C^{**}_{\alpha}| X\in\mathcal{C})$ can be bounded as 
\begin{eqnarray*}
\p\{
|R_0( \hat{\phi}_{k^*} |\mathcal{C}) 
- R_0( \phi^*_{\alpha}|
\mathcal{C}) | > 2[\xi_{\alpha, \delta_0,n_0'}(\delta_0') + \exp\{- (n_0 \wedge n_1)^{1/2}\}] \}
\,\leq\, \delta_0 +  \delta_0'\,,
\end{eqnarray*}
where $\xi_{\alpha, \delta_0,n_0'}(\delta_0') $ is defined in Lemma \ref{prop::R0}.  
%\end{lemma}
\end{lemma}

\subsection{Margin assumption and detection condition}

Margin assumption and detection condition are critical theoretical assumptions in \cite{Tong.2013} and \cite{Zhao2016neyman} for bounding excess type II error of the nonparametric NP classifiers constructed in those papers.  %Unlike \cite{Tong.2013} and \cite{Zhao2016neyman} which assumed bounded feature support, parametric models (e.g., the LDA model) often have the entire $\mathbb{R}^d$ as the support. 
To assist our  proof strategy that divides the $\mathbb{R}^d$ space into a high probability set  (e.g., $\mathcal{C}$ defined in Lemma \ref{lem:large_prob_set}) and its complement, we introduce conditional versions of these assumptions.  
%\textcolor{red}{should add one more line about why we need to reformulate it}

%The current setting is  different from previous works because the support for $f_0$ and $f_1$ is the whole $\mathbb{R}^d$ space.  

\begin{definition}[conditional margin assumption]\label{Def: margin_Anqi}
A function $f(\cdot)$ is said to satisfy conditional margin assumption restricted to $\mathcal{C}^*$ of order $\bar\gamma$ with respect to probability distribution $P$ (i.e., $X\sim P$) at the level $C^{*}$ if there exists a positive constant $M_0$, such that for any $\delta\geq0$,
$$
P\{|f(X)-C^*|\leq \delta | X\in \mathcal{C}^*\} \,\leq\, M_0\delta^{\bar\gamma}\,.
$$
\end{definition}

The unconditional version of such an assumption was first introduced in \cite{Polonik95}. In the classical binary classification framework, \citet{MamTsy99}  proposed a similar condition named ``margin condition" by requiring most data to be away from the optimal decision boundary and this condition has become a common assumption in classification literature.   In the classical classification paradigm, Definition \ref{Def: margin_Anqi} reduces to the margin condition by taking $f=\eta$, $\mathcal{C}^* = \text{support}(X)$ and $C^* =1/2$, with $\{x:|f(x) - C^*|=0\}=\{x:\eta(x)=1/2\}$ giving the decision boundary of the classical Bayes classifier. For a $\mathcal{C}^*$ with nontrivial probability, the conditional margin assumption is weaker than the unconditional version.  For example, suppose $P(X\in \mathcal{C}^*)\geq 1/2$, then the condition $P\{|f(X)-C^*|\leq \delta \} \leq2M_0\delta^{\bar\gamma}$ would imply the conditional margin assumption in view of the Bayes Theorem.

Definition \ref{Def: margin_Anqi} is a high level assumption. In view of explicit Gaussian modeling assumptions,  it is preferable to derive it based on more elementary assumptions on $\mu^0$, $\mu^1$ and $\Sigma$, for our choices of $f$,  $P$, $C^*$ and $\mathcal{C}^*$. Recall that the NP oracle classifier can be written as 
$$
\phi^*_{\alpha}(x)=\1\{(\Sigma^{-1}\mu_d)^{\top}x > C_{\alpha}^{**}\}\,.
$$
Here we take $f(x) = s^*(x) = (\Sigma^{-1}\mu_d)^{\top}x$, $C^* = C_{\alpha}^{**}$, $P= P_0$, and $\mathcal{C}^* = \mathcal{C}$ in Lemma \ref{lem:large_prob_set}. When $X\sim \mathcal{N}(\mu^0, \Sigma)$, $(\Sigma^{-1}\mu_d)^{\top} X\sim \mathcal{N}(\mu_d^{\top}\Sigma^{-1}\mu^0, \mu_d^{\top}\Sigma^{-1}\mu_d)$.  %and it holds that
%\begin{eqnarray*}
%&&P_0\left( |X^{\top}\Sigma^{-1}\mu_d - C_{\alpha}^{**}|\leq \delta \right)\\
% &=& P_0\left( C_{\alpha}^{**} - \delta \leq (\Sigma^{-1}\mu_d)^{\top} X \leq C_{\alpha}^{**} + \delta\right)\\
% &=& \Phi\left(U\right) - \Phi\left(L\right)\,,
%\end{eqnarray*}
%where $\Phi$ is the cdf of the standard normal distribution, $U = (C_{\alpha}^{**} + \delta - \mu_d^{\top}\Sigma^{-1}\mu^0)/ \sqrt{\mu_d\Sigma^{-1}\mu_d}$ and $L = (C_{\alpha}^{**} - \delta - \mu_d^{\top}\Sigma^{-1}\mu^0)/ \sqrt{\mu_d\Sigma^{-1}\mu_d}.$ By the mean value theorem, we have
%$$
%\Phi\left(U\right) - \Phi\left(L\right) = \phi (z) (U - L) = \phi(z) \frac{2\delta}{\sqrt{\mu_d\Sigma^{-1}\mu_d}}\,,
%$$
%where $z$ is some point in $[L, U]$. 
%
%
%Clearly $\phi$ is bounded from above by $\phi(0)$. If we assume that \textcolor{blue}{$\mu_d\Sigma^{-1}\mu_d\geq C$} for some universal positive constant $C$, the margin assumption is met with the constant $M_0 = 2\phi(0)/\sqrt{C}$ and $\bar \gamma = 1$.  
%\textcolor{blue}{what follows might be we actually need as the margin assumption; conditional margin assumption}. 
Lemma \ref{lem:large_prob_set} guarantees that for $\delta_3 = \exp\{- (n_0 \wedge n_1)^{1/2}\}$,  $P_0(X\in\mathcal{C})\geq 1- \delta_3$. Moreover,   %Without loss of generality, we can assume that $\{x:C^{**}_{\alpha} \leq (\Sigma^{-1}\mu_d)^{\top}x\leq C^{**}_{\alpha} + \delta\}\subset \mathcal{C}$. 
\begin{eqnarray*}
&&P_0\left( |s^*(X) - C_{\alpha}^{**}|\leq \delta |X\in\mathcal{C}\right)\\
% &=& P_0\left( C_{\alpha}^{**} - \delta \leq (\Sigma^{-1}\mu_d)^{\top} X \leq C_{\alpha}^{**} + \delta | X\in \mathcal{C}\right)\\
&\leq & P_0\left( C_{\alpha}^{**} - \delta \leq (\Sigma^{-1}\mu_d)^{\top} X \leq C_{\alpha}^{**} + \delta\right)/(1 - \delta_3)\\
 &= & [\Phi\left(U\right) - \Phi\left(L\right)] /(1 - \delta_3)\,,
\end{eqnarray*}
where $\Phi$ is the cumulative   distribution function of the standard normal distribution, $U = (C_{\alpha}^{**} + \delta - \mu_d^{\top}\Sigma^{-1}\mu^0)/ \sqrt{\mu_d^{\top}\Sigma^{-1}\mu_d}$,  and $L = (C_{\alpha}^{**} - \delta - \mu_d^{\top}\Sigma^{-1}\mu^0)/ \sqrt{\mu_d^{\top}\Sigma^{-1}\mu_d}.$ By the mean value theorem, we have
$$
\Phi\left(U\right) - \Phi\left(L\right) = \phi (z) (U - L) = \phi(z) \frac{2\delta}{\sqrt{\mu_d^{\top}\Sigma^{-1}\mu_d}}\,,
$$
where $\phi$ is the probability distribution function of the standard normal distribution, and $z$ is in $[L, U]$. Clearly $\phi$ is bounded from above by $\phi(0)$. Hence, under the assumptions of Lemma \ref{lem:large_prob_set}, if we additionally assume that $\mu_d^{\top}\Sigma^{-1}\mu_d\geq C$ for some universal positive constant $C$, the conditional margin  assumption is met with the restricted set $\mathcal{C}$, the constant $M_0 = 2\phi(0)/(\sqrt{C}(1-\delta_3))$ and $\bar \gamma = 1$. Since $\delta_3 < 1/2$, we can take $M_0 = 4 \phi(0)/\sqrt{C}$.  

\begin{assumption}\label{assumption:1}
i). $\max\{tr(\Sigma_{AA}), tr(\Sigma_{AA}^2), \|\Sigma_{AA}\|, \|\mu^0_A\|^2, \|\mu^1_A\|^2\}\leq c_0 s$ for some constant $c_0$, where $s = \text{cardinality}(A)$, and $A = \{j: \{ \Sigma^{-1}\mu_d \}_j\neq 0 \}$;   ii). $\mu_d^{\top}\Sigma^{-1}\mu_d\geq C$ for some universal positive constant $C$;  iii). the set $\mathcal{C}$ is defined as in Lemma \ref{lem:large_prob_set} . 
\end{assumption}
\begin{remark}
Under Assumption \ref{assumption:1}, the function $s^*(\cdot)$ satisfies the conditional margin assumption restricted to $\mathcal{C}$ of order $\bar\gamma = 1$ with respect to probability distribution $P_0$ at the level $C^{**}_{\alpha}$.  In addition, the constant $M_0$ can be taken as $M_0 = 4 \phi(0)/\sqrt{C}$.   
\end{remark}

Unlike the classical paradigm where the optimal threshold $1/2$ on regression function is known, the optimal threshold level in the NP paradigm is unknown and needs to be estimated, suggesting the necessity of having sufficient data around the decision boundary to detect it. This concern motivated \cite{Tong.2013} to formulate a detection condition that works as an opposite force to the margin assumption, and \cite{Zhao2016neyman} improved upon it and proved its necessity in bounding excess type II error of an NP classifier.  However, formulating a transparent  detection condition for feature spaces of unbounded support is subtle: to generalize the detection condition in the same way as we generalize the margin assumption to a conditional version, it is not obvious what elementary general assumptions one should impose on the $\mu^0$, $\mu^1$ and $\Sigma$.  The good side is that we are able to  establish explicit  conditions for $s\leq 2$, aided by the literature on the truncated normal distribution.     Also, we need a two-sided detection condition as in \cite{Tong.2013}, because the technique in \cite{Zhao2016neyman}  to get rid of one side does not apply in the unbounded feature support situation.

\begin{definition}[conditional detection condition] 
\label{def::detection}
A function $f(\cdot)$ is said to satisfy conditional detection condition restricted to $\mathcal{C}^*$ of order $\uderbar \gamma$ with respect to $P$ (i.e., $X\sim P$) at level $(C^*,\delta^*)$ 
if there exists a positive constant $M_1$, such that for any $\delta \in (0,\delta^*)$, 
$$P\{C^*\leq f(X)  \leq C^* + \delta | X\in\mathcal{C}^*\} \wedge   P\{C^* - \delta \leq f(X)  \leq C^* | X\in\mathcal{C}^*\}\,\geq\, M_1 \delta^{\uderbar\gamma}\,.$$
\end{definition}
Unlike the conditional margin assumption, the conditional detection condition is stronger than its unconditional counterpart, in view of the Bayes Theorem.  Although we do not have a proof of the necessity for the conditional detection condition, much efforts to bound excess type II error without it failed.

\begin{assumption}\label{assumption:2}
The function $s^*(\cdot)$ satisfies conditional detection condition restricted to $\mathcal{C}$ (defined in Lemma \ref{lem:large_prob_set}) of order $\uderbar\gamma\geq 1$ with respect to $P_0$ at the level $(C^{**}_{\alpha}, \delta^*)$.  
\end{assumption}

Proposition \ref{prop:conddetec} in the Appendix shows that under restrictive settings ($s\leq 2$), Assumption \ref{assumption:2} can be implied by more elementary assumptions on the parameters of the LDA model.

\subsection{NP oracle inequalities}
Having introduced the technical assumptions and lemmas, we present the main theorem.  

%\textcolor{red}{Should write something before the Theorem}

%\textcolor{blue}{We will write out the proof first and then summarize the result.}
\begin{theorem}\label{thm:main}
Suppose Assumptions  \ref{assumption:1} and \ref{assumption:2}, and the assumptions for Lemmas 1-4 hold. Further suppose 
$n_0'\geq \max\{4/(\alpha\alpha_0), \delta_0^{-2}, (\delta'_0)^{-2}, (\frac{1}{10} M_1\delta^{*\uderbar{\gamma}})^{-4}\}$, 
$n_0 \wedge n_1 \geq [- \log (M_1 \delta^{*\uderbar{\gamma}}/4)]^2$, and $C_{\alpha}$ and $\mu_a^{\top} \Sigma^{-1} \mu_d$ are bounded from above and below. For $\delta_0, \delta_0' >0$, $\delta_1 \geq \delta_1^*$ and $\delta_2\geq \delta_2^*$,  there exist constants $\bar{c}_1$, $\bar{c}_2$ and $\bar{c}_3$ such that, with probability at least $1 - \delta_0 - \delta_0' - \delta_1 - \delta_2$, it holds that
\begin{eqnarray*}
&\text{(I)}& R_0(\hat \phi_{k^*})\leq \alpha\,,\\
&\text{(II)}& R_1(\hat \phi_{k^*}) - R_1(\phi^*_{\alpha})
\leq \bar{c}_1 (n_0')^{- \frac{1}{4}(\frac{1+\bar\gamma}{\uderbar{\gamma}}\wedge 1)} + \bar{c}_2  (\lambda s)^{1 + \bar\gamma} (n_0\wedge n_1)^{\frac{1 + \bar\gamma}{4}} \\
&&\quad\quad\quad\quad\quad\quad\quad\quad\quad+ \bar{c}_3 \exp\left\{-(n_0\wedge n_1)^{\frac{1}{2}}(\frac{1+\bar\gamma}{\uderbar{\gamma}}\wedge 1)\right\}\,.
\end{eqnarray*}

\end{theorem}
Theorem \ref{thm:main} establishes the NP oracle inequalities for the \verb+NP-sLDA+ classifier $\hat \phi_{k^*}$.  Note that the upper bound for excess type II error does not contain the overall feature dimensionality $d$ explicitly. However, the indirect dependency is two-fold: first, the choice of $\lambda$ might depend on $d$; second, the minimum requirements (i.e., lower bounds) for $\delta_1$ and $\delta_2$, which are $\delta_1^*$ and $\delta_2^*$ defined in Proposition \ref{thm:1}, depend on $d$. 
%By Assumption \ref{assumption:1}, $\bar\gamma = 1$ and by Proposition \ref{prop:conddetec} in the Appendix, $\uderbar{\gamma} = 1$ for $s\leq 2$ under certain conditions.  Substituting in these numbers will greatly simplify the upper bound for the excess type II error. But we choose to keep $\bar\gamma$ so that the explicit dependency on this parameter is clear. 

By Assumption \ref{assumption:1}, $\bar \gamma = 1$ and by Proposition \ref{prop:conddetec} in the Appendix, $\uderbar{\gamma} = 1$ under certain conditions.  Take the special case $\bar \gamma = \uderbar{\gamma}= 1$  and take $\lambda \sim \sqrt{\frac{\log d}{(n_0+n_1)}}$, the upper bound on the excess type II error can be simplified to $B_1 = \bar c_1 (n'_0)^{-1/4}+ \bar c_2 s^2\frac{\log d}{\sqrt{n_0 + n_1}} + \bar c_3 \exp\{-(n_0 \wedge n_1)^{1/2}\}$, where $\bar c_1, \bar c_2, \bar c_3$ are generic constants. The upper bound on excess type II errors when $\bar \gamma = \uderbar \gamma = 1$ in \cite{Zhao2016neyman} [nonparametric plug-in estimators for densities, high-dimensional settings, feature independence assumption] is $B_2 = \bar c_1 (n'_0)^{-1/4} + \bar c_2 s^2\left(\frac{\log n_0}{n_0}\right)^{2\beta/(2\beta+1)} + \bar c_3s^2\left(\frac{\log n_1}{n_1}\right)^{2\beta/(2\beta+1)}$,  where $\beta$ is a smoothness parameter of the kernel function and the densities.  Note that the rate about the left-out class 0 sample size $n'_0$ (for threshold estimate) is the same.   This is because although \verb+NP-sLDA+ relies on an optimal order $k^*$ from the NP umbrella algorithm  while the classifier in \cite{Zhao2016neyman} uses the order $k'$ specified in equation (19) of the current paper, analytic approximation of $k^*$ falls back to $k'$. However, we should note that $B_1$ and $B_2$ are not directly comparable due to the simplifying feature independence assumption and a screening stage in \cite{Zhao2016neyman}. Concretely, there is a marginal screening step before constructing the scoring function and threshold, and that requires reserving some class 0 and class 1 observations.  The sample sizes of these reserved observations, as well as the full feature space dimensionality $d$,  do not enter the bound $B_2$ because the theory part of \cite{Zhao2016neyman} assumes a minimum sample size requirement on the observations for screening in terms of $d$. So the effect of the screening only enters as a probability compromise as opposed to an extra term in the upper bound.  Moreover, after the marginal screening step, \cite{Zhao2016neyman} effectively dealt with $s$ one-dimensional problems due to the feature independence assumption. This explains the appearance of the typical exponent $2\beta/(2\beta+1)$ in one-dimensional nonparametric estimation. Without the feature independence assumption, we would see the exponent as $2\beta/(2\beta + s)$. For a typical $\beta$ value $\beta = 2$ and a moderate $s$ value $s = 20$, we have $2\beta/(2\beta + s) = 1/6$ which is smaller than $1/2$. Hence without the feature independence assumption, the second and third terms in $B_2$ could represent slower rates in terms of the sample sizes $n_0$ and $n_1$ compared to the counterparts in $B_1$. Overall, if $n_0, n_0', n_1 \sim n$, taking $\beta=2$ and $s=20$ without the feature independence assumption, we have $B_2\sim (\log n/ n)^{1/6}$ while $B_1 \sim n^{-1/4}$; the parametric assumptions result in a better rate in this case.  Moreover, in important applications such as severe disease diagnosis, the sample size $n_0$ is much smaller than $n_1$. In $B_1$, these sample sizes appear together as $n_0 + n_1$, but in $B_2$, $n_0$ appears alone as in $(\log n_0 / n_0)^{2\beta/(2\beta+1)}$, which is likely much larger than both $(\log n_1 / n_1)^{2\beta/(2\beta+1)}$ and $1/\sqrt{n_0+ n_1}$,   considering the $n_0 \ll n_1$ situation.

\section{Data-adaptive sample splitting scheme\label{sec::adaptive-split}}

In practice, researchers and practitioners  are not given data as separate sets $\mathcal{S}_0$, $\mathcal{S}_0'$ and $\mathcal{S}_1$.  Instead, they have a single dataset  $\mathcal{S}$ that consists of mixed class $0$ and class $1$ observations. More class $0$ observations to better train the base algorithm and more class $0$ observations to provide more candidates for threshold estimate each has its own merits.     Hence how to split the class $0$ observations into two parts, one to train the base algorithm and the other to estimate the score threshold, is far from intuitive. \cite{tong2016np} adopts a half-half split for class $0$ and ignores this issue in the development of the NP umbrella algorithm, as that paper focuses on the type I error violation rate control.   

Now switching the focus to type II error, we propose a data-adaptive splitting scheme that universally enhances the power (i.e., reduces type II error) of NP classifiers, as demonstrated in the subsequent numerical studies.   The procedure is to choose a split proportion $\tau$ according to rankings of $K$-fold cross-validated type II errors. Concretely, for each split proportion candidate  $\tau\in \{.1, .2, \cdots, .9\}$, the following steps are implemented.   
\begin{enumerate}
\item Randomly split class $1$ observations into $K$-folds. 
\item Use all class $0$ observations and $K-1$ folds of class $1$ observations to train an NP classifier. For class $0$ observations, $\tau$ proportion is used to train the base algorithm, and $1-\tau$ proportion for threshold estimate.  
\item For this classifier, calculate its classification error on the validation fold of the class $1$ observations (type II error). 
\item Repeat steps (2) and (3) for $K$ times, with each of the $K$ folds used exactly once as the validation data.  Compute the mean of type II errors in step (3), and denote it by $e(\tau)$. 
\end{enumerate}
%We denote the average type II error across the $K$-folds by $\text{ave}(\tau)$, and then 
Our choice of the split proportion is 
$$\tau_{\min}=\argmin_{\tau\in\{.1, \cdots, .9\}} e(\tau)\,.$$ %In practice, we choose $\eta$ among $\{0.1, 0.2, \cdots, 0.9\}$.
Note that $\tau_{\min}$ not only depends on the dataset $\mathcal{S}$, but also on the base algorithm and the thresholding algorithm one uses, as well as on the user-specified $\alpha$ and $\delta_0$.  
Merits of this adaptive splitting scheme will be revealed in the simulation section. Here we elaborate on how to reconcile this adaptive scheme with the violation rate control objective.     The type I error violation rate control was proved based on a fixed split proportion of class $0$ observations, so will the adaptive splitting scheme be overly aggressive on type II error such that we can no longer keep the type I error violation rate under control? If for each realization (among infinite realizations) of the mixed sample $\mathcal{S}$, we do adaptive splitting on class $0$ observations before implementing  \verb+NP-sLDA+ $\hat \phi_{k^*}$ (or other NP classifiers), then the overall procedure indeed does not lead to a classifier with type I error violation rate controlled under $\delta_0$. However, this is not how we think about this process; instead, we only adaptively split for one realization of $\mathcal{S}$, getting a split proportion $\hat \tau$, and then fix $\hat \tau$ in all rest realizations of $\mathcal{S}$. This implementation of the overall procedure keeps the type I error violation rate under control.

%\textcolor{blue}{This adaptive splitting procedure being applied to the NP umbrella algorithm does not violate the type I error bound, since we derive the adaptive split proportion on a separate dataset generated from the same mechanism. Later we fix this split proportion and consider it as a constant when running repetitions of simulations. The merits of this adaptive splitting will be revealed in the next section. } 
%\textcolor{red}{Here we should state clearly in what sense we think adding this adaptive splitting procedure to the NP umbrella algorithm does not violate the type I error bound. Also say the merits of this adaptive splitting is revealed in the next section}

%\textcolor{red}{Note that we have changed $\eta$ to $\tau$ because }

\section{Numerical analysis}

%\textcolor{red}{fix $\alpha$, $\delta_0$, explore the chosen order $k^*$ as a function of $n$.  }

Through extensive simulations and real data analysis, we study the performance of  \verb+NP-LDA+, \verb+NP-sLDA+, \verb+pNP-LDA+ and \verb+pNP-sLDA+. In addition, we will study the new adaptive sample splitting scheme. In this section, $N_0$ denotes the total class $0$ training sample size (we do not use $n_0$ and $n_0'$ here, as class $0$ observations are not assumed to be pre-divided into two parts),  $n_1$ denotes the class $1$ training sample size, and $N = N_0 + n_1$ denotes the total sample size.
\subsection{Simulation studies under low-dimensional settings}
In this subsection, we consider the low-dimensional settings with two examples.  In particular, we would like to compare the performance of four NP based methods: \verb+NP-LDA+, \verb+NP-sLDA+, \verb+pNP-LDA+ and \verb+pNP-sLDA+. In all NP methods, $\tau$, the class $0$ split proportion, is fixed at $0.5$. In every simulation setting, the experiments are repeated $1,000$  times. 

\begin{figure}[t]
\caption{Example 1. 
Type II error of \texttt{NP-LDA}, \texttt{NP-sLDA}, \texttt{pNP-LDA} and \texttt{pNP-sLDA}  vs. $N_0$ or $d$. The dashed red line represent the NP oracle.  \label{fig::example 1}}
 \begin{subfigure}[t]{0.5\textwidth}
        \centering
        \includegraphics[scale=0.39]{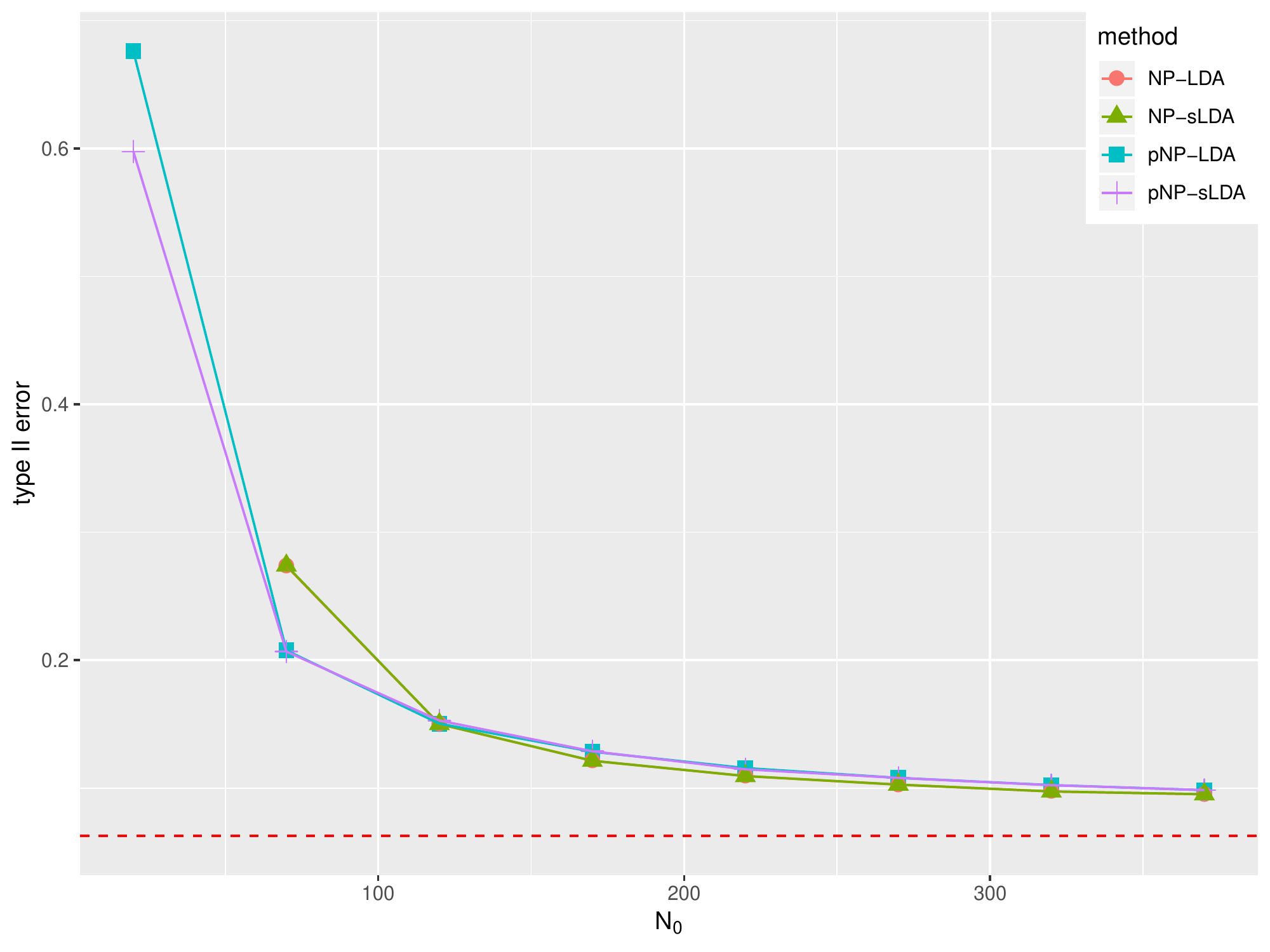}
        \caption{Example 1a}
    \end{subfigure}%
    \begin{subfigure}[t]{0.5\textwidth}
        \centering
        \includegraphics[scale=0.39]{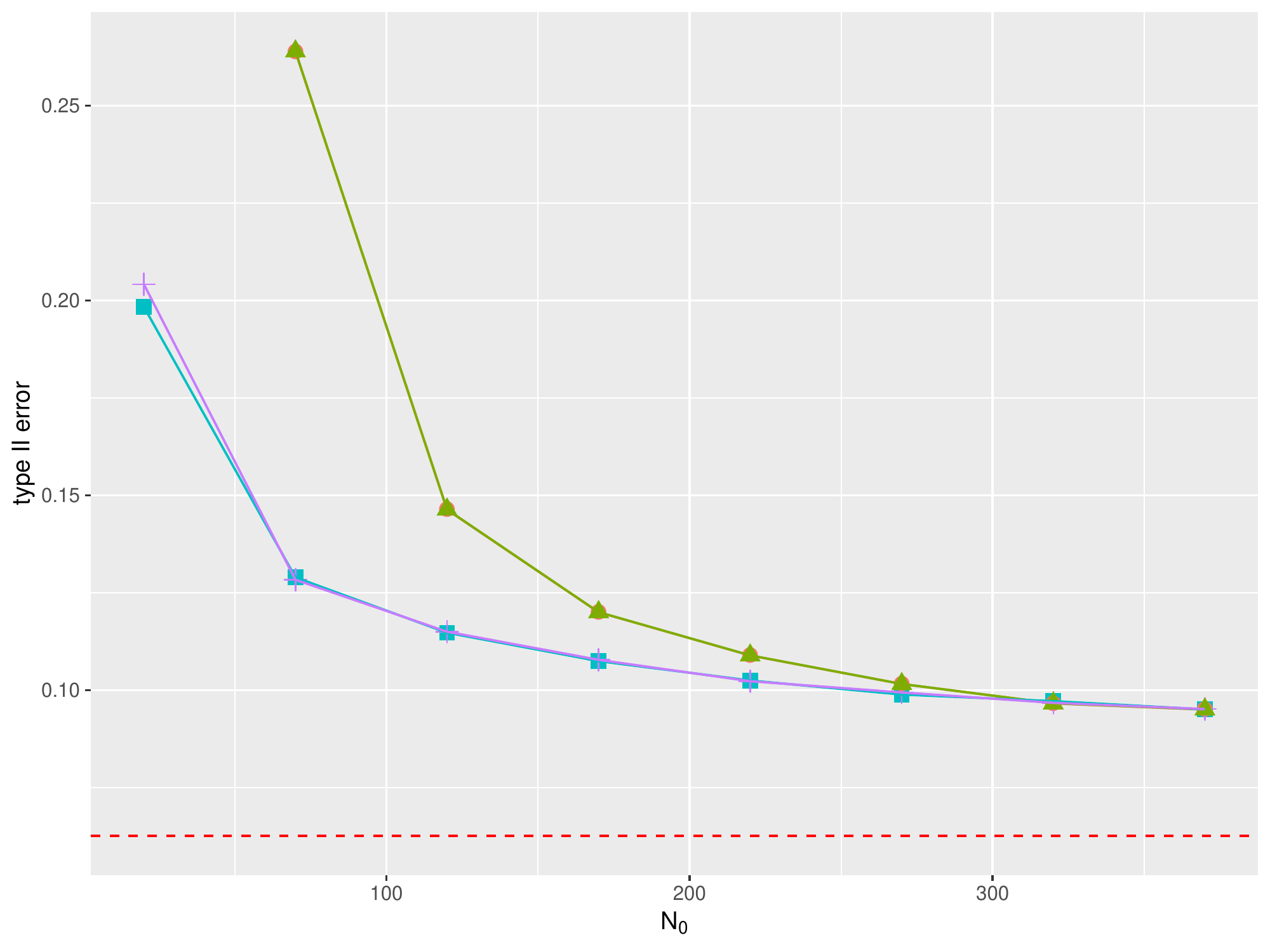}
        \caption{Example 1b}
    \end{subfigure}%    \begin{subfigure}[t]{0.5\textwidth}
    
    \begin{subfigure}[t]{0.5\textwidth}
        \centering
        \includegraphics[scale=0.39]{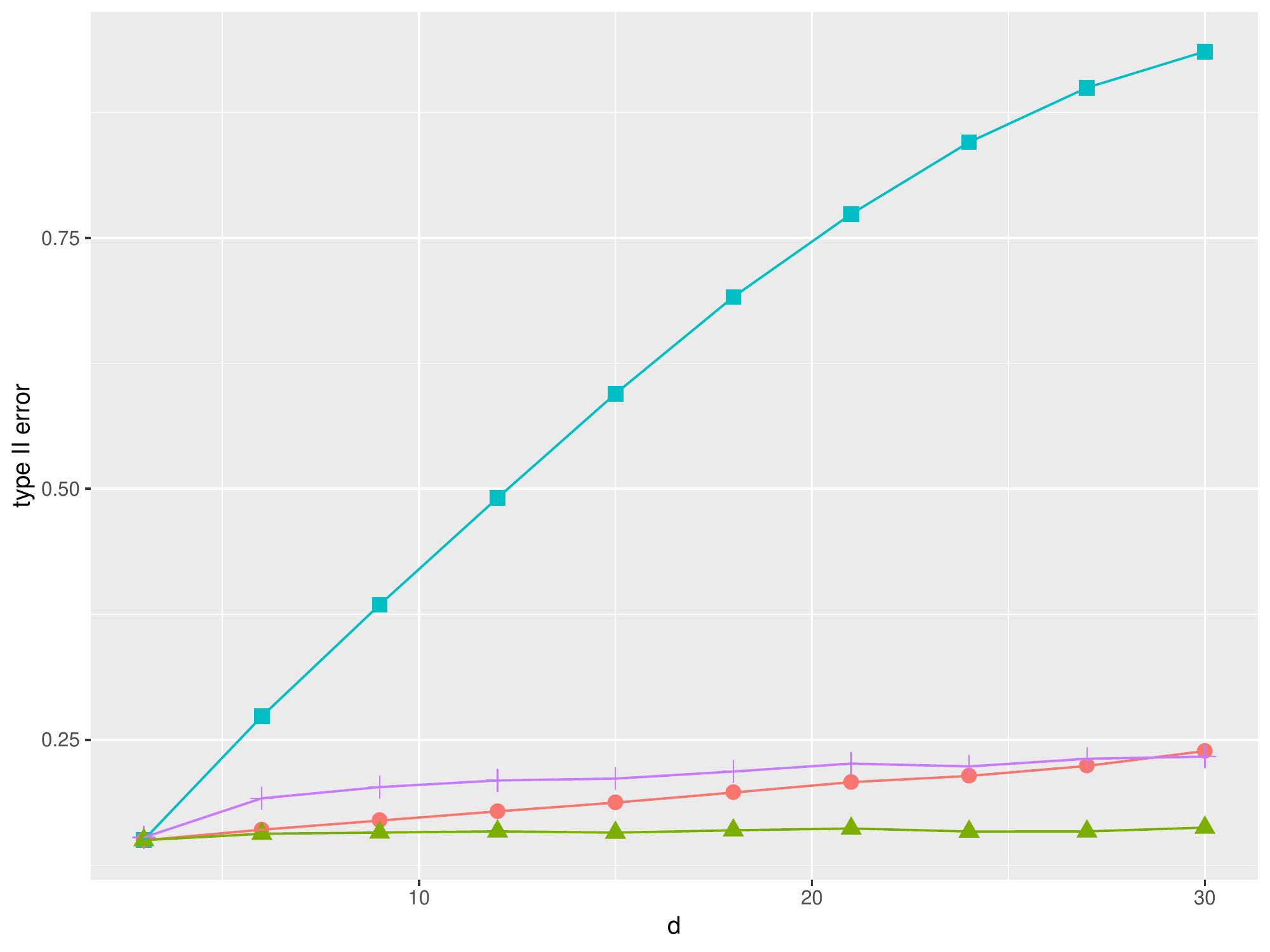}
        \caption{Example 1c}
    \end{subfigure}%    \begin{subfigure}[t]{0.5\textwidth}
    \begin{subfigure}[t]{0.5\textwidth}
        \centering
        \includegraphics[scale=0.39]{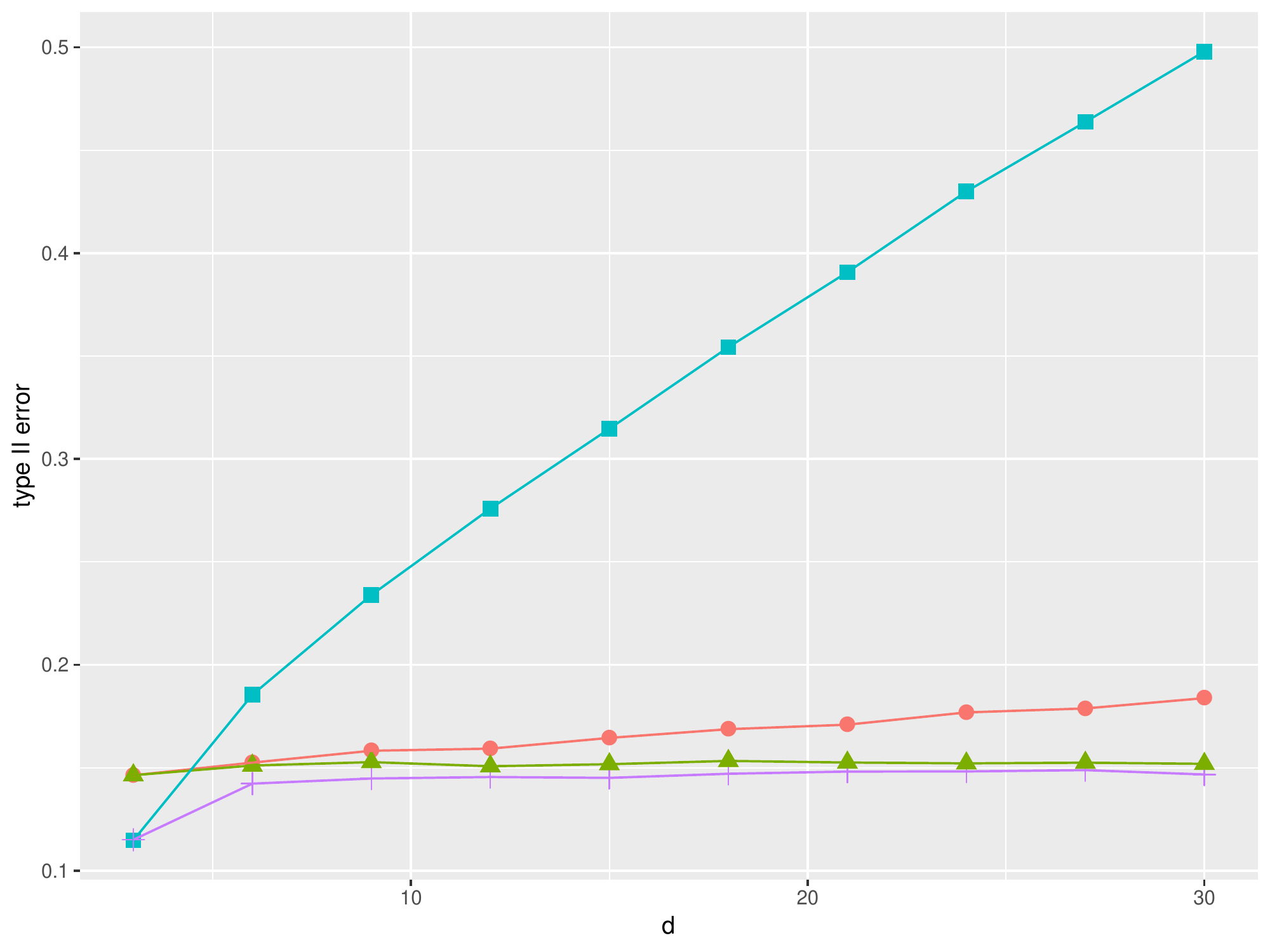}
        \caption{Example 1d}
    \end{subfigure}%
    
%\begin{center} 
%
%\includegraphics[scale=0.39]{figs/Ex1_typeII_vs_n0=n1_d=3.pdf}    
%\includegraphics[scale=0.39]{figs/Ex1_typeII_vs_n0_n1=500_d=3.pdf}    
%\includegraphics[scale=0.39]{figs/Ex1_typeII_vs_d_n0_n1=120.pdf}    
%\includegraphics[scale=0.39]{figs/Ex1_typeII_vs_d_n0=120_n1=500.pdf}    
%\end{center}

\end{figure}
%We consider the following two examples. 
\begin{example}\label{ex-1}
The data are generated from an LDA model with common covariance matrix $\Sigma$, where $\Sigma$ is set to be an AR(1) covariance matrix with $\Sigma_{ij}=0.5^{|i-j|}$ for all $i$ and $j$. The $\beta^{\text{Bayes}}=\Sigma^{-1}\mu_d = 1.2\times(1_{d_0},0_{d-d_0})^{\top}$, $\mu^0=0_d$, $d_0 = 3$. We set $\pi_0 = \pi_1 = 0.5$ and $\alpha=\delta_0=0.1$. 
\begin{enumerate}[(\ref{ex-1}a).]
\item $d=3$, varying $N_0=n_1\in \{20, 70, 120, 170, 220, 270, 320, 370\}$.
\item $d=3$, $n_1=500$, varying $N_0 \in \{20, 70, 120, 170, 220, 270, 320, 370\}$.
\item $N_0=n_1=125$, varying $d\in\{3, 6, 9, 12, 15, 18, 21, 24, 26, 30\}$.
\item $N_0=125$, $n_1=500$, varying $d\in\{3, 6, 9, 12, 15, 18, 21, 24, 26, 30\}$.
\end{enumerate}
\end{example}
The results are summarized in the Figure  \ref{fig::example 1}. Several observations are made in order.  First, from the first row of the figure (\ref{ex-1}a and \ref{ex-1}b), we observe that when $N_0$ is very small, the implementable methods that can achieve the desired type I error control are \verb+pNP-LDA+ and \verb+pNP-sLDA+; the NP umbrella algorithm based methods fail its minimum class 0 sample size requirement. Second,  the type II error of all methods decreases as $N_0$ increases from the first row of the figure (\ref{ex-1}a and \ref{ex-1}b),  and increases when $d$ increases from the second row of the figure (\ref{ex-1}c and \ref{ex-1}d). Third, from the first row of the figure, we see  \verb+pNP-LDA+ and \verb+pNP-sLDA+ have advantages over \verb+NP-LDA+ and \verb+NP-sLDA+ when the sample sizes are small. Finally, from the second row of the figure, we see that the nonparametric NP umbrella algorithm gains more and more advantages over the parametric thresholding rule (\verb+pNP-LDA+) as $d$ increases, since $\widehat C^p_{\alpha}$ specified in \eqref{eqn:parametric threshold} can become loose when $d$ is large. It is worth to mention that by taking advantage of the sparse solution generated by \texttt{sLDA}, the performance of \verb+pNP-sLDA+ does not deteriorates  as $d$ increases and performs the best for Example \ref{ex-1}d. 

%  in fact tThis is  mainly due to that the parametric threshold $\widehat C^p_{\alpha}$ specified in \eqref{eqn:parametric threshold} can become loose for a larger $d$ even if we use the modification that takes advantage of sparsity. 

\begin{example}\label{ex-2}
The data are generated from an LDA model with common covariance matrix $\Sigma$, where $\Sigma$ is set to be AR(1) covariance matrix with $\Sigma_{ij}=0.5^{|i-j|}$ for all $i$ and $j$. The $\beta^{\text{Bayes}}=\Sigma^{-1}\mu_d = C_d \cdot 1_d^{\top}$, $\mu^0=0_d$. We set $\pi_0 = \pi_1 = 0.5$ and $\alpha=\delta_0=0.1$. Here $C_d$ is a constant depending on $d$ such that the oracle classifier always has type II error 0.112 for any choice of $d$.
\begin{enumerate}[(\ref{ex-2}a).]
\item $d=3$, varying $N_0=n_1\in \{20, 70, 120, 170, 220, 270, 320, 370\}$.
\item $d=6$, varying $N_0=n_1\in \{20, 70, 120, 170, 220, 270, 320, 370\}$.
\item $N_0=n_1=125$, varying $d\in\{3, 6, 9, 12, 15, 18, 21, 24, 26, 30\}$.
\item $N_0=125$, $n_1=500$, varying $d\in\{3, 6, 9, 12, 15, 18, 21, 24, 26, 30\}$.
\end{enumerate}
\end{example}

\begin{figure}[t]
\caption{Example 2. 
Type II error of \texttt{NP-LDA}, \texttt{NP-sLDA}, \texttt{pNP-LDA} and \texttt{pNP-sLDA} vs. $N_0$ or $d$. The dashed red line represents the NP oracle.  \label{fig::ex 2}}

 \begin{subfigure}[t]{0.5\textwidth}
        \centering
        \includegraphics[scale=0.39]{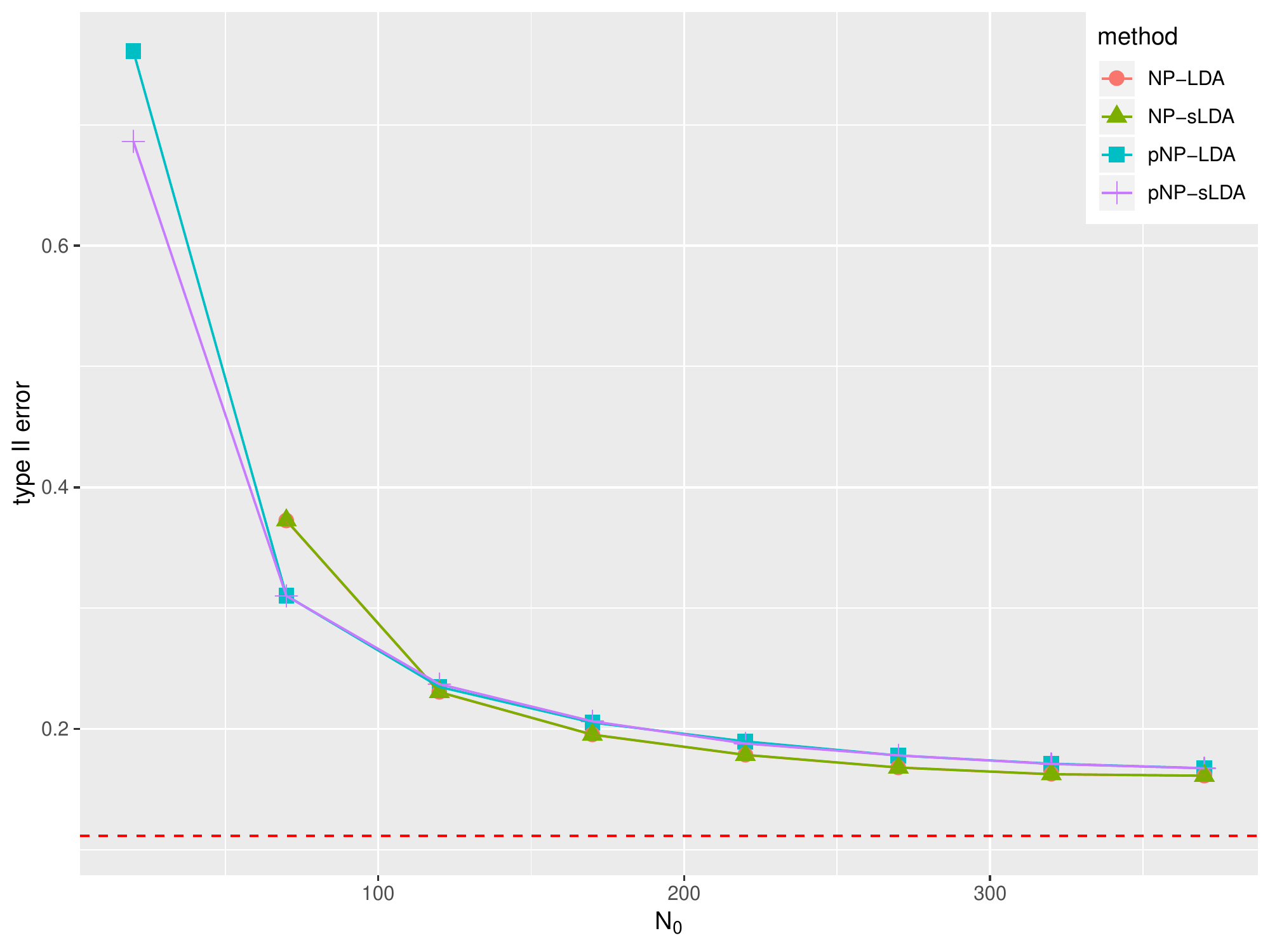}
        \caption{Example 2a}
    \end{subfigure}%
    \begin{subfigure}[t]{0.5\textwidth}
        \centering
        \includegraphics[scale=0.39]{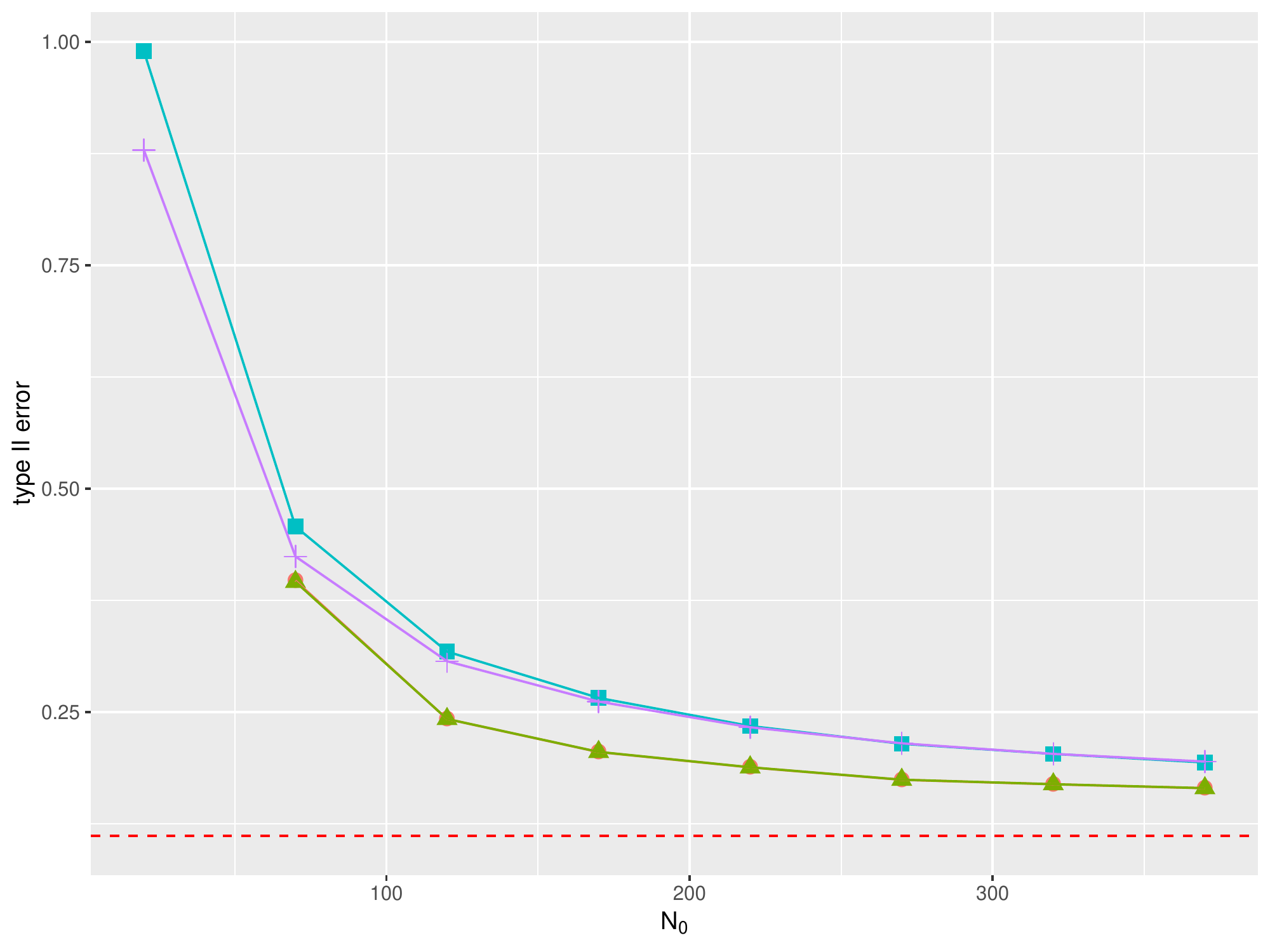}
        \caption{Example 2b}
    \end{subfigure}%    \begin{subfigure}[t]{0.5\textwidth}
    
    \begin{subfigure}[t]{0.5\textwidth}
        \centering
        \includegraphics[scale=0.39]{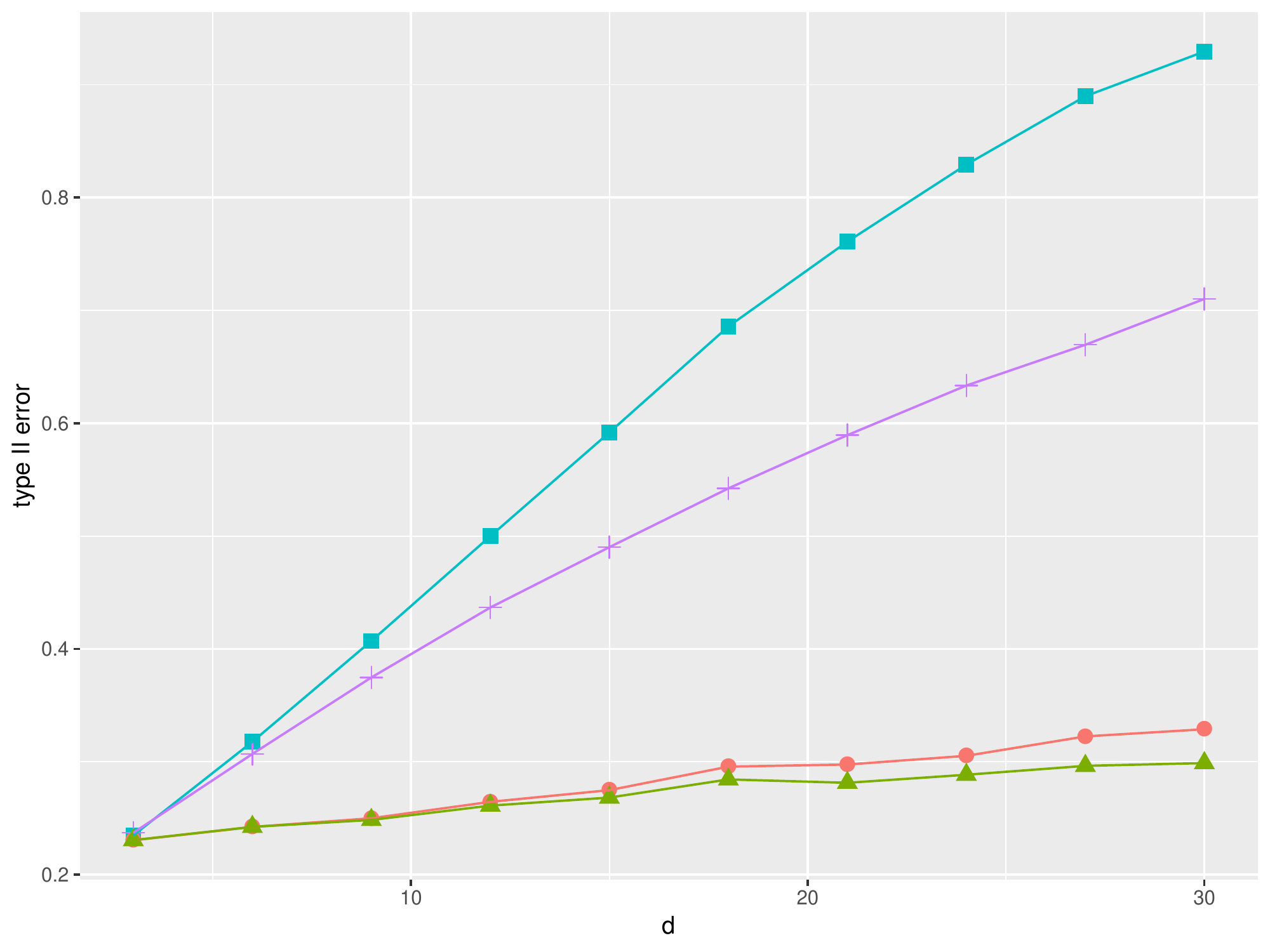}
        \caption{Example 2c}
    \end{subfigure}%    \begin{subfigure}[t]{0.5\textwidth}
    \begin{subfigure}[t]{0.5\textwidth}
        \centering
        \includegraphics[scale=0.39]{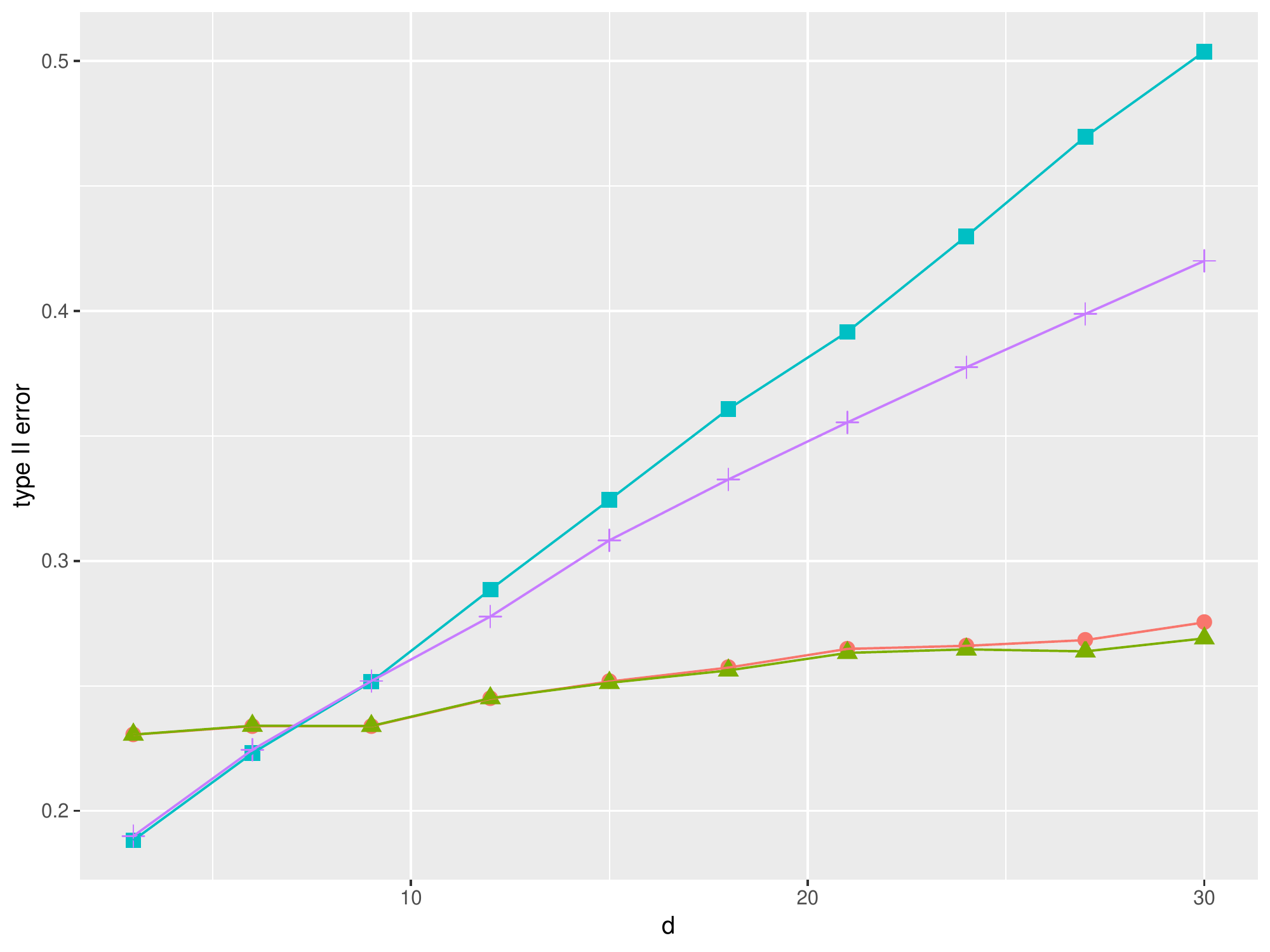}
        \caption{Example 2d}
    \end{subfigure}%
    
%\begin{center}
%\includegraphics[scale=0.39]{figs/Ex2_typeII_vs_n0=n1_d=3}    
%\includegraphics[scale=0.39]{figs/Ex2_typeII_vs_n0=n1_d=6}    
%\includegraphics[scale=0.39]{figs/Ex2_typeII_vs_d_n0_n1=120.pdf}    
%\includegraphics[scale=0.39]{figs/Ex2_typeII_vs_d_n0=120_n1=500.pdf}    
%\end{center}

\end{figure}
Example \ref{ex-2} is a more challenging scenario where the oracle rule depends on all features. Similar observations as in Example  \ref{ex-1} can be made for Example \ref{ex-2}  from Figure \ref{fig::ex 2}. It is worth to mention that in this case, although we still see improvement of \verb+pNP-sLDA+ over \verb+pNP-LDA+ throughout all $d$,  the performance of \verb+pNP-sLDA+ is dominated by \verb+NP-LDA+ and \verb+NP-sLDA+ for moderate sample sizes as all features are important.

From Figure \ref{fig::example 1} (a)(b) and Figure \ref{fig::ex 2}(a)(b), we see that as sample sizes increase, the performance of \verb+NP-sLDA+ gets better and eventually dominates \verb+pNP-sLDA+ even though the latter takes advantage of the parametric model assumption in both training the scoring function and constructing the threshold.  This might seem a little counter-intuitive at  first glance.  The reason is that the construction of threshold estimate $\widehat C_{\alpha}$ in \verb+pNP-sLDA+ replies on a high probability upper bound on the inaccessible model-specific oracle threshold (see Lemma \ref{lem:iii}).  The construction of this upper bound involves bounding a quadratic form of $\Sigma$ by $\lambda_{\max}(\Sigma)$ and studying the relations between $\lambda_{\max}(\Sigma)$ and $\lambda_{\max}(\widehat\Sigma)$ without structural assumptions on $\Sigma$. Thus, the upper bound can be on the conservative side (i.e., larger than what is necessary) for a specific covariance structure. On the other hand, when the sample sizes are large, the number of candidate thresholds in \verb+NP-sLDA+ becomes large, then one can choose an order $k$ in the nonparametric NP umbrella algorithm to make the violation rate $v(k)$ very close to $\delta_0$ in equation \eqref{eqn:kstar}. Thus the loss due to universal handling of covariance structure in \verb+pNP-sLDA+ may outweigh the loss due to discretization in \verb+NP-sLDA+.

\subsection{Simulation studies under high-dimensional settings}
%\textcolor{red}{change all $\eta$ to $\tau$, and ``ratio" to ``proportion"; also I have change the $e$ notation to $\text{ave}$.  }

  In Examples \ref{ex-3}$-$\ref{ex-5}, we conduct simulations to compare the empirical performance of the proposed \verb+NP-sLDA+ and \verb+pNP-sLDA+ with other non-LDA based NP classifiers as well as the \verb+sLDA+ \citep{mai2012direct}.  In every simulation setting, the experiments are repeated $1,000$  times.

\begin{example}\label{ex-3}
The data are generated from an LDA model with common covariance matrix $\Sigma$, where $\Sigma$ is set to be an AR(1)  covariance matrix with $\Sigma_{ij}=0.5^{|i-j|}$ for all $i$ and $j$. The $\beta^{\text{Bayes}}=\Sigma^{-1}\mu_d = 0.556\times(3,1.5,0,0,2,0,\cdots,0)^{\top}$, $\mu^0=0^{\top}$, $d = 1,000$, and $N_0 = n_1 = 200$. Bayes error = $10\%$ under $\pi_0 = \pi_1 = 0.5$. $\alpha = \delta_0 = 0.1$. %\textcolor{red}{when we talk about Bayes error, we need class prior}
\end{example}

\begin{example}\label{ex-4}
The data are generated from an LDA model with common covariance matrix $\Sigma$, where $\Sigma$ is set to be a compound symmetric covariance matrix with $\Sigma_{ij}=0.5$ for all $i\neq j$ and $\Sigma_{ii}=1$ for all $i$. The $\beta^{\text{Bayes}}=\Sigma^{-1}\mu_d = 0.551\times(3,1.7,-2.2,-2.1,2.55,0,\ldots,0)^{\top}$, $\mu^0=0^{\top}$, $d = 2,000$, and $N_0 = n_1 = 300$. Bayes error = $10\%$ under $\pi_0 = \pi_1 = 0.5$.  $\alpha = \delta_0 = 0.1$.
\end{example}
\begin{example}\label{ex-5}
Same as in Example \ref{ex-4}, except $d=3,000$, $N_0=n_1=400$, and the $\beta^{\text{Bayes}}=\Sigma^{-1}\mu_d = 0.362\times(3,1.7,-2.2,-2.1,2.55,0,\ldots,0)^{\top}$. Bayes error = $20\%$ under $\pi_0 = \pi_1 = 0.5$. $\alpha = 0.2$ and $\delta_0 = 0.1$.

\end{example}

In Examples \ref{ex-3}$-$\ref{ex-5}, we compare the empirical type I/II error performance of \verb+NP-sLDA+, \verb+NP-penlog+ ({\it penlog} stands for penalized logistic regression), \verb+NP-svm+, \verb+pNP-sLDA+, and \verb+sLDA+ on a large test data set of size $20,000$ that consist of $10,000$ observations from each class. In all NP methods, $\tau$, the class $0$ split proportion, is fixed at $0.5$. The choices for $\alpha$ in these examples match the corresponding Bayes errors, so that the comparison between NP and classical methods does not obviously favor the former.    

\begin{table}
\caption{Violation rate and type II error for Examples \ref{ex-3}, \ref{ex-4} and \ref{ex-5} over $1,000$ repetitions.\label{tb::simu1}}
\begin{center}
\begin{tabular}{l|lrrrrr}
\hline
&&\texttt{NP-sLDA}&\texttt{NP-penlog}&\texttt{NP-svm}&\texttt{pNP-sLDA}&\texttt{sLDA}\\
\hline
\multirow{3}{*}{Ex 3}&    violation rate&.068&.055&.054&.001&.764\\
&type II error (mean)&.189&.205&.621&.220&.104\\
&type II error (sd)&.057&.063&.077&.052&.010\\\hline
\multirow{3}{*}{Ex 4}&violation rate&.073&.081&.081&.000&1.000\\
&type II error (mean)&.246&.255&.615&.824&.129\\
&type II error (sd)&.051&.053&.070&.121&.010\\
\hline
\multirow{3}{*}{Ex 5}&violation rate&.079&.088&.099&.000&.997\\
&type II error (mean)&.332&.334&.584&.748&.231\\
&type II error (sd)&.044&.044&.045&.128&.012\\
\hline
\end{tabular}
\end{center}
\end{table}
Table \ref{tb::simu1} indicates that  while the classical \verb+sLDA+ method cannot control the type I error violation rate under $\delta_0$, all the NP classifiers are able to do so. In addition, among the four NP classifiers, \verb+NP-sLDA+ gives the smallest mean type II error. \verb+pNP-sLDA+ performs reasonably well in Example \ref{ex-3}, where the NP oracle rule is extremely sparse. In Examples \ref{ex-4} and \ref{ex-5}, however, the threshold $\widehat C_{\alpha}^p$ becomes overly large as a result of more selected features due to the less sparse oracle, thus leading to an overly conservative classifier with 0 violation rate \footnote{Strictly speaking, the observed type I error violation rate is only an approximation to the real violation rate.  The approximation is two-fold: i).  in each repetition of an experiment, the population type I error is approximated by empirical type I error on a large test set; ii). the violation rate should be calculated based on infinite repetitions of the experiment, but we only calculate it based on $1,000$ repetitions. However, such approximation is unavoidable in numerical studies.}. %Some insights regarding this phenomenon may be found through examination of Examples 1 and 2, where the performance of  \verb+LDA+ deteriorates as the dimension increases. 

\subsection{Adaptive sample splitting}
By explanations in the last paragraph of Section \ref{sec::adaptive-split}, the adaptive splitting scheme does not affect the type I error control objective.   
Examples \ref{ex-6} and \ref{ex-7}  investigate the power enhancement as a result of the adaptive splitting scheme over the default half-half choice. These examples include an array of situations, including low-  and high-dimensional settings ($d=20$ and $1,000$), balanced and imbalanced  classes ($N_0:n_1 = 1:1$ to $1:256$), and small to medium sample sizes ($N_0 = 100$ to $500$).  %In Examples $6-7$, we study how the adaptive splitting scheme enhances power (i.e.,  reduces type II error) upon the default half-half choice.  

\begin{example}\label{ex-6}
Same as in Example \ref{ex-3}, except taking the following sample sizes. \begin{enumerate}[(\ref{ex-6}a).]
\item     $N_0=100$ and varying $n_1/N_0 \in \{1,2,4,8,16,32,64,128,256\}$. 
\item  Varying $n_1=N_0\in \{100,150,200,250,300,350,400,450,500\}$.
\end{enumerate}

\end{example}

\begin{example}\label{ex-7} Same as in Example \ref{ex-3}, except that $d=20$, $N_0=100$ and varying $n_1/N_0\in\{1,2,4,8,16\}$. 
\end{example}

Note that Examples \ref{ex-6}a and \ref{ex-6}b each includes $9$ different simulation settings, and Example \ref{ex-7} includes $5$.  For each simulation setting, we generate $1,000$ (training) datasets and a common test set of size $100,000$ from class $1$. Only class $1$ test data are needed because only type II error is investigated in these examples. In each simulation setting, we train $10$ NP classifiers of the same base algorithm using each of the $1,000$ datasets. Nine of these $10$ NP classifiers use fixed split proportions in $\{.1, \cdots, .9\}$, and the last one uses adaptive split proportion using $K=5$. Overall in Examples  \ref{ex-6} and \ref{ex-7}, we set $\alpha=\delta_0=0.1$, and train an enormous number of NP classifiers.  For instance,  in Example \ref{ex-6}a, we train $9 \times 1,000 \times 10 = 90, 000$ \verb+NP-sLDA+ classifiers, and the same number of NP classifiers  for any other base algorithm under investigation. We fix the thresholding rule as the NP umbrella algorithm in this subsection.

For each simulation setting, denote by $\widetilde R_1(\cdot)$ the empirical type II error on the test set. We fix a simulation setting so that we do not need to have overly complex sub or sup indexes in the following discussion. Denote by $\hat h_{i, b, \tau}$ an NP classifier with base algorithm $b$, trained on the $i$th dataset ($i\in\{1, \cdots, 1000\}$) using split proportion $\tau$. This classifier also depends on users' choices of $\alpha$ and $\delta_0$, but we suppress these dependencies here to highlight our focus. In fixed proportion scenarios, $\tau\in\{.1, \cdots, .9\}$. Let $\tau^{\text{ada}}(j,b)$ represent the adaptive split proportion trained on the $j$-th dataset with base algorithm $b$ using adaptive splitting scheme described in Section 5.  Therefore, $\hat h_{i, b, \tau^{\text{ada}}(j,b)}$ refers to the NP classifier with base algorithm $b$, trained on the $i$-th dataset using the split proportion  $\tau^{\text{ada}}(j,b)$ pre-determined in the $j$-th dataset,  where $i, j\in \{1, \cdots, 1000\}$.  Let $\text{Ave}_{b, \tau}$ and $\text{Ave}_{b, \hat \tau}$ be our performance measures for fix proportion and adaptive proportion respectively, which are defined by, 
$$\text{Ave}_{b, \tau} = \frac{1}{1000} \sum_{i=1}^{1000} \widetilde R_1(\hat h_{i, b, \tau})\,, \text{and }\text{Ave}_{b, \hat \tau} = \text{median}_{j = 1, \cdots, 1000}\left(\frac{1}{1000}\sum_{i=1}^{1000}\widetilde R_1\left(\hat h_{i, b, \tau^{\text{ada}}(j,b)}\right)\right)\,.$$
%\textcolor{red}{might need to switch the median to the mean after sheep completes new results.}
While the meaning of the measure  $\text{Ave}_{b, \tau}$ is almost self-evident, $\text{Ave}_{b, \hat \tau}$ deserves some elaboration. As we explained in the last paragraph of Section \ref{sec::adaptive-split}, the adaptive splitting scheme returns a proportion based on one realization of $\mathcal{S}$, and then we adopt it in all subsequent realizations.  Let
$$w_b(j) = \frac{1}{1000}\sum_{i=1}^{1000}\widetilde R_1\left(\hat h_{i, b, \tau^{\text{ada}}(j,b)}\right)\,,$$
then $w_b(j)$ is a performance measure of the adaptive scheme if the proportion is returned from training on the $j$-th dataset.  To account for the variation among $w_b(j)$'s for different choices of $j$, we take the median over $w_b(j)$'s  as our final measure. Also, we denote the average of adaptively selected proportions by $\tau_{b, \text{ada}} = \frac{1}{1000}\sum_{j=1}^{1000}\tau^{\text{ada}}(j,b)$, and define the  average optimal split proportion $\tau_{b,\text{opt}}$ by 
$$
\tau_{b,\text{opt}} = \frac{1}{1000}\sum_{i=1}^{1000}\argmin_{\tau\in\{.1\cdots, .9\}} \widetilde  R_1 (\hat h_{i, b, \tau})\,.
$$

With Example \ref{ex-6}, we investigate i). the effectiveness (in terms of type II error) of the adaptive splitting strategy compared to a fixed half-half split, illustrated by the left panels of Figures \ref{fig::split.proportion} and \ref{fig::n0}; ii). how close is $\tau_{b, \text{ada}}$ compared to $\tau_{b, \text{opt}}$, illustrated by the right panels of Figures \ref{fig::split.proportion} and \ref{fig::n0}; iii). how the class imbalance affects \verb+NP-sLDA+ and \verb+NP-penlog+, illustrated by both panels of Figure \ref{fig::split.proportion}; and iv). how the absolute class $0$  sample size affects \verb+NP-sLDA+ and \verb+NP-penlog+,  
illustrated by both panels of Figure  \ref{fig::n0}.     

In Figure \ref{fig::split.proportion} (Example 6a), the left panel
presents the trend of type II errors ($\text{Ave}_{b, .5}$ and $\text{Ave}_{b, \hat \tau}$) as the sample size ratio $n_1/N_0$ increases from $1$ to $256$ for fixed $N_0 = 100$. For both \verb+NP-penlog+ and \verb+NP-sLDA+, type II error decreases as $n_1/N_0$ increases from $1$ to $16$ and gradually stabilizes afterwards. Neither \verb+NP-penlog+ nor \verb+NP-sLDA+ suffers from training on imbalanced classes.   In terms of type II error performance, the adaptive splitting strategy significantly improves over the fixed split proportion $0.5$. The right panel of Figure \ref{fig::split.proportion} shows that, on average the 
adaptive split proportion is very close to the optimal one throughout all sample size ratios. 

%In Figure \ref{fig::typeII.error}, we plot the trend of the  average type II error as we vary the split proportion for $N_0=n_1=300$. It is clear from the figure that the type II error is very sensitive to the  choice of the split proportion. 

%
%\begin{figure}
%\caption{left panel: optimal split proportion vs. $n_1/N_0$; right panel: optimal split proportion vs. $N_0$\label{fig::split.proportion}}
%
%\includegraphics[scale=0.45]{figs/split_proportion_vs_n1_over_n0.pdf}    
%\includegraphics[scale=0.45]{figs/split_proportion_vs_n0.pdf}    
%
%\end{figure}
\begin{figure}
\caption{Example \ref{ex-6}a. 
Left panel: type II error ($\text{Ave}_{b, .5}$ and $\text{Ave}_{b, \hat \tau}$) of \texttt{NP-sLDA} and \texttt{NP-penlog} vs. $n_1/N_0$; Right panel:  average split proportion ($\tau_{b, \text{ada}}$ and $\tau_{b, \text{opt}}$) vs. $n_1/N_0$. $N_0$ is fixed to be $100$ for both panels. \label{fig::split.proportion}}
\begin{center}
\includegraphics[scale=0.39]{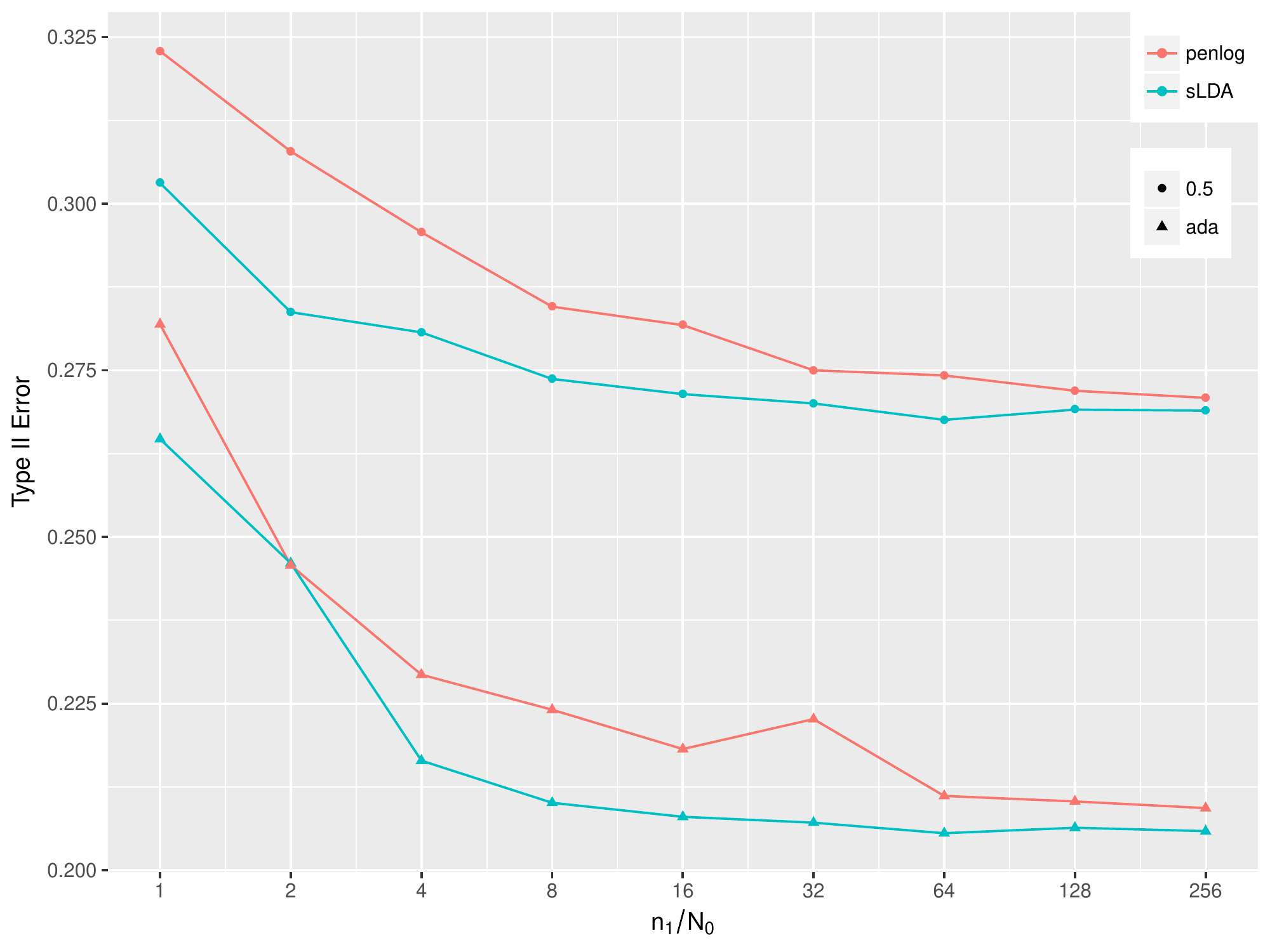}    
\includegraphics[scale=0.39]{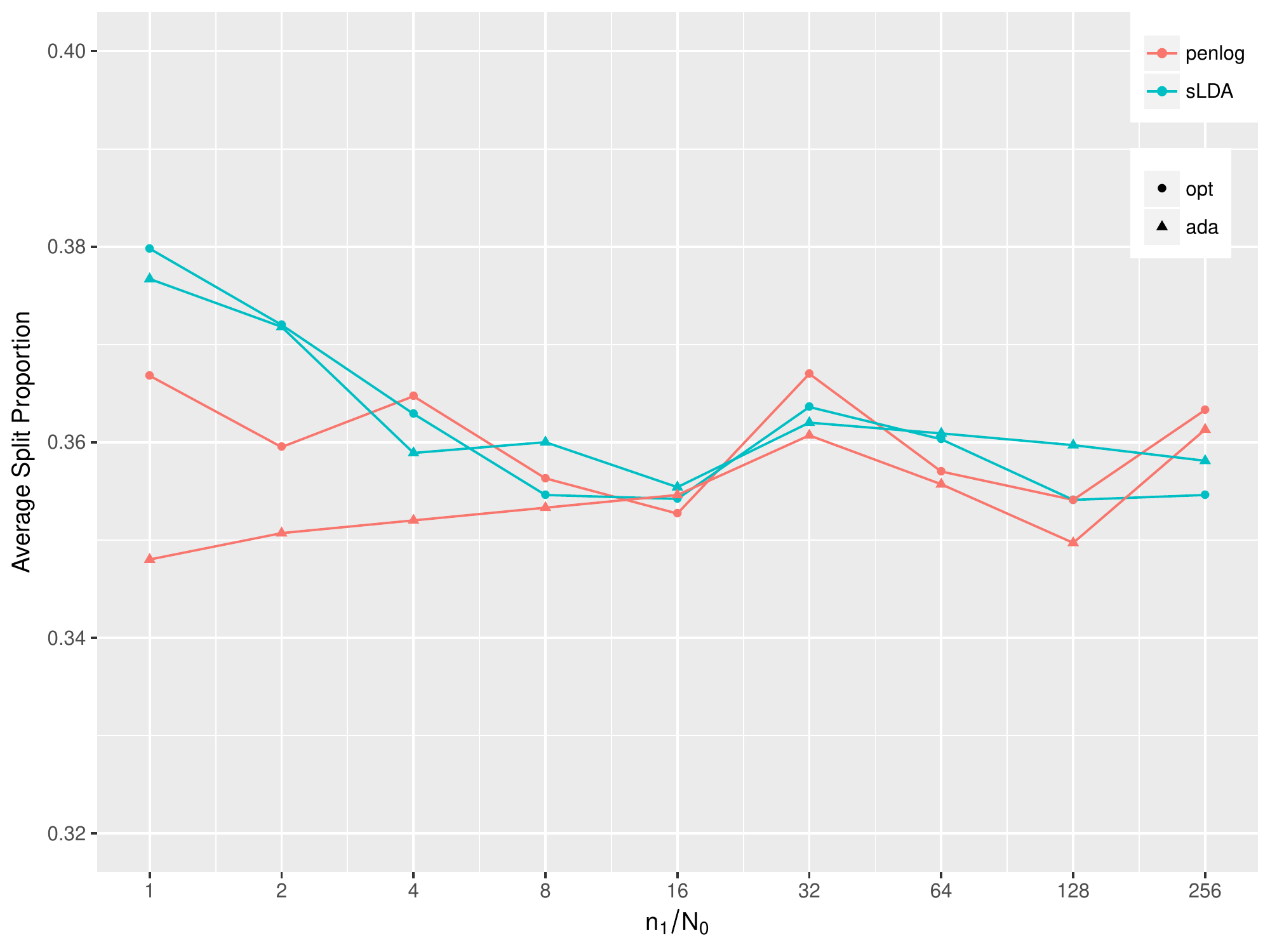}
    
\end{center}

\end{figure}

In Figure \ref{fig::n0} (Example \ref{ex-6}b), the left panel presents the trend of type II errors ($\text{Ave}_{b, .5}$ and $\text{Ave}_{b, \hat \tau}$) as the class $0$ sample  size $N_0$ ($n_1=N_0$) increases from $100$ to $500$, indicating that type II error clearly benefits from increasing training sample sizes of both classes. For the same base algorithm, the adaptive splitting strategy significantly improves over the fixed split proportion $0.5$ for $N_0$ and $n_1$ small, although the improvement diminishes as both sample sizes become large. The right panel of Figure \ref{fig::n0} shows that on average, the adaptive split proportion is very close to the optimal one throughout all sample sizes. Furthermore, the average optimal split proportion seems to increase as $N_0$ increases  in general. An intuition is that when $N_0$ is smaller,  a higher proportion of class $0$ observations is needed for threshold estimate, to guarantee the type I error violation rate control.

\begin{figure}
\caption{Example \ref{ex-6}b. Left panel: type II error ($\text{Ave}_{b, .5}$ and $\text{Ave}_{b, \hat \tau}$) of \texttt{NP-sLDA} and \texttt{NP-penlog} vs. $N_0$; Right panel:  average split proportion ($\tau_{b, \text{ada}}$ and $\tau_{b, \text{opt}}$) vs. $N_0$. $n_1=N_0$ for both panels. \label{fig::n0}}
\begin{center}
\includegraphics[scale=0.39]{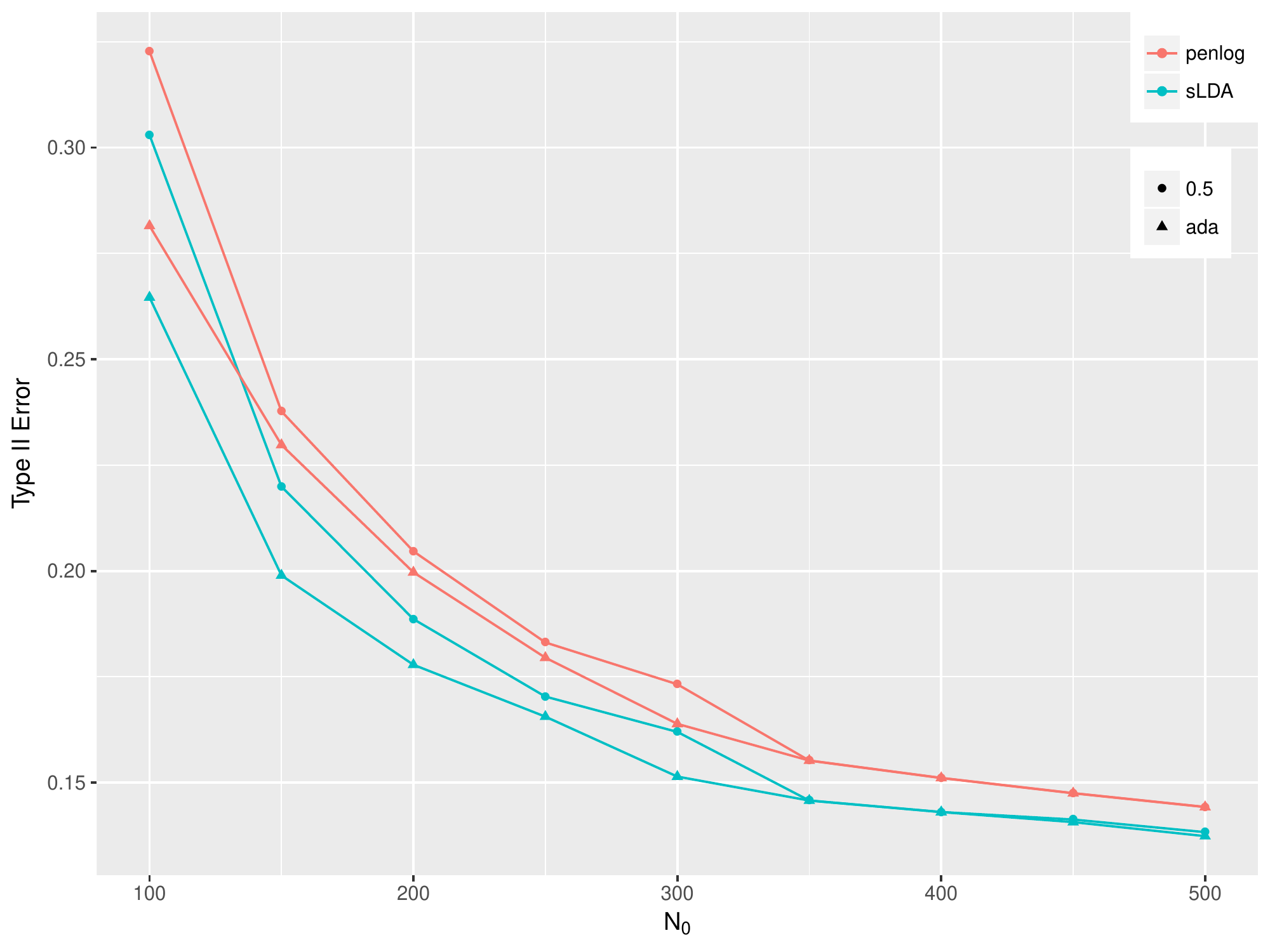}
\includegraphics[scale=0.39]{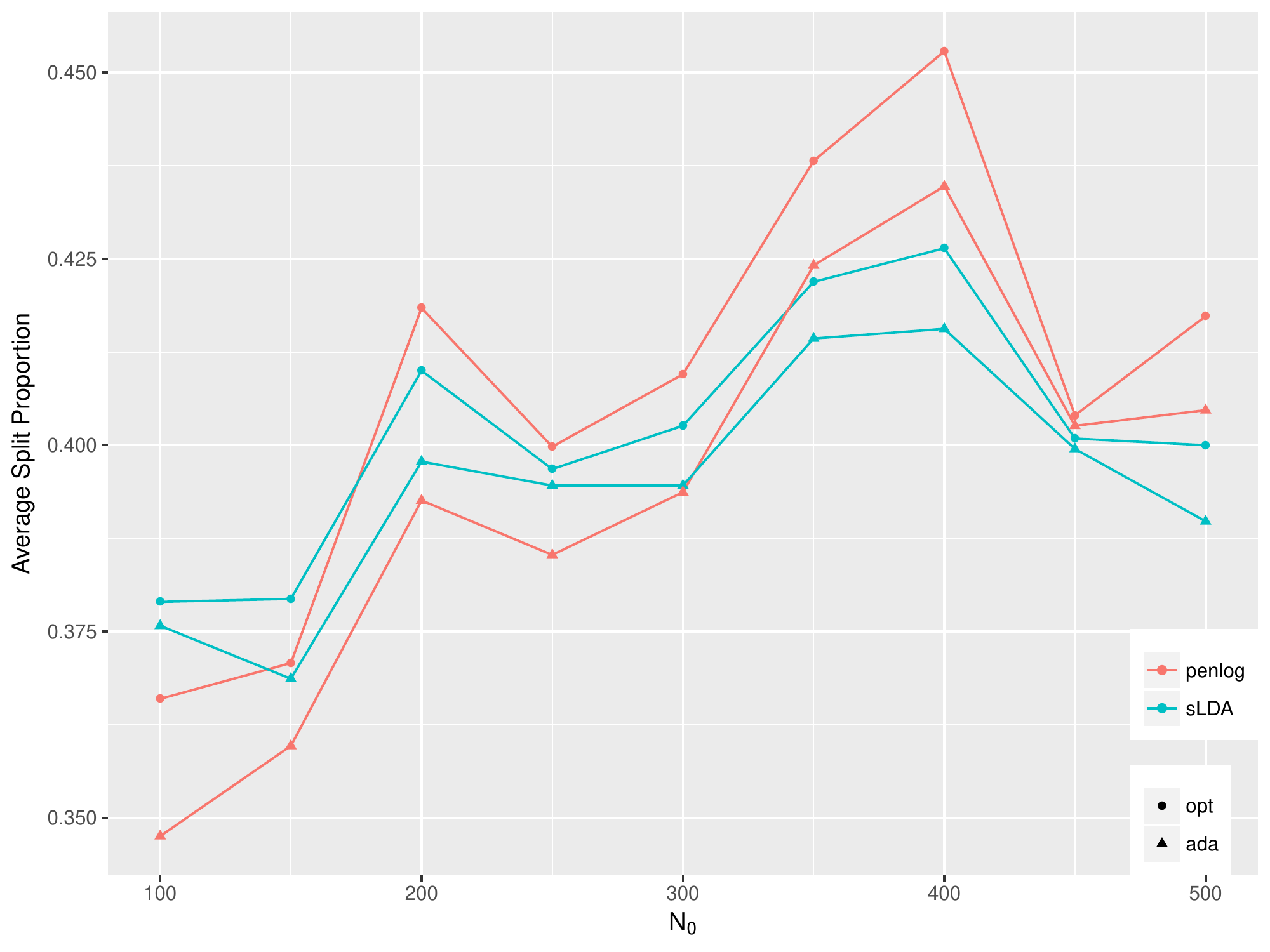}
    
\end{center}

\end{figure}

\begin{figure}
\caption{Example \ref{ex-7}. Type II error ($\text{Ave}_{b, .5}$ and $\text{Ave}_{b, \hat \tau}$) vs. sample size ratios  for four NP classifiers (\texttt{NP-sLDA}, \texttt{NP-penlog}, \texttt{NP-svm}, \texttt{NP-randomforest}), with both multiple random splits ($M=11$) and single random split. $N_0=100$ for all sample size ratios. \label{fig::multiple_adaptive}}
\begin{center}

\includegraphics[scale=0.7]{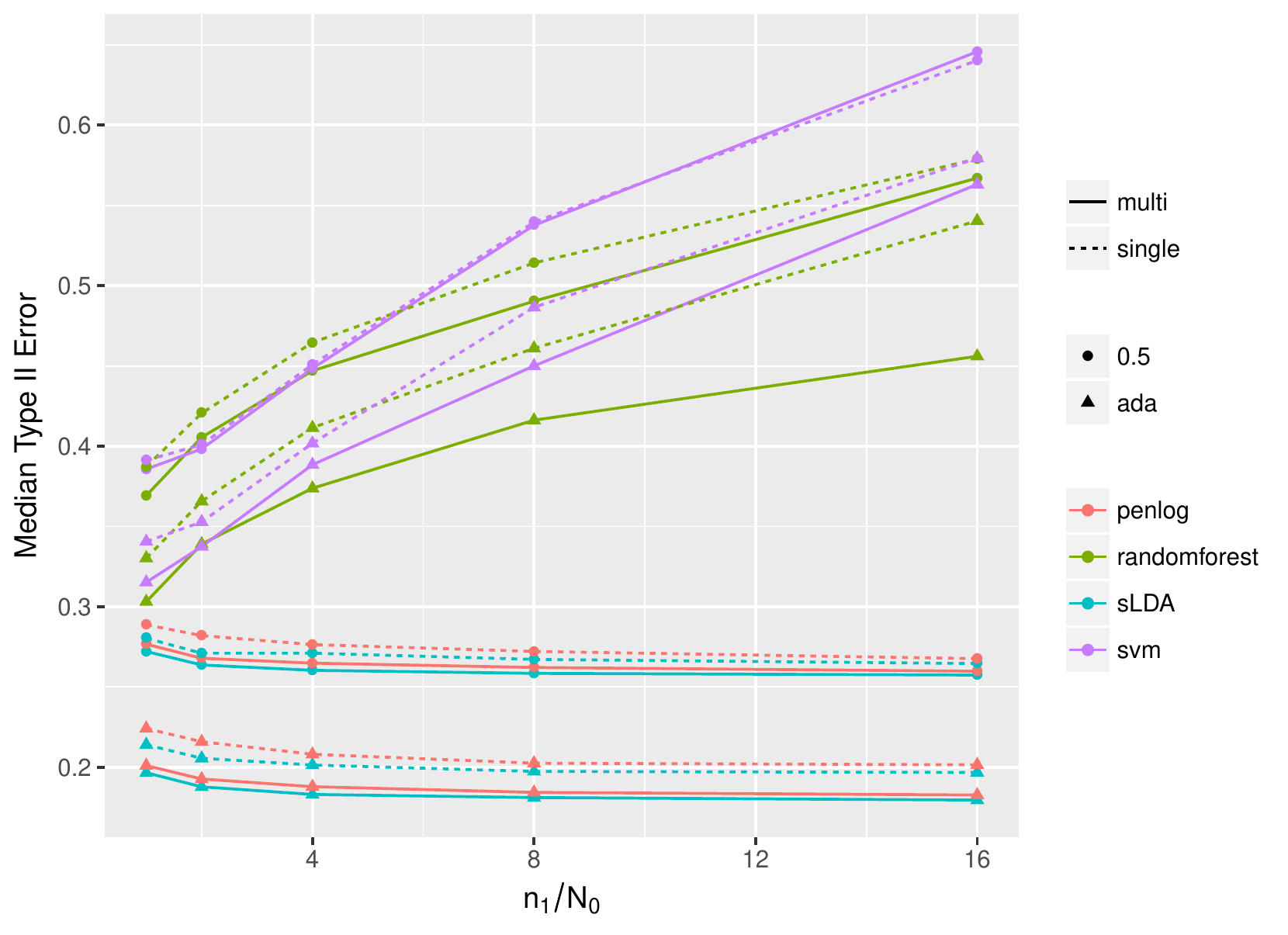}
    
\end{center}
    
\end{figure}

\begin{table}[ht]
\caption{Example \ref{ex-7}. Average computational cost (in seconds) for four NP-classifiers (\texttt{NP-sLDA}, \texttt{NP-penlog}, \texttt{NP-randomforest}, \texttt{NP-svm}) over 1,000 repetitions (standard deviation in parentheses). \label{tb::ex7:computation_cost_comparison}}
\begin{center}
\begin{tabular}{lrrrr}
\hline
$n_1/N_0$&\texttt{NP-sLDA}&\texttt{NP-penlog}&\texttt{NP-randomforest}&\texttt{NP-svm}\\
\hline
1&1.17(0.14)&3.58(0.58)&1.49(4.83)&33.81(2.08)\\
2&1.19(0.13)&5.24(1.03)&1.43(0.16)&36.09(1.98)\\
4&1.19(0.14)&7.63(1.79)&2.06(0.10)&41.44(2.38)\\
8&1.22(0.09)&11.11(2.31)&3.43(0.19)&53.25(6.21)\\
16&1.30(0.08)&16.08(3.79)&6.92(0.25)&84.64(5.28)\\
\hline
\end{tabular}
\end{center}
\end{table}

With Example \ref{ex-7}, we investigate the interaction between adaptive splitting strategy and multiple random splits on different  NP classifiers.  Multiple random splits of class $0$ observations were proposed in the NP umbrella algorithm in \cite{tong2016np} to increase the stability of the type II error performance.  When an NP classifier uses $M>1$ multiple splits, each split will result in a classifier, and the final prediction rule is a majority vote of these classifiers.   Figure \ref{fig::multiple_adaptive} shows the trend of type II error of \verb+NP-sLDA+, \verb+NP-penlog+, \verb+NP-randomforest+, and \verb+NP-svm+,  as the sample size ratio $n_1/N_0$  increases from $1$ to $16$ while keeping $N_0=100$. For each base algorithm, four scenarios are considered: (fixed 0.5 split proportion, single split), (adaptive split proportion, single split), (fixed 0.5 split proportion, multiple splits), and (adaptive split proportion, multiple splits).  Figure \ref{fig::multiple_adaptive} suggests the following interesting findings: i). type II error decreases for \verb+NP-sLDA+ and \verb+NP-penlog+ but increases for \texttt{NP-randomforest} and \texttt{NP-svm}, as a function of $n_1/N_0$ while keeping $N_0$ constant; ii). with both fixed $0.5$ split proportion and adaptive splitting strategy, performing multiple splits leads to a smaller type II error compared with their single split counterparts; iii). for both single split and multiple splits, the adaptive split always improves upon the fixed $0.5$ split proportion;  iv). \verb+NP-svm+ and \verb+NP-randomforest+ are affected by the imbalance scenario, and  one might consider downsampling or upsampling methods before applying an NP algorithm; and v). adding multiple splits to the adaptive splitting strategy leads to a further reduction on the type II error. Nevertheless, the reduction in type II error from the adaptive splitting scheme alone is much larger than the marginal gain from adding multiple splits on top of it. Therefore, when computation power is limited, one should implement the adaptive splitting scheme before considering multiple splits.  

 Lastly, from Table \ref{tb::ex7:computation_cost_comparison}, we would like to point out \verb+NP-sLDA+ is the fastest method to compute among the four NP classifiers with more evident advantages as the sample size increases. \footnote{All numerical experiments were performed on HP Enterprise XL170r with CPU E5-2650v4 (2.20 GHz) and 16 GB memory. }

\subsection{Real data analysis}
We study two high-dimensional datasets in this subsection.  
The first is 
a neuroblastoma dataset containing $d=43,827$ gene expression measurements
from $N=498$ neuroblastoma samples generated by the
Sequencing Quality Control (SEQC) consortium \citep{wang2014concordance}. The
samples fall into two classes: $176$ high-risk (HR) samples and $322$ non-HR
samples. It is usually understood that misclassifying an HR sample
as non-HR will have more severe consequences than the other way
around. Formulating this problem under the NP classification
framework, we label the HR samples as class $0$ observations  and the non-HR
samples as class $1$ observations  and, use all gene expression measurements
as features to perform classification. We set $\alpha=\delta_0=0.1$, and compare \verb+NP-sLDA+ with \verb+NP-penlog+, \verb+NP-randomforest+ and \verb+NP-svm+. We randomly split the dataset $1,000$ times into a training set ($70\%$) and a test set ($30\%$), and then train the NP classifiers on each training data and compute their empirical type I and type II errors over the corresponding test data.  We consider each fixed split proportion in $\{.1, .2, .3, .4, .5, .6, .7, .8\}$ 
 as well as  the adaptive splitting strategy. Here, the split proportion $0.9$ is not considered since it leads to a left-out sample size which is too small to control the type I error at the given $\alpha$ and $\delta_0$ values. Figure \ref{fig::SEQC} indicates that the average type I error is less than $\alpha$ across different split proportions for all four methods considered. Regarding the average type II error, \verb+NP-sLDA+ has the smallest values for a wide range of split proportions.  In particular, the smallest average type II error for \verb+NP-sLDA+ corresponds to split proportion $0.4$. The average location of the split proportion chosen by the adaptive splitting scheme would lead to a type II error close to the minimum. This demonstrates that the adaptive splitting scheme works well for different NP classifiers.  Lastly, we note that for a specific splitting proportion, the median computation time is 213.95 seconds for \verb+pNP-sLDA+ vs. 1717.69 seconds for \verb+NP-randomforest+ over 1,000 random splits.
\begin{figure}
\caption{The average type I and type II errors vs. splitting proportion on the neuroblastoma data set for \texttt{NP-sLDA}, \texttt{NP-penlog}, \texttt{NP-randomforest} and \texttt{NP-svm} over $1,000$ random splits.     The ``*" point on each line represents the average split proportion chosen by adapting splitting. \label{fig::SEQC}}
\begin{center}
\includegraphics[scale=0.39]{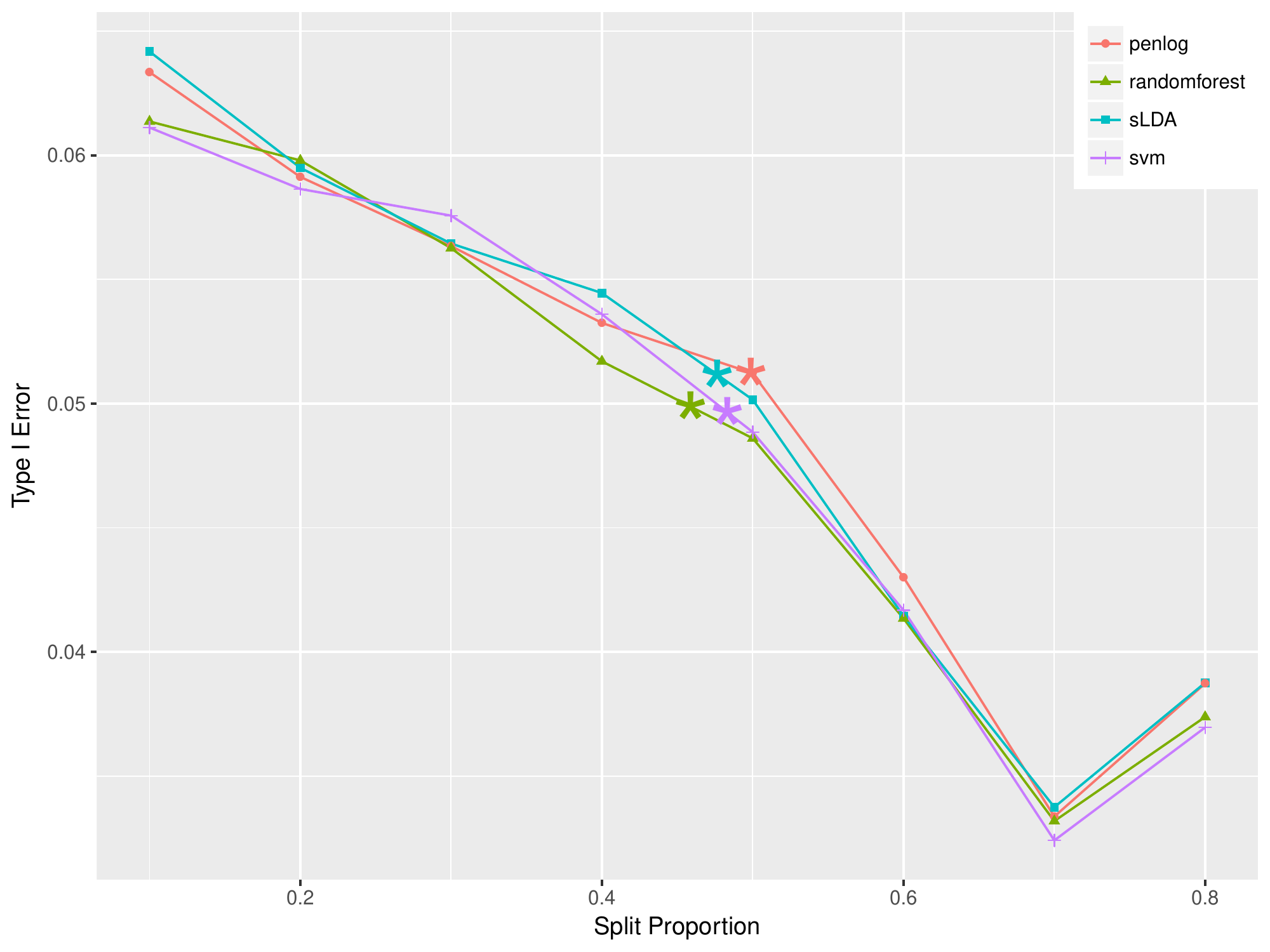}
\includegraphics[scale=0.39]{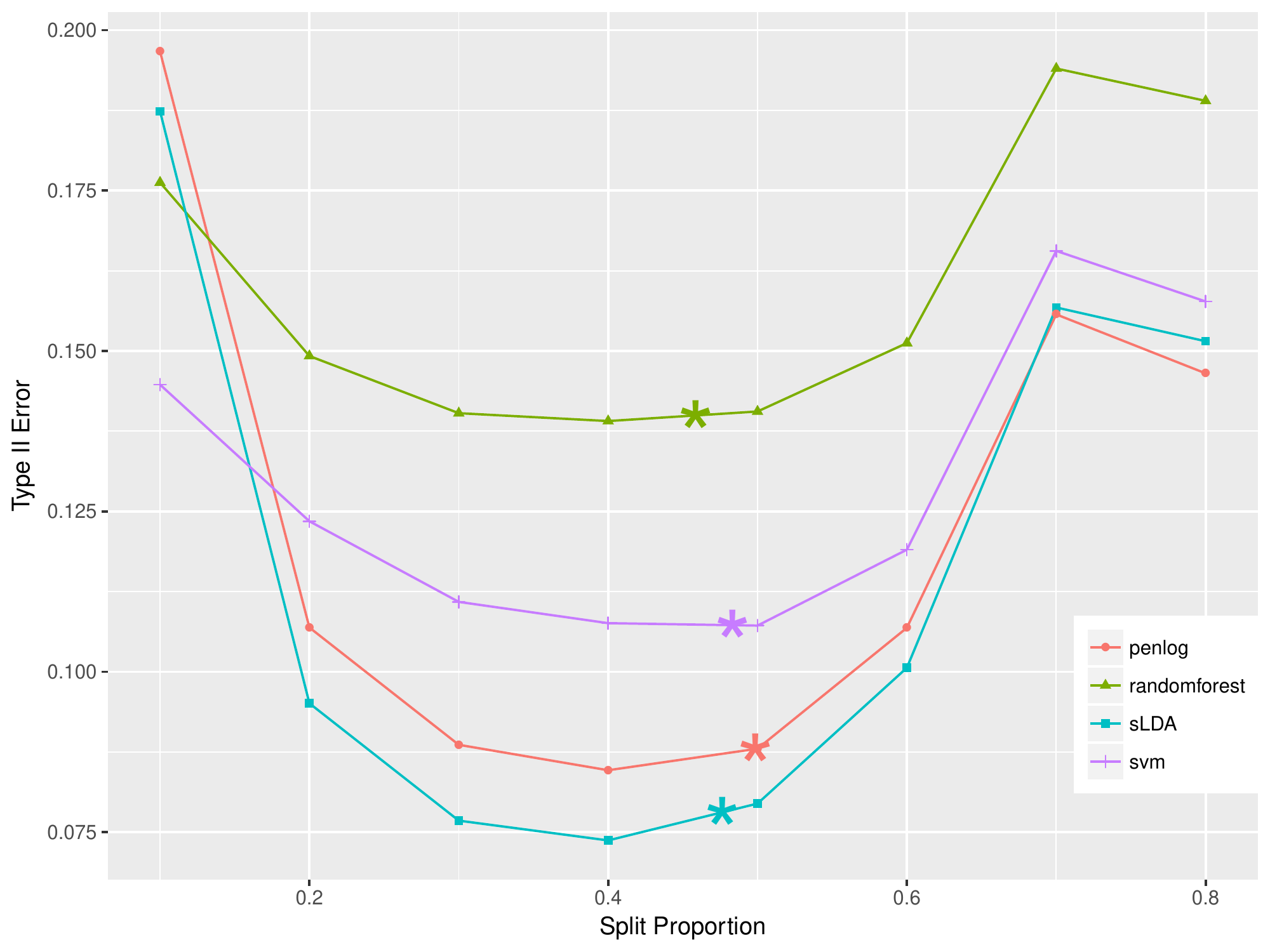}
\end{center}
\end{figure}

The second dataset is a high-dimensional breast cancer dataset ($d=22,215, N=118$)  \citep{Chin:2006bt} with gene expression measurements of subjects that fall into positive (class 0, $N=75$) and negative (class 1, $n_1=43$) groups.  We set $\delta_0=0.1$ and vary $\alpha$ from $0.05$ to $0.15$, and compare the performance of \verb+NP-sLDA+ with \verb+pNP-sLDA+.  We randomly split the dataset 1,000 times into a training set  ($2/3$) and a test set ($1/3$), train the two methods on the training set,  and  compute  the empirical type I and type II errors on the corresponding test set. Due to the limited sample size, it is clear from Table \ref{tb::breast:cancer} that when $\alpha$ is small, the minimum sample size requirement for the NP umbrella algorithm is not satisfied. 
%textcolor{red}{we probably talked about this before, when the sample size requirement of the umbrella algorithm is not met, how can the algorithm run and get a result?}. 
The \verb+pNP-sLDA+, on the other hand, took advantage of the parametric assumption and the corresponding violation rate is under $\delta_0$ throughout all choices of $\alpha$. 
%\textcolor{red}{Also, the type II error of NP-sLDA is a little strange, if the type I error violation rate is 1, the type II error should be small.  If I remember correctly from our previous discussion, these 1's should be replaced by N/A?}. \textcolor{blue}{Another thing is that how about just varying alpha from 0.05 to 0.08 where umbrella algorithm all fails?}
%
%\textcolor{magenta}{Changed the elements to NA, which is the current algorithm output. We can talk about this further. When the sample size requirement of the umbrella algorithm is not met, another way is to let the algorithm always classify to class 0, resulting type I 0 and type II error 1. Ranging from 0.05 to 0.08 is a bit strange. How about 0.05 to 0.1?}

\begin{table}[ht]
\caption{Violation rate and average type II errors over 1,000 replications for the breast cancer data. Here, $\delta_0=0.1$ and $\alpha$ varies from 0.05 to 0.15.   \label{tb::breast:cancer}}
\begin{center}
\begin{tabular}{lrrrr}
\hline
$\alpha$&vio (NP)&vio (pNP)&type II (NP)&type II (pNP)\\
\hline
0.05&NA&0.021&NA&0.663\\
0.06&NA&0.025&NA&0.626\\
0.07&NA&0.029&NA&0.594\\
0.08&NA&0.015&NA&0.565\\
0.09&0.111&0.015&0.448&0.536\\
0.1&0.111&0.015&0.448&0.515\\
%0.11&0.111&0.017&0.448&0.493\\
%0.12&0.053&0.012&0.448&0.474\\
%0.13&0.053&0.012&0.448&0.457\\
%0.14&0.053&0.013&0.448&0.442\\
%0.15&0.183&0.014&0.272&0.429\\
\hline
\end{tabular}
\end{center}
\end{table}

\section{Discussion}
%\textcolor{red}{Rewrite this section to respond to abstract.}
%\textcolor{red}{rewrite this part after the introduction; should not focus on the theoretical contribution, but rather a different dimension.}
This work develops Neyman-Pearson (NP) classification theory and methodology under parametric model assumptions. Most specifically, based on the linear discriminant analysis (LDA) model, we develop a new parametric model-based thresholding rule for high probability type I error control, and this complements the nonparametric NP umbrella algorithm when the minimum sample size requirement of the latter is not met. In practice, when the minimum sample size requirement is met and the scoring function depends on more than a few features, the NP umbrella algorithm is still recommended based on better empirical performance.   For future work, it would be interesting to investigate NP classifiers under other parametric  settings, such as quadratic discriminant analysis (QDA) model and heavy-tailed distributions which are appropriate to model financial data. We expect that new model-specific thresholding rules will be developed for NP classification.  %The theoretical assumptions set up in this work, although minted for the LDA model, shed light on the development of NP theory for other parametric models.     
%\section*{Supplementary Material}
%The supplementary materials contain additional technical lemmas, a proposition, and  proofs. 
%\newpage 
\section*{Acknowledgement}
The authors would like to thank the Action Editor and three anonymous referees for many constructive comments which greatly improved the paper. This work was partially supported by National Science Foundation grants DMS-1554804 and  DMS-1613338, and National Institutes of Health grant R01 GM120507.%

\appendix 

%\begin{center}
%{\huge\center{Supplementary Materials for ``Neyman-Pearson classification: parametrics and power enhancement"}}
%\end{center}
%\section{Supplementary Material}
%The supplementary materials contain additional technical lemmas, a proposition, and  proofs. 

%\section{Appendix}
%The appendix contains additional technical lemmas, a proposition, and  proofs.  

\section{Neyman-Pearson Lemma} \label{NP lemma}
The oracle classifier under the NP paradigm arises from its close connection to the Neyman-Pearson Lemma in statistical hypothesis testing. 
Hypothesis testing bears strong resemblance to binary classification if we assume the following  model. Let $P_1$ and $P_0$ be two \textit{known} probability distributions on $\mathcal{X}\subset \mathbb{R}^d$. 
Assume that $Y\sim \text{Bern}(\zeta)$ for some $\zeta \in (0,1)$, and the conditional distribution of $X$ given $Y$ is $P_Y$.
Given such a model, the goal of statistical hypothesis testing is to determine if we should reject the null hypothesis that $X$ was generated from $P_0$. 
To this end, we construct a randomized test $\phi:\mathcal{X} \to [0,1]$ that rejects the null with probability $\phi(X)$.  
Two types of errors arise: type~I error occurs when $P_0$ is rejected yet $X\sim P_0$, and type~II error occurs when $P_0$ is not rejected yet $X\sim P_1$. 
The Neyman-Pearson paradigm in hypothesis testing amounts to choosing $\phi$ that solves the following constrained optimization problem
$$
\text{maximize } \E[\phi(X)|Y=1]\,,
\text{ subject to }  \E[\phi(X)|Y= 0 ]\leq\alpha\,,
$$
where $\alpha \in (0,1)$ is the significance level of the test. A solution to this constrained optimization problem is called  \emph{a most powerful test} of level $\alpha$. The Neyman-Pearson Lemma gives mild sufficient conditions for the existence of such a test.

\begin{lemma}[Neyman-Pearson Lemma]\label{lemma:NP}
Let $P_1$ and $P_0$ be two probability measures with densities $f_1
$ and $f_0$ respectively, and denote the density ratio as $r(x)=f_1(x)/f_0(x)$.
For a given significance level $\alpha$, let $C_{\alpha}$ be such that
$P_0\{r(X)>C_{\alpha}\}\leq\alpha$ and $P_0\{r(X)\geq C_{\alpha}\}\geq\alpha$.  
Then,
the most powerful test of level $\alpha$ is
\begin{equation*}
\phi^*_{\alpha}(X)=\left\{
 \begin{array}{ll}
     1 & \text{if $\,\,r(X)>C_{\alpha}$}\,,\\
     0 & \text{if $\,\,r(X)<C_{\alpha}$}\,,\\
     \frac{\alpha-P_0\{r(X)>C_{\alpha}\}}{P_0\{r(X)=C_{\alpha}\}} & \text{if $\,\,r(X)=C_{\alpha}$}\,.
   \end{array}      \right.
\end{equation*}
\end{lemma}
Under mild continuity assumption, we take the \emph{NP oracle classifier} 
\begin{align}\label{eq::oracle}
\phi^*_{\alpha}(x) \,=\, \1\{f_1(x)/f_0(x) > C_{\alpha}\} \,=\, \1\{r(x) > C_{\alpha}\}\,,
\end{align}
as our plug-in target for NP classification.

\section{Additional Lemmas and Propositions}

\begin{lemma}[\cite{Hsu.Kakade.Zhang.2012}]\label{lem: concentration}
Let $A\in R^{m\times n}$ be a matrix, and let $\Sigma := A^{\top} A$. Let $x = (x_1, \cdots, x_n)^{\top}$ be an isotropic multivariate Gaussian random vector with mean zero. For all $t>0$, 
$$
\p\left(\|Ax\|^2 > tr(\Sigma) + 2 \sqrt{tr(\Sigma^2)t} + 2\|\Sigma\|t\right) \leq e^{-t}\,.
$$
\end{lemma}

\begin{lemma} \label{lem:beta}
Recall that $\beta^{\text{Bayes}} = \Sigma^{-1}\mu_d = \Sigma^{-1} (\mu^1 - \mu^0)$ and $A = \{j: \{ \Sigma^{-1}\mu_d \}_j\neq 0 \}$. Denote by  $\beta^* = (\Sigma_{AA})^{-1}(\mu^1_{A}-\mu^0_{A})$\, and $\widetilde{\beta}^{\text{Bayes}}$ by letting $\widetilde{\beta}^{\text{Bayes}}_A = \beta^{*}$ and $\widetilde{\beta}^{\text{Bayes}}_{A^c} = 0$.  Then $\widetilde{\beta}^{\text{Bayes}} = \beta^{\text{Bayes}}$.
\end{lemma}

Recall these notations for the following lemma: let $\mathcal{S}_0 = \{x^0_{1}, \cdots, x^0_{n_0}\}$ be an i.i.d. sample of class $0$ of size  $n_0$ and $\mathcal{S}_1 = \{x^1_{1}, \cdots, x^1_{n_1}\}$ be an i.i.d. sample of class $1$ of size  $n_1$, and $n = n_0 + n_1$.  We use $\mathcal{S}_0$ and $\mathcal{S}_1$ to find an estimate of $\beta^{\text{Bayes}}$.   Let $\widetilde{X}$ be the $(n\times d)$ centered predictor matrix, whose column-wise mean is zero, which can be decomposed into $\widetilde{X}^0$, the $(n_0\times d)$ centered predictor matrix based on class $0$ observations and $\widetilde{X}^1$,  the $(n_1\times d)$ centered predictor matrix based on class $1$ observations.  Let $C^{(n)} = (\widetilde{X})^{\top}\widetilde{X} /n$, then $$C^{(n)} = \frac{n_0}{n}\widehat{\Sigma}^0 + \frac{n_1}{n}\widehat{\Sigma}^1\,,$$
where
$\widehat{\Sigma}^0 = (\widetilde{X}^{0})^{T} \widetilde{X}^{0} / n_0\,, \text{ and }\widehat{\Sigma}^1 = (\widetilde{X}^{1})^{T} \widetilde{X}^{1} / n_1\,.$
\begin{lemma}\label{lem:A1}
Suppose there exists $c > 0$ such that $\Sigma_{jj}\leq c$ for all $j = 1, \cdots, d$.  
There exist constants $\varepsilon_0$ and $c_1$, $c_2$ such that for any $\varepsilon \leq \varepsilon_0$ we have, %\textcolor{red}{This proof can be omitted.  Yang, please check the inequality (2) proof, and why we need an upper bound for epsilon in this inequality.}

\begin{equation}\label{eqn:A1}
\p\left(|(\widehat{\mu}^1_{j} - \widehat{\mu}^0_{j})-(\mu_{j}^1 - \mu_{j}^0)|\geq \varepsilon)\right) \leq 2\exp(-n_0 \varepsilon^2 c_2)  + 2\exp(-n_1 \varepsilon^2 c_2)\,, \text{ for } j = 1, \cdots, d\,. 
\end{equation}

\begin{equation} \label{eqn:A2}
\p\left( |\widehat{\Sigma}^l_{ij} - \Sigma_{ij}| \geq \varepsilon \right) \leq 2 \exp (-n_l \varepsilon^2 c_1)\,, \text{  for } l = 0, 1, \text{ }i,j = 1, \cdots, d\,.
\end{equation}

\begin{equation}\label{eqn:A3}
\p\left(\|\widehat{\Sigma}^l_{AA} - \Sigma_{AA}\|_{\infty}\geq \varepsilon\right)\leq 2 s^2 \exp(-n_l s^{-2} \varepsilon^2 c_1)\,.
\end{equation}

\begin{equation}\label{eqn:A4}
\p\left(\|\widehat{\Sigma}_{A^cA} - \Sigma_{A^cA}\|_{\infty}\geq \varepsilon\right)\leq (d-s)s\exp(-n_l s^{-2} \varepsilon^2 c_1)\,.
\end{equation}

\begin{equation}\label{eqn:A5}
\p\left(\| (\widehat{\mu}^1 -\widehat{\mu}^0) - (\mu^1 - \mu^0) \|_{\infty}\geq \varepsilon  \right) \leq 2d\exp(-n_0 \varepsilon^2 c_2)  + 2d\exp(-n_1 \varepsilon^2 c_2)\,. 
\end{equation}

\begin{equation}\label{eqn:A6}
\p\left(\| (\widehat{\mu}^1_A -\widehat{\mu}^0_A) - (\mu^1_A - \mu^0_A) \|_{\infty}\geq \varepsilon  \right) \leq 2s\exp(-n_0 \varepsilon^2 c_2)  + 2s\exp(-n_1 \varepsilon^2 c_2)\,. 
\end{equation}

\begin{equation}\label{eqn:A7}
\p\left(|C^{(n)}_{ij} -  \Sigma_{ij}|\geq \varepsilon\right)
 \leq  2 \exp\left( - \frac{c_1 \varepsilon^2 n^2}{4 n_0}\right) + 2 \exp\left( - \frac{c_1 \varepsilon^2 n^2}{4 n_1}\right)\,.
\end{equation}

\begin{equation}\label{eqn:A8}
\p\left(|C^{(n)}_{AA} - \Sigma_{AA}|\geq \varepsilon\right)
\leq  2s^2 \exp\left(- \frac{c_1 \varepsilon^2 n^2}{4  s^2n_0}  \right) + 2s^2 \exp\left(- \frac{c_1 \varepsilon^2 n^2}{4  s^2n_1}  \right)\,.
\end{equation}

\begin{equation}\label{eqn:A9}
\p\left(|C^{(n)}_{A^cA} - \Sigma_{A^c A}|\geq \varepsilon\right)
\leq  (d-s) s \exp\left( -\frac{c_1  \varepsilon^2n^2}{4s^2n_0} \right)  + (d-s) s \exp\left( -\frac{c_1 \varepsilon^2n^2 }{4s^2n_1} \right)\,.
\end{equation}

\end{lemma}

\begin{lemma}\label{lem:prod_deviation}
Recall that $\kappa = \| \Sigma_{A^c A} (\Sigma_{AA})^{-1}\|_{\infty}$, $\varphi = \|(\Sigma_{AA})^{-1}\|_{\infty}$ and $\Delta = \|\mu^1_A - \mu^0_A\|_{\infty}$. Let $C^{(n)}_{A^cA} = \frac{n_0}{n} (\widetilde X^{0}_{A^c})^{\top}  \widetilde X^{0}_{A} + \frac{n_1}{n} (\widetilde X^{1}_{A^c})^{\top}  \widetilde X^{1}_{A} = \frac{n_0}{n} \widehat{\Sigma}^0_{A^cA} + \frac{n_1}{n} \widehat{\Sigma}^1_{A^cA}$, and    $C^{(n)}_{AA} = \frac{n_0}{n}\widehat{\Sigma}^0_{AA} + \frac{n_1}{n}\widehat{\Sigma}^1_{AA}$. There exist constants $c_1$ and $\varepsilon_0$ such that for any $\varepsilon \leq \min (\varepsilon_0, 1/\varphi)$, we have
$$
\p\left(\|C^{(n)}_{A^cA}(C^{(n)}_{AA})^{-1} - \Sigma_{A^cA}(\Sigma_{AA})^{-1}\|_{\infty}\geq (\kappa + 1)\varepsilon \varphi (1-\varphi \varepsilon)^{-1}\right)  \leq f(d, s, n_0, n_1, \varepsilon)\,,
$$
where $f(d, s, n_0, n_1, \varepsilon) = (d+s) s \exp\left( -\frac{c_1  \varepsilon^2n^2}{4s^2n_0} \right)  + (d+s) s \exp\left( -\frac{c_1 \varepsilon^2n^2 }{4s^2n_1} \right)$, and $n = n_0 + n_1$. 

%We would like a high probability upper bound for $\|C^{(n)}_{A^cA}(C^{(n)}_{AA})^{-1} - \Sigma_{A^cA}(\Sigma_{AA})^{-1}\|_{\infty}$
\end{lemma}

\begin{lemma}\label{lemma-sub1}
Let $\widetilde{\mathcal{C}}^0$ be in Equation \eqref{eqn:C0tilde}, $\mathcal{C}$ as in Lemma \ref{lem:large_prob_set}, and $\widetilde{X}_A = \Sigma_{AA}^{-1/2} X_A$. Assume $\lambda_m=\lambda_{\min}(\Sigma^{-1/2}_{AA})$ is bounded from below,  then we have
\[
P_0\{C^{**}_\alpha\leq s^*(X)\leq C^{**}_\alpha+\delta | X\in \mathcal{C}\}\geq (1-\delta_3)P_0(C^{**}_\alpha\leq (\mu^1_A-\mu^0_A)^{\top}\Sigma^{-1/2}_{AA}\widetilde{X}_A\leq C^{**}_\alpha+\delta |\widetilde{\mathcal{C}}^0)\,,
\]
where $\delta_3 = \exp\{-(n_0 \wedge n_1)^{1/2}\}$.  
\end{lemma}

\begin{lemma}\label{lemma-sub2}
Let us denote $a=\Sigma^{-1/2}_{AA}(\mu^1_A-\mu^0_A)$. Assume there exist $M>0$ such that the following conditions hold:
\begin{itemize}
\item[i)] $C^{**}_\alpha - (\mu^1_A-\mu^0_A)^{\top}\Sigma^{-1}_{AA}\mu^0_A\in (C^1, C^2)$ for some constants $C^1$, $C^2$.  
\item[ii)] When $s=1$, $a$ is a scalar. $f_{\mathcal{N}(0,|a|)}$ is bounded below on interval $(C^1, C^2+\delta^*)$ by $M$.
\item[iii)] When $s=2$, $a=(a_1, a_2)^{\top}$ is a vector. 
\[\left(\frac{1}{\sqrt{2\pi}|a_1|\big(\frac{a^2_2}{a^2_1}+1\big)}\exp\{-\frac{a^2_1+a^2_2-a^2_2 a^2_1}{a^2_1 (a^2_1+a^2_2)} t^2\}\right) \left(2\Phi(\sqrt{\widetilde{L}^2-\frac{a^2_1+a^2_2-a^2_2 a^2_1}{a^2_1 (a^2_1+a^2_2)} t^2})-1\right)\]
 is bounded below on interval $t\in(C^1, C^2+\delta^*)$ by $M$.
\end{itemize}
Then, for $s\leq 2$, for any $\delta\in(0,\delta^*)$, there exists $M_1$ which is a constant depending on $M$, such that the following inequality holds
\[
P_0(C^{**}_\alpha\leq (\mu^1_A-\mu^0_A)^{\top}\Sigma^{-1/2}_{AA}\widetilde{X}_A\leq C^{**}_\alpha+\delta |\widetilde{\mathcal{C}}^0)\geq M_1\delta.
\]
\end{lemma}

\begin{proposition}\label{prop:conddetec}
Suppose that $\lambda_{\min}(\Sigma^{-1/2}_{AA})$, the minimum eigenvalue of $\Sigma^{-1/2}_{AA}$, is bounded from  below. Let us denote $a=\Sigma^{-1/2}_{AA}(\mu^1_A-\mu^0_A)$. Let us also assume that there exists $M>0$ such that the following conditions hold:
\begin{itemize}
\item[i)] $C^{**}_\alpha - (\mu^1_A-\mu^0_A)^{\top}\Sigma^{-1}_{AA}\mu^0_A\in (C^1, C^2)$ for some constants $C^1$, $C^2$.  
\item[ii)] When $s=1$, $a$ is a scalar. $f_{\mathcal{N}(0,|a|)}$ is bounded below on interval $(C^1-\delta^*, C^2+\delta^*)$ by $M$.
\item[iii)] Let $\widetilde{L}=\lambda_{\min}(\Sigma^{-1/2}_{AA})c'_1 s^{1/2}(n_0\wedge n_1)^{1/4}$. When $s=2$, $a=(a_1, a_2)$ is a vector. 
\[\left(\frac{1}{\sqrt{2\pi}|a_1|\big(\frac{a^2_2}{a^2_1}+1\big)}\exp\{-\frac{a^2_1+a^2_2-a^2_2 a^2_1}{a^2_1 (a^2_1+a^2_2)} t^2\}\right) \left(2\Phi(\sqrt{\widetilde{L}^2-\frac{a^2_1+a^2_2-a^2_2 a^2_1}{a^2_1 (a^2_1+a^2_2)} t^2})-1\right)\]
 is bounded below on interval $t\in(C^1-\delta^*, C^2+\delta^*)$ by $M$.
\end{itemize}
Then for $s\leq 2$, the function $s^*(\cdot)$ satisfies conditional detection condition restricted to $\mathcal{C}$ of order $\uderbar\gamma= 1$ with respect to $P_0$ at the level $(C^{**}_{\alpha}, \delta^*)$.  In other words, Assumption \ref{assumption:2} is satisfied.
\end{proposition}

\section{Simulation for the control of largest empirical eigenvalues}\label{sec::eigen:bound:simu}
Here, we present a simulation to verify the inequality on the largest eigenvalue of the sample covariance matrix in relation to that of the population covariance matrix, as stated in \eqref{eqn:max_eigen}. 
\begin{example}\label{ex::eigen:bound}
The data are generated from an LDA model with common covariance matrix $\Sigma$. We set $\epsilon = 1e-3$ and  consider the 8 combinations for the following three factors and vary $N_0=n_1\in \{20, 40, 60, 80, 100, 120, 140, 160, 180, 200\}$. $\mu^0=0_d$. 
\begin{enumerate}[(\ref{ex::eigen:bound}a).]
 \item The dimension $d\in\{3, 10\}$. $\mu^1=1.16 \cdot  1_3$ when $d=3$ and $\mu^1=0.75\cdot 1_{10}$ when $d=10$. 
	\item $\Sigma$ is set to be an AR(1) covariance matrix with $\Sigma_{ij}=\rho^{|i-j|}$ or an CS (compound symmetry) covariance matrix with $\Sigma_{ii}=1$ and $\Sigma_{ij}=\rho$ for all $i\neq j$. 
	\item The correlation parameter $\rho\in\{.5, .9\}$.
\end{enumerate}
%The true $\beta^{\text{Bayes}}=\Sigma^{-1}\mu_d = 1.2\times(1_{d_0},1_{d-d_0})^{\top}$, $\mu^0=0_d$, $d_0 = 3$. We set $\pi_0 = \pi_1 = 0.5$. $\alpha=\delta_0=0.1$. 
\end{example}
From Table \ref{tb::eigen:bound}, it is clear that the bound is satisfied with very high probability across all scenarios considered. 
\begin{table}[ht]
\caption{The probability that the bound stated in \eqref{eqn:max_eigen} is satisfied over 1000 replications.\label{tb::eigen:bound}}
\begin{center}
\begin{tabular}{l|c|c|rrrrrrrrrr}
\hline
&Covariance&\diagbox{$\rho$}{$N_0$}&20&40&60&80&100&120&140&160&180&200\\
\hline
\multirow{4}{*}{$d=3$}&\multirow{2}{*}{AR(1)}&.5&1.0&1.0&1.0&1.0&1.0&.999&.999&1.0&.999&1.0\\
&&.9&1.0&1.0&1.0&1.0&1.0&.999&.999&1.0&.999&1.0\\
\cline{2-13}
&\multirow{2}{*}{CS}&.5&1.0&1.0&1.0&1.0&1.0&1.0&1.0&1.0&1.0&1.0\\
&&.9&1.0&1.0&1.0&1.0&1.0&1.0&1.0&1.0&1.0&1.0\\\hline
\multirow{4}{*}{$d=10$}&\multirow{2}{*}{AR(1)}&.5&1.0&1.0&1.0&1.0&1.0&1.0&1.0&1.0&1.0&1.0\\
&&.9&1.0&1.0&1.0&1.0&1.0&1.0&1.0&1.0&1.0&1.0\\\cline{2-13}
&\multirow{2}{*}{CS}&.5&1.0&1.0&1.0&1.0&1.0&1.0&1.0&1.0&1.0&1.0\\
&&.9&1.0&1.0&1.0&1.0&1.0&1.0&1.0&1.0&1.0&1.0\\
\hline
\end{tabular}
\end{center}
\end{table}
\section{Proofs}

\begin{proof}[Proof of Lemma \ref{lem:2}]
By Corollary \ref{cor:kstarAndkprime}, $k^* \leq k'$.  This implies that $R_0(\hat \phi_{k^*})\geq R_0(\hat \phi_{k'})$. Moreover, by Lemma \ref{prop::R0}, for any $\delta_0' \in (0,1)$ and $n'_0 \geq 4/(\alpha\delta_0)$, 
$$
\p\left(|R_0( \hat{\phi}_{k'} ) 
- R_0( \phi^*_{\alpha}) | > \xi_{\alpha, \delta_0,n_0'}(\delta_0')\right)\leq \delta_0'\,.
$$
Let $\mathcal{E}_0 = \{R_0(\hat \phi_{k^*})\leq \alpha\}$ and $\mathcal{E}_1 = \{|R_0( \hat{\phi}_{k'} ) 
- R_0( \phi^*_{\alpha}) | \leq \xi_{\alpha, \delta_0,n_0'}(\delta_0')\}$.
On the event $\mathcal{E}_0 \cap \mathcal{E}_1$, we have
$$
\alpha = R_0(\phi^*_{\alpha})\geq R_0(\hat \phi_{k^*})\geq R_0(\hat \phi_{k'})\geq R_0(\phi^*_{\alpha}) - \xi_{\alpha, \delta_0,n_0'}(\delta_0')\,,
$$
This implies that
\begin{eqnarray*}
|R_0( \hat{\phi}_{k^*} ) 
- R_0( \phi^*_{\alpha}) | \leq \xi_{\alpha, \delta_0,n_0'}(\delta_0') \,.	
\end{eqnarray*}

\end{proof}

\begin{proof}[Proof of Lemma \ref{lem:large_prob_set}]
Note that $\Sigma^{-1/2}_{AA} (X_A - \mu^0_A) \sim \mathcal{N}(0, I_s)$.  By Lemma \ref{lem: concentration}, for all $t> 0$,
$$
P_0\left(\|X_A - \mu^0_A \|^2  > tr(\Sigma_{AA}) + 2\sqrt{tr(\Sigma^2_{AA})t} + 2 \|\Sigma_{AA}\|t \right)  \leq e ^{-t}\,.
$$
For $t = (n_0 \wedge n_1)^{1/2}$ ($>1$), the above inequality implies there exists some $c_1'' > 0$ such that 
$$
P_0(\|X_A - \mu^0_A \|^2  > c_1'' st     )\leq e^{-t}\,.
$$
Similarly, $P_1(\|X_A - \mu^1_A \|^2  > c_1'' st    )\leq e^{-t}$.
Let $\mathcal{C}^0 = \{X: \|X_A - \mu^0_A \|^2  \leq  c_1'' st     \}$ and $\mathcal{C}^1 = \{X: \|X_A - \mu^1_A \|^2  \leq  c_1'' st     \}$.  There exists some $c_1' > 0$, such that both $\mathcal{C}^0$ and $\mathcal{C}^1$ are subsets of $\mathcal{C} = \{X:\|X_A\|  \leq  c_1' s^{1/2} t^{1/2}  \}$. %Clearly, 
%Hence there exists some $c_1' > 0$ such that 
%$
%P_j(\mathcal{C})\leq e^{-t} \text{ for } j= 0, 1\,.
%$
%
%Define $\mathcal{C}$ by
%
%$$
%\mathcal{C} = \{X\in \mathbb{R}^d: \|X_A\| \leq c_1' s^{1/2}(n_0 \wedge n_1)^{1/4} \}\,.
%$$
Then $P_0(X\in\mathcal{C})\geq 1- \delta_3$ and $P_1(X\in\mathcal{C})\geq 1- \delta_3$, for $\delta_3 = \exp\{-(n_0 \wedge n_1)^{-1/2}\}$. 

%\textcolor{red}{should it be $\delta_3 = \exp\{-(n_0 \wedge n_1)^{1/2}\}$.?}
%\textcolor{red}{we will use $\mathcal{C}^0$ in the conditional detection condition}

By Proposition \ref{thm:1}, for $\delta_1 \geq \delta_1^*$ and $\delta_2 \geq \delta_2^*$, we have with probability at least $1-\delta_1 - \delta_2$, $\hat{\beta}^{\text{lasso}}_{A^c} = \beta^{\text{Bayes}}_{A^c} = 0$. Moreover, 
\begin{eqnarray*}
\|\hat s - s^*\|_{\infty, \mathcal{C}}
%&=& \max_{x\in\mathcal{C}}|x^{\top} \hat{\beta}^{\text{lasso}} - x^{\top} \beta^{\text{Bayes}}|\\
&\leq & \max_{x\in\mathcal{C}}|x^{\top}_A \hat{\beta}^{\text{lasso}}_A - x^{\top}_A \beta^{\text{Bayes}}_A| + \max_{x\in\mathcal{C}}|x^{\top}_{A^c} \hat{\beta}^{\text{lasso}}_{A^c} - x^{\top}_{A^c} \beta^{\text{Bayes}}_{A^c}|\\
&=&  \max_{x\in\mathcal{C}}|x^{\top}_A \hat{\beta}^{\text{lasso}}_A - x^{\top}_A \beta^{\text{Bayes}}_A|\\
&\leq & \|\hat \beta^{\text{lasso}}_A - \beta^{\text{Bayes}}_A\|_{\infty}\cdot \max_{x\in\mathcal{C}}\|X_A\|_1\\
&\leq &  \|\hat \beta^{\text{lasso}}_A - \beta^{\text{Bayes}}_A\|_{\infty}\cdot \sqrt{s} \max_{x\in\mathcal{C}}\|X_A\|_2\\
&\leq & 4 \varphi \lambda \cdot c_1' s (n_0\wedge n_1)^{1/4}\,,
\end{eqnarray*}
where the last inequality uses a relation $\beta^* = \beta^{\text{Bayes}}_A$, which is derived in Lemma \ref{lem:beta}. 
\end{proof}

%\subsubsection{Proof of Lemma \ref{lem:4}}

\begin{proof}[Proof of Lemma \ref{lem:4}]
Note that by Lemma \ref{lem:large_prob_set}, $P_0(X\in\mathcal{C}) \geq 1- \exp\{- (n_0 \wedge n_1)^{1/2}\}$, 
%\adc{Also change to $P_0(X\in \mathcal{C})$?}
so we have
\begin{eqnarray*}
&&|R_0( \hat{\phi}_{k^*} ) 
- R_0( \phi^*_{\alpha}) |\\
&=&   |[R_0( \hat{\phi}_{k^*} |\mathcal{C}) 
- R_0( \phi^*_{\alpha}|\mathcal{C})]P_0(X\in\mathcal{C}) + [R_0( \hat{\phi}_{k^*} |\mathcal{C}^c) 
- R_0( \phi^*_{\alpha}|\mathcal{C}^c)]P_0(X\in\mathcal{C}^c)|\\
&\geq &  |[R_0( \hat{\phi}_{k^*} |\mathcal{C}) 
- R_0( \phi^*_{\alpha}|\mathcal{C})]P_0(X\in\mathcal{C})| - |[R_0( \hat{\phi}_{k^*} |\mathcal{C}^c) 
- R_0( \phi^*_{\alpha}|\mathcal{C}^c)]P_0(X\in\mathcal{C}^c)|\\
&\geq &  |R_0( \hat{\phi}_{k^*} |\mathcal{C}) 
- R_0( \phi^*_{\alpha}|\mathcal{C})|(1- \exp\{- (n_0 \wedge n_1)^{1/2}\}) - 1 \cdot \exp\{- (n_0 \wedge n_1)^{1/2}\}\,.
\end{eqnarray*}

Lemma \ref{lem:2} says that 
\begin{eqnarray*}
\p\{
|R_0( \hat{\phi}_{k^*} ) 
- R_0( \phi^*_{\alpha}) | > \xi_{\alpha, \delta_0,n_0'}(\delta_0') \}
\,\leq\, \delta_0 +  \delta_0'\,.
\end{eqnarray*}
This combined with the above inequality chain implies 
\begin{eqnarray*}
\p\{
|R_0( \hat{\phi}_{k^*} |\mathcal{C}) 
- R_0( \phi^*_{\alpha}|
\mathcal{C}) | > \frac{[\xi_{\alpha, \delta_0,n_0'}(\delta_0') + \exp\{- (n_0 \wedge n_1)^{1/2}\}]}{1 - \exp\{- (n_0 \wedge n_1)^{1/2}\}} \}
\,\leq\, \delta_0 +  \delta_0'\,. 
\end{eqnarray*}
Since $\exp\{- (n_0 \wedge n_1)^{1/2}\}\leq 1/2$, the conclusion follows. 
\end{proof}

%\section{Lemmas related to sLDA}

%\textcolor{blue}{will not use the boldface letters. Just follow Mai's notations}

\begin{proof}[Proof of Lemma \ref{lem:beta}]
Note that $\mu^1 - \mu^0  = \Sigma \beta^{\text{Bayes}}$. After shuffling the $A$ coordinates to the front if necessary, we have
$$
\mu^1 - \mu^0 = \begin{bmatrix}
   \Sigma_{AA} & \Sigma_{A A^c} \\
   \Sigma_{A^cA} & \Sigma_{A^c A^c} 
   \end{bmatrix} \begin{bmatrix} \beta^{\text{Bayes}}_A \\ \beta^{\text{Bayes}}_{A^c}      \end{bmatrix}\,.
$$
Then, 
$\mu^1_{A} - \mu^0_{A} = \left(\Sigma_{AA}\right)\beta^{\text{Bayes}}_A$ as $\beta^{\text{Bayes}}_{A^c} = 0\in \mathbb{R}^{|A^c|}$ by definition.  Therefore we have, $$\beta^* = \Sigma_{AA}^{-1} (\mu^1_{A} - \mu^0_{A}) = \beta^{\text{Bayes}}_A\,,$$
which combined with $\widetilde{\beta}^{\text{Bayes}}_{A^c} = \beta^{\text{Bayes}}_{A^c} = 0$ leads to $\widetilde{\beta}^{\text{Bayes}} = \beta^{\text{Bayes}}$.
\end{proof}

\begin{proof}[Proof of Lemma \ref{lem:A1}]

Inequalities \eqref{eqn:A1}-\eqref{eqn:A6} can be proved similarly as in \cite{mai2012direct}, so proof is omitted.  

Inequalities \eqref{eqn:A7}-\eqref{eqn:A9} can be proved by applying \eqref{eqn:A2}-\eqref{eqn:A4} respectively and observe that $A + B \geq \varepsilon $ implies $A \geq \varepsilon /2$ or $B \geq \varepsilon /2$.  More concretely, they are proven by the following arguments:

\begin{eqnarray*}
\p\left(|C^{(n)}_{ij} -  \Sigma_{ij}|\geq \varepsilon\right)
 &=& \p\left(|\frac{n_0}{n} \widehat{\Sigma}^0_{ij} + \frac{n_1}{n} \widehat{\Sigma}^1_{ij} - \Sigma_{ij}| \geq \varepsilon  \right)\\
 &\leq & \p\left( \frac{n_0}{n}|\widehat{\Sigma}^0_{ij} - \Sigma_{ij}|\geq \varepsilon /2  \right)   + \p\left( \frac{n_1}{n}|\widehat{\Sigma}^1_{ij} - \Sigma_{ij}|\geq \varepsilon /2 \right) \\
 &= & \p\left(|\widehat{\Sigma}^0_{ij} - \Sigma_{ij}|\geq \frac{n\varepsilon}{2 n_0 }  \right) + \p\left(|\widehat{\Sigma}^1_{ij} - \Sigma_{ij}|\geq \frac{n\varepsilon}{2 n_1 }  \right)\\
 & \leq & 2 \exp\left( - \frac{c_1 \varepsilon^2 n^2}{4 n_0}\right) + 2 \exp\left( - \frac{c_1 \varepsilon^2 n^2}{4 n_1}\right)\,.
\end{eqnarray*}

\begin{eqnarray*}
\p\left(|C^{(n)}_{AA} - \Sigma_{AA}|\geq \varepsilon\right)
&=& \p\left(|\frac{n_0}{n}\widehat{\Sigma}^0_{AA} + \frac{n_1}{n}\widehat{\Sigma}^1_{AA}- \Sigma_{AA}|\geq \varepsilon \right)\\
&\leq &  \p\left( \frac{n_0}{n}|\widehat{\Sigma}^0_{AA} - \Sigma_{AA}|\geq \varepsilon /2  \right)   + \p\left( \frac{n_1}{n}|\widehat{\Sigma}^1_{AA} - \Sigma_{AA}|\geq \varepsilon /2 \right) \\
&\leq &  2s^2 \exp\left(- \frac{c_1 \varepsilon^2 n^2}{4 n_0 s^2}  \right) + 2s^2 \exp\left(- \frac{c_1 \varepsilon^2 n^2}{4 n_1 s^2}  \right)\,.
\end{eqnarray*}

\begin{eqnarray*}
\p\left(|C^{(n)}_{A^cA} - \Sigma_{A^c A}|\geq \varepsilon\right)
&=& \p\left(|\frac{n_0}{n}\widehat{\Sigma}^0_{A^cA} + \frac{n_1}{n}\widehat{\Sigma}^1_{A^cA}- \Sigma_{A^c A}|\geq \varepsilon \right)\\
&\leq &  \p\left( \frac{n_0}{n}|\widehat{\Sigma}^0_{A^cA} - \Sigma_{A^cA}|\geq \varepsilon /2  \right)   + \p\left( \frac{n_1}{n}|\widehat{\Sigma}^1_{A^cA} - \Sigma_{A^cA}|\geq \varepsilon /2 \right) \\
&\leq & (d-s) s \exp\left( -\frac{c_1 n^2 \varepsilon^2}{4s^2n_0} \right)  + (d-s) s \exp\left( -\frac{c_1 n^2 \varepsilon^2}{4s^2n_1} \right)\,.
\end{eqnarray*}

\end{proof}

%\textcolor{blue}{Based on the above Lemma, we need to prove the counter part of Lemma A1}

%\begin{lemma}
%There exist constants $c_1$ and $\varepsilon_0$ such that for any $\varepsilon \leq \min (\varepsilon_0, 1/\varphi)$, we have
%$$
%\p\left(\|\widehat{\Sigma}^l_{A^cA}(\widehat{\Sigma}^l_{AA})^{-1} - \Sigma_{A^cA}(\Sigma_{AA})^{-1} \|_{\infty}\geq \varepsilon \varphi (\kappa + 1)(1-\varphi \varepsilon)^{-1}\right)  \leq 2 d s \exp(-n_l s^{-2} \varepsilon^2 c_1)\,.
%$$
%\end{lemma}
%The proof is identical to a similar result in \cite{mai2012direct}, so we omit it here. \textcolor{blue}{In retrospect, this lemma is useless. But the following one will be useful.}  %\textcolor{blue}{From this lemma, we would like to show the Lemma A1 in Mai et al. 2012, though the probability now must be in terms of the two sample sizes, and we still want to use perhaps $\Sigma$ instead of $C$ still}

\begin{proof}[Proof of Lemma \ref{lem:prod_deviation}]
 Let $\eta_1 = \|\Sigma_{AA} - C^{(n)}_{AA}\|_{\infty}$, $\eta_2 = \|\Sigma_{A^cA} - C^{(n)}_{A^cA}\|_{\infty}$, and $\eta_3 = \|(C^{(n)}_{AA})^{-1} - (\Sigma_{AA})^{-1}\|_{\infty}$. 

\begin{eqnarray*}
\|C^{(n)}_{A^cA}(C^{(n)}_{AA})^{-1} - \Sigma_{A^cA}(\Sigma_{AA})^{-1}\|_{\infty}
&\leq &  \|C^{(n)}_{A^cA} - \Sigma_{A^cA}\|_{\infty} \times \|(C^{(n)}_{AA})^{-1} - (\Sigma_{AA})^{-1}\|_{\infty}\\
&&+ \|C^{(n)}_{A^c A} - \Sigma_{A^cA}\|_{\infty} \times \|(\Sigma_{AA})^{-1}\|_{\infty}\\
&& + \| \Sigma_{A^cA} (\Sigma_{AA})^{-1}\|_{\infty} \times \|\Sigma_{AA} - C^{(n)}_{AA}\|_{\infty} \times \|(\Sigma_{AA})^{-1}\|_{\infty} \\
&& + \|\Sigma_{A^cA}(\Sigma_{AA})^{-1}\|_{\infty} \times \|\Sigma_{AA} - C^{(n)}_{AA}\|_{\infty}\\
&&\times \|(C^{(n)}_{AA})^{-1} - (\Sigma_{AA})^{-1}\|_{\infty}\\
&\leq & (\kappa\eta_1 + \eta_2)(\varphi + \eta_3)\,.
\end{eqnarray*}
Moreover, $\eta_3\leq \|(C^{(n)}_{AA})^{-1}\|_{\infty}\times \|C^{(n)}_{AA} - \Sigma_{AA}\|_{\infty}\times \|(\Sigma_{AA})^{-1}\|_{\infty}\leq (\varphi + \eta_3)\varphi\eta_1$. Hence, if $\varphi \eta_1 < 1$, we have $\eta_3 \leq \varphi^2 \eta_1 (1-\varphi \eta_1)^{-1}$. Hence we have,
$$
\|C^{(n)}_{A^cA}(C^{(n)}_{AA})^{-1} - \Sigma_{A^cA}(\Sigma_{AA})^{-1}\|_{\infty}\leq (\kappa \eta_1 + \eta_2)\varphi(1-\varphi \eta_1)^{-1}\,.
$$
Then we consider the event $\max(\eta_1, \eta_2)\leq \varepsilon$.  Note that $\varepsilon <  1/\varphi$ ensures that $\varphi \eta_1 < 1$ on this event.   The conclusion follows from inequalities \eqref{eqn:A8} and \eqref{eqn:A9}.
\end{proof}

\begin{proof}[Proof of Lemma \ref{lemma-sub1}]
Since  $(\Sigma^{-1}\mu_d)_A = \Sigma_{AA}^{-1}(\mu^1_A - \mu^0_A)$ (by Lemma \ref{lem:beta}) and $\widetilde{\mathcal{C}}^0\subset\mathcal{C}^0\subset \mathcal{C}$, we have 
\begin{align*}
&P_0\{C^{**}_\alpha\leq s^*(X)\leq C^{**}_\alpha+\delta | X \in \mathcal{C}\} \cr
&\geq P_0(\{C^{**}_\alpha\leq (\mu^1_A-\mu^0_A)^{\top}\Sigma^{-1}_{AA}X_A\leq C^{**}_\alpha+\delta\}\cap \mathcal{C})\cr
&\geq P_0(\{C^{**}_\alpha\leq (\mu^1_A-\mu^0_A)^{\top}\Sigma^{-1}_{AA}X_A\leq C^{**}_\alpha+\delta\}\cap \widetilde{\mathcal{C}^0})\cr
&= P_0(\{C^{**}_\alpha\leq (\mu^1_A-\mu^0_A)^{\top}\Sigma^{-1/2}_{AA}\widetilde{X}_A\leq C^{**}_\alpha+\delta\}|\widetilde{\mathcal{C}}^0)P_0(\widetilde{\mathcal{C}}^0)\cr
&\geq (1-\delta_3)P_0(C^{**}_\alpha\leq (\mu^1_A-\mu^0_A)^{\top}\Sigma^{-1/2}_{AA}\widetilde{X}_A\leq C^{**}_\alpha+\delta |\widetilde{\mathcal{C}}^0)\,,
\end{align*}
%\adc{Change to $P_0(X\in \widetilde{\mathcal{C}}^0)$?}
where the last inequality uses $P_0(\widetilde{C}^0)\geq 1-\delta_3$. To derive this inequality, let $V^0$ (defined in the proof of Proposition \ref{prop:conddetec}) play the role of $x$ and take $A = I_s$ in Lemma \ref{lem: concentration}, then we have
$$
\p\left(\|V^0\|^2 \geq s + 2 \sqrt{st} + 1 \cdot t  \right)\leq e^{-t}\,, \text{ for all } t> 0\,.
$$  
For $s, t \in\mathbb{N}$, the above inequality clearly implies $\p(\|V^0\|^2 \geq 4st)\leq \exp(-t)$.  Take $t = (n_0\wedge n_1)^{1/2}$, then as long as $c_1'\geq 2/\lambda_m$,
$$
\{x: \|V^0\|^2\leq 4st\}\subset \{x:  \|V^0\|^2\leq \lambda_m^2(c'_1)^2 st\} = \widetilde{C}^0\,.
$$
Since $\lambda_m$ is bounded from below, we can certainly  take $c'_1\geq 2/\lambda_m$ is the proof of Lemma \ref{lem:large_prob_set} in constructing $\widetilde{C}^0$. Therefore, $\p(\|V^0\|^2\leq s + 2\sqrt{st} + t)\geq 1- \exp(-t)$ implies that $\p(\widetilde{C}^0)\geq 1- \exp(-t)$ for $t = (n_0\wedge n_1)^{1/2}$.  
\end{proof}

%\textcolor{blue}{Next, we need to reproduce a counterpart of Theorem 1 in Mai et al. (2012)}

%Same as in \cite{mai2012direct}, the lassoed discriminant analysis direction is computed by
%
%$$
%(\hat\beta^{\text{lasso}}, \hat{\beta}^{\lambda}_0) = \arg\min_{\beta, \beta_0}\left\{n^{-1} \sum_{i=1}^n (y_i - \beta_0 - x_i^{\top} \beta)^2 + \lambda  \sum_{j=1}^p |\beta_j|\right\}\,.
%$$

%\subsubsection*{ }

\begin{proof}[Proof of Lemma \ref{lemma-sub2}]
Since $V^0 =\widetilde{X}_A-\Sigma^{-1/2}_{AA}\mu^0_A$, it follows that,  
\begin{align*}
&P_0(C^{**}_\alpha\leq (\mu^1_A-\mu^0_A)^{\top}\Sigma^{-1/2}_{AA}\widetilde{X}_A\leq C^{**}_\alpha+\delta|\widetilde{\mathcal{C}}^0)\cr
=&P_0(C^{**}_\alpha-(\mu^1_A-\mu^0_A)^{\top}\Sigma^{-1}_{AA}\mu^0_A\leq a^{\top}V^0\leq C^{**}_\alpha+\delta-(\mu^1_A-\mu^0_A)^{\top}\Sigma^{-1}_{AA}\mu^0_A |\widetilde{\mathcal{C}}^0).
\end{align*}
By \cite{mukerjee2015variance}, the probability density function of $V^0|\widetilde{\mathcal{C}}^0$ is given by
\begin{equation} \label{eqn:pdf_truncated_normal}
f_{V^0|\widetilde{\mathcal{C}}^0}(v)=\begin{cases} k_{\widetilde{L},s}\Pi_{i=1}^s \phi(v_i)\ \ \text{if}\ \  \|v\|\leq \widetilde{L}\\
0\,, \ \ \text{otherwise},
\end{cases}
\end{equation}
where $\phi$ is the pdf for the standard normal random variable, $\widetilde{L}$ is defined in equation \eqref{eqn:C0tilde},  and  $k_{\widetilde{L},s}$ is a normalizing constant. Note that $k_{\widetilde{L},s}$ is a monotone decreasing function of $\widetilde{L}$ for each $s$, and when $\widetilde{L}$ goes to infinity, $k_{\widetilde{L},s}=k^0_{s}$ is a positive constant. Therefore, $k_{\widetilde{L},s}$ is bounded below by $k^0_{s}$. Since we only consider $s\in\{1, 2\}$, we can take $k^0$ as a universal constant independent of $s$, and $k_{\widetilde{L},s}$ is bounded below by $k^0$ universally.  

Let $f_{a^{\top}V^0|\widetilde{\mathcal{C}}^0}(z)$ be the density of  $a^{\top}V^0|\widetilde{\mathcal{C}}^0$. Thus, we want to lower bound
\begin{align*}
&P_0(C^{**}_\alpha-(\mu^1_A-\mu^0_A)^{\top}\Sigma^{-1}_{AA}\mu^0_A\leq a^{\top}V^0\leq C^{**}_\alpha+\delta-(\mu^1_A-\mu^0_A)^{\top}\Sigma^{-1}_{AA}\mu^0_A |\widetilde{\mathcal{C}}^0)\cr
=&\int_{C^{**}_\alpha-(\mu^1_A-\mu^0_A)^{\top}\Sigma^{-1}_{AA}\mu^0_A}^{C^{**}_\alpha+\delta-(\mu^1_A-\mu^0_A)^{\top}\Sigma^{-1}_{AA}\mu^0_A} f_{a^{\top}V^0|\widetilde{\mathcal{C}}^0}(z) dz\,.
\end{align*}
Let us  analyze $f_{a^{\top}V^0|\widetilde{\mathcal{C}}^0}(z)$ when $s=1$ and $s=2$.

\noindent{\bf Case 1} ($s=1$): $a$ is a scalar. Hence
\begin{equation*}
f_{aV^0|\widetilde{\mathcal{C}}^0}(z)=\begin{cases}\frac{k_{\widetilde{L},1}}{|a|}\phi(\frac{z}{a}), \ \ \text{for} \ \ |z|\leq |a|\widetilde{L} \\
0, \ \ \text{otherwise},
\end{cases}
\end{equation*}
which is the density function of a truncated Normal random variable with parent distribution $\mathcal{N}(0, |a|)$ symmetrically truncated to $-|a|\widetilde{L}$ and $|a|\widetilde{L}$, i.e. $TN(0,|a|,-|a|\widetilde{L},|a|\widetilde{L})$. Here $|a|$ is the standard deviation of the parent Normal distribution.  Therefore, 
\[
f_{aV^0|\widetilde{\mathcal{C}}^0}(z)\geq f_{\mathcal{N}(0,|a|)}(z), \ \ \text{for}\ \  |z|\leq |a|\widetilde{L}.
\]
This implies  
\begin{align*}
&\int_{C^{**}_\alpha-(\mu^1_A-\mu^0_A)^{\top}\Sigma^{-1}_{AA}\mu^0_A}^{C^{**}_\alpha+\delta-(\mu^1_A-\mu^0_A)^{\top}\Sigma^{-1}_{AA}\mu^0_A} f_{aV^0|\widetilde{\mathcal{C}}^0}(z) dz\cr
&\geq \delta \min \{ f_{{\mathcal{N}(0,|a|)}}(C^{**}_\alpha-(\mu^1_A-\mu^0_A)^{\top}\Sigma^{-1}_{AA}\mu^0_A), f_{{\mathcal{N}(0,|a|)}}(C^{**}_\alpha-(\mu^1_A-\mu^0_A)^{\top}\Sigma^{-1}_{AA}\mu^0_A+\delta^*)\} \cr
&\geq \delta M.
\end{align*}
where the inequality follows from the mean-value theorem and our assumption (ii).

\noindent {\bf Case 2} ($s=2$): $a=(a_1, a_2)$ is a vector. Now let us do the following change of variable from $(V_1,V_2)=V^0|\widetilde{\mathcal{C}}$ to $(Z_1,Z_2)=(a^{\top}V^{0}|\mathcal{C}, V_2)$. 
\begin{equation}
\begin{cases}Z_1=a_1V_1+a_2V_2\\ Z_2=V_2 \end{cases} \ \ \ \text{, and thus}\ \ \ \begin{cases}V_1=\frac{Z_1-a_2 Z_2}{a_1}\\ V_2=Z_2 \end{cases}
\end{equation}
The original event  $S_{V_1,V_2}=\{V^2_1+V^2_2 \leq \widetilde{L}^2\}$ is equivalent to 
\begin{align*}
S_{Z_1,Z_2}=& \left(\frac{Z_1-a_2 Z_2}{a_1}\right)^2+Z^2_2\leq \widetilde{L}^2\cr
\Leftrightarrow &\left(\frac{a^2_2}{a^2_1}+1\right)\left(Z_2-\left(\frac{a_1a_2}{a^2_1+a^2_2}\right)Z_1\right)^2+\frac{a^2_1+a^2_2-a^2_1 a^2_2}{a^2_1(a^2_1+a^2_2)}Z^2_1\leq \widetilde{L}^2.
\end{align*}
Now for any $z_1$, the marginal density of $a^\top V^0$ can be carried out as
\begin{align*}
&\quad f_{Z_1}(z_1)\\
 &= \int_{S_{z_1,z_2}} \frac{k_{\widetilde{L},2}}{|a_1|} \phi\left(\frac{z_1-a_2 z_2}{a_1}\right)\phi(z_2) dz_2\cr
&=  \int_{S_{z_1,z_2}} \frac{k_{\widetilde{L},2}}{2\pi |a_1|} \exp\{-\frac{1}{2}\left(\frac{z_1-a_2 z_2}{a_1}\right)^2-\frac{z^2_2}{2}\} dz_2\cr
&=\left(\frac{k_{\widetilde{L},2}}{2\pi|a_1|}\exp\{-\frac{a^2_1+a^2_2-a^2_2 a^2_1}{2a^2_1 (a^2_1+a^2_2)} z^2_1\}\right)\int_{S_{z_1,z_2}} \exp\{-\frac{(z_2-\left(\frac{a_1a_2}{a^2_1+a^2_2}\right)z_1)^2}{\frac{2}{{a^2_2}/{a^2_1}+1}}\}dz_2\cr
&=\left(\frac{k_{\widetilde{L},2}}{|a_1|\sqrt{2\pi(\frac{a^2_2}{a^2_1}+1)}}\exp\{-\frac{a^2_1+a^2_2-a^2_2 a^2_1}{2a^2_1 (a^2_1+a^2_2)} z^2_1\}\right)\int _{S_{z_1,z_2}}\phi_{\mathcal{N}(\big(\frac{a_1a_2}{a^2_1+a^2_2}\big)z_1,\frac{1}{\sqrt{{a^2_2}/{a^2_1}+1}})}(z_2) dz_2\cr
\end{align*}
\begin{align*}
&=\left(\frac{k_{\widetilde{L},2}}{\sqrt{2\pi}|a_1|\big(\frac{a^2_2}{a^2_1}+1\big)}\exp\{-\frac{a^2_1+a^2_2-a^2_2 a^2_1}{a^2_1 (a^2_1+a^2_2)} z^2_1\}\right)\int _{-\sqrt{\widetilde{L}^2-\frac{a^2_1+a^2_2-a^2_2 a^2_1}{a^2_1 (a^2_1+a^2_2)} z^2_1}}^{\sqrt{\widetilde{L}^2-\frac{a^2_1+a^2_2-a^2_2 a^2_1}{a^2_1 (a^2_1+a^2_2)} z^2_1} }\phi_{\mathcal{N}(0,1)}(z) dz\cr
&= \left(\frac{k_{\widetilde{L},2}}{\sqrt{2\pi}|a_1|\big(\frac{a^2_2}{a^2_1}+1\big)}\exp\{-\frac{a^2_1+a^2_2-a^2_2 a^2_1}{a^2_1 (a^2_1+a^2_2)} z^2_1\}\right) \left(2\Phi(\sqrt{\widetilde{L}^2-\frac{a^2_1+a^2_2-a^2_2 a^2_1}{a^2_1 (a^2_1+a^2_2)} z^2_1})-1\right).
\end{align*}

This implies  
\begin{align*}
&\int_{C^{**}_\alpha-(\mu^1_A-\mu^0_A)^{\top}\Sigma^{-1}_{AA}\mu^0_A}^{C^{**}_\alpha+\delta-(\mu^1_A-\mu^0_A)^{\top}\Sigma^{-1}_{AA}\mu^0_A} f_{a^{\top}V^0|\widetilde{\mathcal{C}}^0}(z) dz\cr
&\geq \delta \min \{ f_{Z_1}(C^{**}_\alpha-(\mu^1_A-\mu^0_A)^{\top}\Sigma^{-1}_{AA}\mu^0_A), f_{Z_1}(C^{**}_\alpha-(\mu^1_A-\mu^0_A)^{\top}\Sigma^{-1}_{AA}\mu^0_A+\delta^*)\} \cr
&\geq \delta Mk_0.
\end{align*}
where the inequality follows from the mean-value theorem and our assumption (iii).

We can safely conclude our proof by combining cases $s=1$ and $s=2$, and taking $M_1=\min\{M, Mk_0\}$.
\end{proof}

\begin{proof}[Proof of Proposition \ref{thm:1}]
The proof is largely identical to that of Theorem 1 in \cite{mai2012direct}, except the differences due to a different sampling scheme.   

Similarly to \cite{mai2012direct}, by the definition of $\hat\beta_A$, we can write $\hat\beta_A = (n^{-1}\widetilde{X}^{\top}_A \widetilde{X}_A)^{-1}\{(\widehat{\mu}^1_A  - \widehat{\mu}^0_A) - \lambda t_A/2\}$, where $t_A$ represents the subgradient such that $t_j = \text{sign}(\hat\beta_j)$ if $\hat\beta_j \neq 0$ and $-1 < t_j < 1$ if $\hat\beta_j = 0$.   To show that $\hat{\beta}^{\text{lasso}} = (\hat{\beta}_A, 0)$, it suffices to verify that 
\begin{equation}\label{eqn:kkt}
\|n^{-1}\widetilde{X}^{\top}_{A^c}\widetilde{X}_A \hat\beta_A - (\widehat{\mu}^1_{A^c} - \widehat{\mu}^0_{A^c})\|_{\infty}\leq \lambda/2\,.
\end{equation}
The left-hand side of \eqref{eqn:kkt} is equal to 
\begin{equation}\label{eqn:kkt_equi}
\|C^{(n)}_{A^cA}(C^{(n)}_{AA})^{-1} (\widehat{\mu}^1_{A} - \widehat{\mu}^0_{A}) - C^{(n)}_{A^cA}(C^{(n)}_{AA})^{-1}  \lambda t_A/2 - (\widehat{\mu}^1_{A^c} - \widehat{\mu}^0_{A^c}) \|_{\infty}.
\end{equation}
Using $\Sigma_{A^cA}\Sigma_{AA}^{-1} (\mu^1_A - \mu^0_A) = (\mu^1_{A^c} - \mu^0_{A^c})$, \eqref{eqn:kkt_equi} is bounded from above by 
\begin{eqnarray*}
U_1 &=& \|C^{(n)}_{A^cA}(C^{(n)}_{AA})^{-1} - \Sigma_{A^cA}\Sigma_{AA}^{-1} \|_{\infty}\Delta + \|(\widehat{\mu}^1_{A^c} - \widehat{\mu}^0_{A^c}) - (\mu^1_{A^c} - \mu^0_{A^c})\|_{\infty}\\
&& + (\|C^{(n)}_{A^cA}(C^{(n)}_{AA})^{-1} - \Sigma_{A^cA}\Sigma_{AA}^{-1} \|_{\infty} + \kappa)\|(\widehat{\mu}^1_{A} - \widehat{\mu}^0_{A}) - (\mu^1_{A} - \mu^0_{A})\|_{\infty}\\
&& + (\|C^{(n)}_{A^cA}(C^{(n)}_{AA})^{-1} - \Sigma_{A^cA}\Sigma_{AA}^{-1} \|_{\infty} + \kappa)\lambda /2\,.
\end{eqnarray*}

If $\|C^{(n)}_{A^cA}(C^{(n)}_{AA})^{-1} - \Sigma_{A^cA}\Sigma_{AA}^{-1} \|_{\infty}\leq (\kappa+1)\varepsilon \varphi (1 - \varphi \varepsilon)^{-1}$ (invoke Lemma \ref{lem:prod_deviation}), and $\|(\widehat{\mu}^1-\widehat{\mu}^0) - (\mu^1 - \mu^0)\|_{\infty}\leq 4^{-1} \lambda (1-\kappa - 2\varepsilon \varphi)/ (1 + \kappa)$, and given $\varepsilon \leq \min [ \varepsilon_0, \lambda (1-\kappa) (4\varphi)^{-1}(\lambda/2 + (1+\kappa)\Delta)^{-1}]$, then $U_1\leq \lambda/2$.

Therefore, by Lemmas \ref{lem:A1} and \ref{lem:prod_deviation}, we have
\begin{align*}
&\quad\,\, \p \{\|n^{-1}\widetilde{X}^{\top}_{A^c}\widetilde{X}_A \hat\beta_A - (\widehat{\mu}^1_{A^c} - \widehat{\mu}^0_{A^c})\|_{\infty}\leq \lambda/2\}\\
&\geq   1 - 2d\exp(-n_0 \varepsilon^{*2} c_2)  - 2d\exp(-n_1 \varepsilon^{*2} c_2) - f(d, s, n_0, n_1, (\kappa+1)\varepsilon \varphi (1 - \varphi \varepsilon)^{-1})\,,
\end{align*}
where $\varepsilon^* = 4^{-1} \lambda (1-\kappa - 2\varepsilon \varphi)/ (1 + \kappa)$, and $f$ is the same as in Lemma \ref{lem:prod_deviation}. 
Tidy up the algebra a bit, we can write 
$$
\delta^*_1  = \sum_{l=0}^1 2d \exp\left(-c_2 n_l \frac{\lambda^2 (1 - \kappa - 2 \varepsilon \varphi)^2}{16(1 + \kappa)^2}\right) + f(d, s, n_0, n_1, (\kappa+1)\varepsilon \varphi (1 - \varphi \varepsilon)^{-1})\,.
$$

To prove the 2nd conclusion, note that 
\begin{align}\label{eqn:betahatA}
\hat\beta_{A} = &\quad (\Sigma_{AA})^{-1} (\mu^1_A - \mu^0_A) + (C^{(n)}_{AA})^{-1}\{(\widehat{\mu}^1_A - \widehat{\mu}^0_A)-(\mu^1_A - \mu^0_A)\}\\
& + \{(C^{(n)}_{AA})^{-1} - (\Sigma_{AA})^{-1}\}(\mu^1_A - \mu^0_A) -  \lambda (C^{(n)}_{AA})^{-1} t_A /2\,.
\end{align}
Let $\xi = |\beta^*|_{\text{min}}/ (\Delta \varphi)$.  Write $\eta_1 = \|\Sigma_{AA} - C^{(n)}_{AA}\|_{\infty}$ and $\eta_3 = \|(C^{(n)}_{AA})^{-1} - \Sigma_{AA}^{-1}\|_{\infty}$.  Then for any $j\in A$, 
$$
|\hat \beta_j|\geq \xi \Delta \varphi - (\eta_3 + \varphi)\{\lambda /2 + \|(\widehat{\mu}^1_A - \widehat{\mu}^0_A)-(\mu^1_A - \mu^0_A)\|_{\infty}\} - \eta_3 \Delta\,.
$$
When $\eta_1\varphi < 1$, we have shown that $\eta_3 < \varphi^2 \eta_1 (1 - \eta_1 \varphi)^{-1}$ in Lemma \ref{lem:prod_deviation}. Therefore, 
$$
|\hat \beta_j|\geq \xi \Delta \varphi - (1 - \eta_1 \varphi)^{-1}\{\lambda \varphi / 2 +  \|(\widehat{\mu}^1_A - \widehat{\mu}^0_A)-(\mu^1_A - \mu^0_A)\|_{\infty}\varphi + \varphi^2\eta_1\Delta  \} \equiv L_1\,.
$$

Because $\|\beta^*\|_{\infty}\leq \Delta \varphi$, $\xi \leq 1$. Hence $\lambda \leq |\beta^*|_{
\text{min}}/(2\varphi)\leq 2 |\beta^*|_{\text{min}}/\{(3 + \xi)\varphi\}$. Under the events $\eta_1 \leq \varepsilon$ and $\|(\widehat{\mu}^1_A - \widehat{\mu}^0_A)-(\mu^1_A - \mu^0_A)\|_{\infty} \leq \varepsilon$, together with restriction on $\varepsilon$, we have $L_1 > 0$. Therefore, 
$$
\p(L_1 > 0) \geq 1 - \sum_{l=0}^1 2s\exp(-n_l \varepsilon^2 c_2)  - \sum_{l=0}^1 2s^2 \exp\left(- \frac{c_1 \varepsilon^2 n^2}{4 n_l s^2}  \right)\,.
$$

To prove the 3rd conclusion, equation \eqref{eqn:betahatA} and $\eta_1 \varphi < 1$ imply that 
$$
\|\hat\beta_A - \beta^*\|_{\infty}\leq (1 - \eta_1 \varphi)^{-1}\{\lambda \varphi / 2 +  \|(\widehat{\mu}^1_A - \widehat{\mu}^0_A)-(\mu^1_A - \mu^0_A)\|_{\infty}\varphi + \varphi^2\eta_1\Delta  \}\,.
$$
On the events $\{\eta_1 < \varepsilon\}$ and $\{\|(\widehat{\mu}^1_A - \widehat{\mu}^0_A)-(\mu^1_A - \mu^0_A)\|_{\infty}\leq \varepsilon\}$, and under restrictions for $\varepsilon$ and $\lambda$ in the assumption, we have $\|\hat\beta_A - \beta^*\|_{\infty}\leq 4 \varphi \lambda$.  Hence,
\begin{eqnarray*}
\p(\|\hat\beta_A - \beta^*\|_{\infty}\leq 4 \varphi \lambda)\geq 1 - \sum_{l=0}^1 2s\exp(-n_l \varepsilon^2 c_2) - \sum_{l=0}^1 2s^2 \exp\left(- \frac{c_1 \varepsilon^2 n^2}{4  s^2n_l}  \right)\,.
\end{eqnarray*}
\end{proof}

\begin{proof}[Proof of Proposition \ref{prop:conddetec}]
	For simplicity, we will derive the lower bound for one of the two probabilities in the definition:
\begin{eqnarray}\label{eqn:one_sided}
P_0\{C^{**}_\alpha\leq s^*(X)\leq C^{**}_\alpha+\delta |  X\in \mathcal{C}\}\geq (1-\delta_3)M_1 \delta\,, \text{ for } \delta \in (0, \delta^*)\,.
\end{eqnarray}
The lower bound for the other probability can be derived similarly.  

Recall that 
$
\mathcal{C}^0=\{X\in\mathbb{R}^d: \|X_A-\mu^{0}_A\|\leq c'_1 s^{1/2}(n_0\wedge n_1)^{1/4}\doteq L\}
$ (in the proof of Lemma \ref{lem:large_prob_set}). Let $V^0  = \Sigma^{-1/2}_{AA}(X_A-\mu^0_A)=\widetilde{X}_A-\Sigma^{-1/2}_{AA}\mu^0_A$, where $\widetilde{X}_A = \Sigma_{AA}^{-1/2} X_A$, then $V^0\sim\mathcal{N}(0,I_s)$ under $P_0$.
Define an event
\begin{eqnarray}\label{eqn:C0tilde}
\widetilde{\mathcal{C}}^0=\left\{X\in\mathbb{R}^d: \|V^0\|\leq \lambda_m L \doteq \widetilde{L}\right\}\,,
\end{eqnarray}
where $\lambda_m=\lambda_{\min}(\Sigma^{-1/2}_{AA})$ and $\lambda_{\min}(\cdot)$ denotes the minimum eigenvalue of a matrix. Since
$
\|V^0\|\geq \lambda_{\min}(\Sigma^{-1/2}_{AA})\|X_A-\mu^0_A\|=\lambda_m \|X_A-\mu^0_A\|
$, we have $\widetilde{\mathcal{C}}^0\subset\mathcal{C}^0$. Then inequality \eqref{eqn:one_sided} holds by invoking Lemma \ref{lemma-sub1} and Lemma \ref{lemma-sub2}.
\end{proof}

\begin{remark}\label{remark2}
Proof of Proposition \ref{prop:conddetec} indicates that the same conclusion  would hold for a general $s\in \mathbb{N}$, if the density of $(\mu^1_A-\mu^0_A)^{\top}\Sigma^{-1/2}_{AA}\widetilde{X}_A|\widetilde{\mathcal{C}}^0$ is bounded below on $(C^1-\delta^*, C^2+\delta^*)$ by some constant.
\end{remark}

\begin{proof}[Proof of Theorem \ref{thm:main}]
	The first inequality follows from Proposition \ref{prop} and the choice of $k^*$ in \eqref{eqn:kstar} of the main paper.  In the following, we prove the second inequality.

Let $G^* = \{ s^* \leq C^{**}_{\alpha} \}$ and $\widehat G = \{\hat s \leq \widehat C_{\alpha}\}$.  
The excess type II error can be decomposed as
\begin{equation}\label{eqn:typeII decomposition}
P_1(\widehat G) - P_1(G^*) = \int_{\widehat G \backslash G^*}|r - C_{\alpha}|dP_0 + \int_{G ^* \backslash \widehat G }|r - C_{\alpha}|dP_0 + C_{\alpha}\{ R_0(\phi^*_{\alpha}) - R_0(\hat \phi_{k^*}) \}\,.
\end{equation}
In the above decomposition, the third part can be bounded via Lemma \ref{lem:2}.  For the first two parts, let $$T = \|\hat s - s^* \|_{\infty, \mathcal{C}} := \max_{x\in \mathcal{C}}|\hat s(x) - s^*(x)|\,, \text{ and }$$  $$\Delta R_{0, \mathcal{C}} := |R_0(\phi^*_{\alpha} |\mathcal{C}) - R_0(\hat\phi_{k^*} |\mathcal{C})|= |P_0(s^*(X) > C^{**}_{\alpha}| X\in\mathcal{C}) - P_0(\hat s(X) >  \widehat C_{\alpha}|X\in\mathcal{C})|\,,$$ where $\mathcal{C}$ is defined in Lemma \ref{lem:large_prob_set}. 
A high probability bound for $\Delta R_{0, \mathcal{C}}$ was derived in Lemma \ref{lem:4}.

%\textcolor{red}{need to add some lines about $(\Delta R_{0,\mathcal{C}} / M_1)^{1/\uderbar{\gamma}}$ fails into the detection proper range. This is slightly more complicated than Zhao (2016) because now $\Delta R_{0 \mathcal{C}}$ contains two parts.}

It follows from Lemma \ref{prop::R0} that if $n_0'\geq \max\{4/(\alpha\alpha_0), \delta_0^{-2}, (\delta'_0)^{-2}, (\frac{1}{10} M_1\delta^{*\uderbar{\gamma}})^{-4}\}$,
$$
\xi_{\alpha, \delta_0, n'_0}(\delta_0') \leq \frac{5}{2}(n_0')^{-1/4}\leq \frac{1}{4}M_1(\delta^*)^{\uderbar{\gamma}}\,.
$$
Because the lower bound in the detection condition should be smaller than $1$ to make sense,  $M_1 \delta^{*\uderbar{\gamma}} < 1$.  
This together with $n_0 \wedge n_1 \geq [- \log (M_1 \delta^{*\uderbar{\gamma}}/4)]^2$ implies that $\exp\{-(n_0\wedge n_1)^{1/2}\}\leq M_1 \delta^{*\uderbar{\gamma}}/4$.  

Let $\mathcal{E}_2 = \{R_{0, \mathcal{C}}\leq 2[\xi_{\alpha, \delta_0,n_0'}(\delta_0') + \exp\{- (n_0 \wedge n_1)^{1/2}\}]\}$.  On the event $\mathcal{E}_2$ we have 

$$
\left\{ \frac{R_{0, \mathcal{C}}}{M_1} \right\}^{1/\uderbar{\gamma}}\leq \left\{ \frac{2[\xi_{\alpha, \delta_0,n_0'}(\delta_0') + \exp\{- (n_0 \wedge n_1)^{1/2}\}]}{M_1}\right\}^{1/\uderbar{\gamma}} \leq \delta^*\,.
$$

To find the relation between $C^{**}_{\alpha}$ and $\widehat C_{\alpha}$, we invoke the detection condition as follows: 
\begin{align*}
&\quad \,\, P_0\left(s^*(X) \geq C^{**}_{\alpha} + (\Delta R_{0,\mathcal{C}} / M_1)^{1/\uderbar{\gamma}} | X \in \mathcal{C} \right)\\
&=  R_0(\phi^*_{\alpha} |\mathcal{C}) - P_0(C^{**}_{\alpha} < s^*(X) < C^{**}_{\alpha} + (\Delta R_{0,\mathcal{C}} / M_1)^{1/\uderbar{\gamma}} | X\in \mathcal{C} ) \\
&\leq  R_0(\phi^*_{\alpha} |\mathcal{C}) - \Delta R_{0,\mathcal{C}} \text{ }\text{ }\text{ }\text{  (by detection condition) }\\
&\leq   R_0(\hat \phi_{k^*}|\mathcal{C}) = P_0(\hat s(X) > \widehat C_{\alpha} | X\in \mathcal{C})\\
&\leq  P_0(s^*(X) > \widehat C_{\alpha} - T | X\in \mathcal{C})\,.
\end{align*}
This implies that $C^{**}_{\alpha} + (\Delta R_{0,\mathcal{C}} / M_1)^{1/\uderbar{\gamma}} \geq \widehat C_{\alpha} - T$, which further implies that
$$
\widehat C_{\alpha}\leq C^{**}_{\alpha} + (\Delta R_{0,\mathcal{C}} / M_1)^{1/\uderbar{\gamma}} + T\,.
$$

Note that 
\begin{align*}
&\quad\,\, \mathcal{C}\cap (\widehat G \backslash G^*) \\
&= \mathcal{C}\cap \{ s^* > C^{**}_{\alpha}, \hat s \leq \widehat C_{\alpha} \} \\
&= \mathcal{C}\cap\{s^* >  C^{**}_{\alpha}, \hat s \leq  C^{**}_{\alpha}+(\Delta R_{0,\mathcal{C}} / M_1)^{1/\uderbar{\gamma}} + T \} \cap \{\hat s \leq \widehat C_{\alpha} \}\\
&\subset  \mathcal{C}\cap \{C^{**}_{\alpha} + (\Delta R_{0,\mathcal{C}} / M_1)^{1/\uderbar{\gamma}} + 2T \geq s^* \geq C^{**}_{\alpha}, \hat s \leq C^{**}_{\alpha}+(\Delta R_{0,\mathcal{C}} / M_1)^{1/\uderbar{\gamma}} + T \}\cap \{\hat s \leq \widehat C_{\alpha} \}\\
&\subset  \mathcal{C}\cap\{C^{**}_{\alpha} + (\Delta R_{0,\mathcal{C}} / M_1)^{1/\uderbar{\gamma}} + 2T \geq s^* \geq C^{**}_{\alpha}\}\,.
\end{align*}

%Let $\Delta R_0 = |R_0(\phi^*) - R_0(\hat \phi_{k^*})|$.      Suppose \textcolor{blue}{$P_0(X\in\mathcal{C})\geq 1 - \delta_3$}. We would like to bound $\Delta R_{0, \mathcal{C}}$ in terms of $\Delta R_0$.
%
%\begin{eqnarray*}
%\Delta R_{0, \mathcal{C}} 
%&=& \frac{|P_0(s^*(X)\geq C^{**}_{\alpha}, \mathcal{C}) - P_0(\hat s(X)\geq \widehat C_{\alpha}, \mathcal{C})|}{P_0(X\in\mathcal{C})}\\
%&\leq & \frac{1}{1-\delta_3} |P_0(s^*(X)\geq C^{**}_{\alpha}, \mathcal{C}) - P_0(\hat s(X)\geq \widehat C_{\alpha}, \mathcal{C})|\\
%&\leq & \frac{1}{1-\delta_3} |P_0(s^*(X)\geq C^{**}_{\alpha}) - P_0(\hat s(X)\geq \widehat C_{\alpha})|\\
%&=& \frac{\Delta R_0}{1 - \delta_3}\,.
%\end{eqnarray*}
%\textcolor{red}{The above inequality is actually not correct!  }

\noindent We decompose as follows
$$
\int_{(\widehat G \backslash G^*)}|r - C_{\alpha}|d P_0 = \int_{(\widehat G \backslash G^*) \cap \mathcal{C}}|r - C_{\alpha}|d P_0 + \int_{(\widehat G \backslash G^*) \cap \mathcal{C}^c}|r - C_{\alpha}|d P_0 =: (\text{I}) + (\text{II})\,.
$$
To bound (I), recall that  
$$
r(x)  = \frac{f_1(x)}{f_0(x)} =  \exp\left(s^*(x) - \mu_a ^{\top} \Sigma^{-1}\mu_d\right)\,,
$$
and that, $r(x) > C_{\alpha}$ is equivalent to $s^*(x) > C^{**}_{\alpha} = \log C_{\alpha} + \mu_a^{\top} \Sigma^{-1}\mu_d$.  By the mean value theorem, we have
$$
|r(x) - C_{\alpha}| = e^{-\mu_a^{\top} \Sigma^{-1}\mu_d}|e^{s^*(x)} - e^{C^{**}_{\alpha}}| = e^{-\mu_a^{\top} \Sigma^{-1}\mu_d}\cdot e^{z'} |s^*(x) - C^{**}_{\alpha}|\,,
$$
where 
$z'$ is some quantity between $s^*(x)$ and $C^{**}_{\alpha}$.  %We can treat $\mu_a\Sigma^{-1}\mu_d$ as if it is a constant (or bounded from below).  
%
%\textcolor{blue}{The problem is how do we deal with the assumption on $C^{**}_{\alpha}$ Suppose $C^{**}_{\alpha}$ is bounded from above. Further more, as $C^{**}_{\alpha} = \log C_{\alpha} + \mu_a \Sigma^{-1}\mu_d$, this restriction can be put as more transparent assumptions.} %\textcolor{red}{The s should be changed to s star, because the s is also used to denote the sparsity}
Denote by $\mathcal{C}_1 = \{x:  C^{**}_{\alpha} + (\Delta R_{0,\mathcal{C}} / M_1)^{1/\uderbar{\gamma}} + 2T \geq s^*(x) \geq C^{**}_{\alpha}\}$.
Restricting to  
$\mathcal{C}\cap \mathcal{C}_1$, we have 
$$
z'\leq C^{**}_{\alpha} + (\Delta R_{0,\mathcal{C}} / M_1)^{1/\uderbar{\gamma}} + 2T\,.
$$
This together with $\mathcal{C}\cap (\widehat G \backslash G^*) \subset \mathcal{C}\cap \mathcal{C}_1$  implies that 
\begin{align*}
(\text{I})&\leq  \int_{\mathcal{C}\cap \mathcal{C}_1}|r - C_{\alpha}|dP_0\\
& =  \int_{\mathcal{C}\cap \mathcal{C}_1}\exp\{z' - \mu_a^{\top} \Sigma^{-1} \mu_d\}|s^*(x) - C^{**}_{\alpha}|dP_0\\
&\leq  \int_{\mathcal{C}\cap \mathcal{C}_1} \exp\left\{C^{**}_{\alpha} + (\Delta R_{0,\mathcal{C}} / M_1)^{1/\uderbar{\gamma}} + 2T - \mu_a^{\top} \Sigma^{-1} \mu_d\right\}|s^*(x) - C^{**}_{\alpha}|dP_0\,.
\end{align*}

Since $C_{\alpha}$ and $\mu_a^{\top} \Sigma^{-1} \mu_d$ are assumed to be bounded, $C^{**}_{\alpha} = \log C_{\alpha} + \mu_a^{\top} \Sigma^{-1}\mu_d$ is also bounded.    Let $\mathcal{E}_2 = \{R_{0, \mathcal{C}}\leq 2[\xi_{\alpha, \delta_0,n_0'}(\delta_0') + \exp\{- (n_0 \wedge n_1)^{1/2}\}]\}$.  By Lemma \ref{lem:4}, $\p(\mathcal{E}_2)\geq 1 - \delta_0 - \delta_0'$.  Let $\mathcal{E}_3 = \{T\leq 4c_1 \varphi \lambda s (n_0\wedge n_1)^{1/4}\}$. By Lemma \ref{lem:large_prob_set}, $\p(\mathcal{E}_3)\geq 1- \delta_1 - \delta_2$.  Restricting to the event $\mathcal{E}_2 \cap \mathcal{E}_3$, $R_{0, \mathcal{C}}$ and $T$ are bounded.  Therefore on the event  $\mathcal{E}_2 \cap \mathcal{E}_3$, there exists a positive constant $c'$ such that 
\begin{align*}
(\text{I})&\leq c' \int_{\mathcal{C}\cap \mathcal{C}_1}|s^*(x) - C^{**}_{\alpha}|dP_0\\
&\leq  c'\left((\Delta R_{0,\mathcal{C}} / M_1)^{1/\uderbar{\gamma}} + 2T\right)P_0(\mathcal{C}\cap \mathcal{C}_1)\,.\\
%&\leq &  c_2(\Delta R_0 / [M_1(1-\delta_3)] + 2T)\cdot P_0(\mathcal{C}\cap \mathcal{C}_1)\,.
\end{align*}
Note that by the margin assumption (we know $\bar\gamma = 1$, but we choose to reserve the explicit dependency of $\bar\gamma$ by not substituting the numerical value), 
\begin{align*}
\quad\,\, P_0(\mathcal{C}\cap \mathcal{C}_1)
&=  P_0(C^{**}_{\alpha} + (\Delta R_{0,\mathcal{C}} / M_1)^{1/\uderbar{\gamma}} + 2T \geq s^* \geq C^{**}_{\alpha}, \mathcal{C}  )\\
&\leq  P_0(C^{**}_{\alpha} + (\Delta R_{0,\mathcal{C}} / M_1)^{1/\uderbar{\gamma}} + 2T \geq s^* \geq C^{**}_{\alpha} |\mathcal{C}  )\\
&\leq  M_0 \left((\Delta R_{0,\mathcal{C}} / M_1)^{1/\uderbar{\gamma}} + 2T\right)^{\bar{\gamma}}\,.
%&=& \frac{M_0}{M_1}\Delta R_{0, \mathcal{C}} + 2M_0T \leq \frac{M_0 \Delta R_0}{M_1(1-\delta_3)} + 2M_0 T \,. 
\end{align*}
Therefore,
$$
(\text{I})\leq c' M_0 \left((\Delta R_{0,\mathcal{C}} / M_1)^{1/\uderbar{\gamma}} + 2T\right)^{1+\bar{\gamma}}\,.
$$
%\textcolor{blue}{Expand this later}.  
Regarding (II), by Lemma \ref{lem:large_prob_set} we have 
\begin{align*}
	(\text{II})&\leq \int_{\mathcal{C}^c}|r - C_{\alpha}|dP_0\leq \int_{\mathcal{C}^c} rdP_0 + C_{\alpha}\int_{\mathcal{C}^c}dP_0= P_1(\mathcal{C}^c) + C_{\alpha} P_0(\mathcal{C}^c)\\
	&\leq (1 + C_{\alpha}) \exp\{- (n_0 \wedge n_1)^{1/2}\}\,.
\end{align*}

Therefore, 
$$
\int_{(\widehat G \backslash G^*)}|r - C_{\alpha}|d P_0 \leq c' M_0 \left((\Delta R_{0,\mathcal{C}} / M_1\right)^{1/\uderbar{\gamma}} + 2T)^{1+\bar{\gamma}} + (1 + C_{\alpha}) \exp\{- (n_0 \wedge n_1)^{1/2}\}\,.
$$

To bound $\int_{(G^* \backslash \widehat G )}|r - C_{\alpha}|d P_0$, we decompose 
$$
\int_{(G^* \backslash \widehat G )}|r - C_{\alpha}|d P_0 = \int_{(G^* \backslash \widehat G ) \cap \mathcal{C}}|r - C_{\alpha}|d P_0 + \int_{(G^* \backslash \widehat G )\cap \mathcal{C}^c}|r - C_{\alpha}|d P_0 =: (\text{I}') + (\text{II}'')\,.
$$
To bound (I$'$), we invoke both the margin assumption and the detection condition, and we need to define a new a new quantity $\bar{\Delta} R_{0, \mathcal{C}} := P_0(s^*(X) > C^{**}_{\alpha}| X\in\mathcal{C}) - P_0(\hat s(X) >  \widehat C_{\alpha}|X\in\mathcal{C})$.   When $\bar{\Delta} R_{0, \mathcal{C}} \geq 0$, we have
\begin{align*}
&\quad\,\, P_0\left(s^*(X) \geq C^{**}_{\alpha} + (\bar{\Delta} R_{0,\mathcal{C}} / M_0)^{1/\bar{\gamma}} | X \in \mathcal{C} \right)\\
&=  P_0(s^*(X) > C^{**}_{\alpha} | X\in \mathcal{C}) - P_0(C^{**}_{\alpha} < s^*(X) <  C^{**}_{\alpha} + (\bar{\Delta} R_{0,\mathcal{C}} / M_0)^{1/\bar{\gamma}} | X\in \mathcal{C} ) \\
&\geq  P_0(s^*(X) > C^{**}_{\alpha} | X\in \mathcal{C})  - \bar{\Delta} R_{0,\mathcal{C}} \text{ }\text{ }\text{ }\text{  (by margin assumption) }\\
&=   P_0(\hat s(X) >  \widehat C_{\alpha}|X\in\mathcal{C})\\
&\geq  P_0(s^*(X) > \widehat C_{\alpha} + T | X\in \mathcal{C})\,.
\end{align*}
%\textcolor{red}{ The first equality in the above equation is problematic.}
%\textcolor{red}{The first part is problematic.}
So when $\bar{\Delta} R_{0, \mathcal{C}} \geq 0$, $\widehat C_{\alpha} \geq C^{**}_{\alpha} - (\bar{\Delta} R_{0, \mathcal{C}}/M_0)^{1/\bar{\gamma}} - T$.  On the other hand, when $\bar{\Delta} R_{0, \mathcal{C}} < 0$, 
\begin{align*}
&\quad \,\, P_0\left(s^*(X) \geq C^{**}_{\alpha} - (- \bar{\Delta} R_{0,\mathcal{C}} / M_1)^{1/\uderbar{\gamma}} | X \in \mathcal{C} \right)\\
&=  P_0(s^*(X) > C^{**}_{\alpha} | X\in \mathcal{C}) + P_0(C^{**}_{\alpha} \geq s^*(X) \geq  C^{**}_{\alpha} - (- \bar{\Delta} R_{0,\mathcal{C}} / M_1)^{1/\uderbar{\gamma}} | X\in \mathcal{C} ) \\
&\geq  P_0(s^*(X) > C^{**}_{\alpha} | X\in \mathcal{C})  + |\bar{\Delta} R_{0,\mathcal{C}} |\text{ }\text{ }\text{ }\text{  (by detection condition) }\\
&=   P_0(\hat s(X) >  \widehat C_{\alpha}|X\in\mathcal{C})\\
&\geq  P_0(s^*(X) > \widehat C_{\alpha} + T | X\in \mathcal{C})\,.
\end{align*}
So when $\bar{\Delta} R_{0, \mathcal{C}} < 0$, $\widehat C_{\alpha} \geq C^{**}_{\alpha} - (- \bar{\Delta} R_{0,\mathcal{C}} / M_1)^{1/\uderbar{\gamma}} - T$.  Note that $\Delta R_{0, \mathcal{C}} =| \bar{\Delta} R_{0, \mathcal{C}}|$. Therefore we have in both cases, 
$$
\widehat C_{\alpha} \geq C^{**}_{\alpha} - (\Delta R_{0, \mathcal{C}}/M_0)^{1/\bar{\gamma}} \wedge (\Delta R_{0, \mathcal{C}}/M_1)^{1/\uderbar{\gamma}} - T\,.
$$
%Because $\uderbar{\gamma}\geq \bar{\gamma} = 1$, and when $\Delta R_{0, \mathcal{C}}/ (M_0 \vee M_1) <1$ (as it is on the event $\mathcal{E}_2$),   the previous inequality implies that 
%$$
%\widehat{C}_{\alpha}\geq C^{**}_{\alpha} - (\Delta R_{0, \mathcal{C}}/(M_0 \vee M_1))^{1/\uderbar{\gamma}} - T\,.
%$$

Using the above inequality, we have  
\begin{align*}
&\quad\,\, \mathcal{C}\cap (G^*  \backslash \widehat G) \\
&= \mathcal{C}\cap \{ s^* \leq  C^{**}_{\alpha}, \hat s > \widehat C_{\alpha} \} \\
&= \mathcal{C}\cap\{s^* \leq  C^{**}_{\alpha}, \hat s \geq  C^{**}_{\alpha} - (\Delta R_{0, \mathcal{C}}/M_0)^{1/\bar{\gamma}} \wedge (\Delta R_{0, \mathcal{C}}/M_1)^{1/\uderbar{\gamma}} - T \} \cap \{\hat s > \widehat C_{\alpha} \}\\
&\subset  \mathcal{C}\cap \{C^{**}_{\alpha} - (\Delta R_{0, \mathcal{C}}/M_0)^{1/\bar{\gamma}} \wedge (\Delta R_{0, \mathcal{C}}/M_1)^{1/\uderbar{\gamma}} - 2T \leq s^* \leq C^{**}_{\alpha}\}\cap \{\hat s \geq  \widehat C_{\alpha} \}\\
&\subset  \mathcal{C}\cap\{C^{**}_{\alpha} - (\Delta R_{0, \mathcal{C}}/M_0)^{1/\bar{\gamma}} \wedge (\Delta R_{0, \mathcal{C}}/M_1)^{1/\uderbar{\gamma}} - 2T \leq s^*  \leq  C^{**}_{\alpha}\}\,.
\end{align*}
Denote by $\mathcal{C}_2 = \{x: C^{**}_{\alpha} - (\Delta R_{0, \mathcal{C}}/M_0)^{1/\bar{\gamma}} \wedge (\Delta R_{0, \mathcal{C}}/M_1)^{1/\uderbar{\gamma}} - 2T \leq s^*(x)  \leq C^{**}_{\alpha}\}$. Then we just showed that $\mathcal{C}\cap (G^*  \backslash \widehat G)  \subset \mathcal{C}\cap \mathcal{C}_2$.    
Recall that 
$$
|r(x) - C_{\alpha}| = e^{-\mu_a^{\top} \Sigma^{-1}\mu_d}|e^{s^*(x)} - e^{C^{**}_{\alpha}}| = e^{-\mu_a^{\top} \Sigma^{-1}\mu_d}\cdot e^{z'} |s^*(x) - C^{**}_{\alpha}|\,,
$$
where 
$z'$ is some quantity between $s^*(x)$ and $C^{**}_{\alpha}$.  %We can treat $\mu_a\Sigma^{-1}\mu_d$ as if it is a constant (or bounded from below).  
%
%\textcolor{blue}{The problem is how do we deal with the assumption on $C^{**}_{\alpha}$ Suppose $C^{**}_{\alpha}$ is bounded from above. Further more, as $C^{**}_{\alpha} = \log C_{\alpha} + \mu_a \Sigma^{-1}\mu_d$, this restriction can be put as more transparent assumptions.} %\textcolor{red}{The s should be changed to s star, because the s is also used to denote the sparsity}
Restricting to  
$\mathcal{C}\cap \mathcal{C}_2$, we have 
$$
z'\leq C^{**}_{\alpha}\,.
$$
This together with  $\mathcal{C}\cap (G^*  \backslash \widehat G)  \subset \mathcal{C}\cap \mathcal{C}_2$  implies that 
\begin{align*}
(\text{I}')&\leq  \int_{\mathcal{C}\cap \mathcal{C}_2}|r - C_{\alpha}|dP_0\\
& =  \int_{\mathcal{C}\cap \mathcal{C}_2}\exp\{z' - \mu_a \Sigma^{-1} \mu_d\}|s^*(x) - C^{**}_{\alpha}|dP_0\\
&\leq  \int_{\mathcal{C}\cap \mathcal{C}_2} c''|s^*(x) - C^{**}_{\alpha}|dP_0\\
&\leq  c'' \left((\Delta R_{0, \mathcal{C}}/M_0)^{1/\bar{\gamma}} \wedge (\Delta R_{0, \mathcal{C}}/M_1)^{1/\uderbar{\gamma}} + 2T \right) P_0(\mathcal{C}\cap \mathcal{C}_2)\,.
\end{align*}
Note that by the margin assumption, 
\begin{align*}
&\quad P_0(\mathcal{C}\cap \mathcal{C}_2)\\
\leq & \quad P_0(C^{**}_{\alpha} - (\Delta R_{0, \mathcal{C}}/M_0)^{1/\bar{\gamma}} \wedge (\Delta R_{0, \mathcal{C}}/M_1)^{1/\uderbar{\gamma}} - 2T 
\leq  s^*(X)  \leq C^{**}_{\alpha})\\
\leq & \quad  M_0\left((\Delta R_{0, \mathcal{C}}/M_0)^{1/\bar{\gamma}} \wedge (\Delta R_{0, \mathcal{C}}/M_1)^{1/\uderbar{\gamma}} + 2T\right)^{\bar\gamma}\,.
\end{align*}
Therefore,
$$
(\text{I}')\leq c'' M_0 \left((\Delta R_{0, \mathcal{C}}/M_0)^{1/\bar{\gamma}} \wedge (\Delta R_{0, \mathcal{C}}/M_1)^{1/\uderbar{\gamma}} + 2T\right)^{1+\bar\gamma}\,.
$$
Regarding (II$'$), by Lemma \ref{lem:large_prob_set} we have 
$$
(\text{II}')\leq \int_{\mathcal{C}^c} rdP_0 + C_{\alpha}\int_{\mathcal{C}^c}dP_0 = P_1(\mathcal{C}^c) + C_{\alpha} P_0(\mathcal{C}^c)\leq (1 + C_{\alpha}) \exp\{- (n_0 \wedge n_1)^{1/2}\}\,.
$$

Therefore, by the excess type II error decomposition equation \eqref{eqn:typeII decomposition},
\begin{align*}
P_1(\widehat G) - P_1(G^*) = (\text{I}) + (\text{II}) + (\text{I}') + (\text{II}') + C_{\alpha}\{ R_0(\phi^*_{\alpha}) - R_0(\hat \phi_{k^*}) \}\,.
\end{align*}
Using the upper bounds for (I), (II), (I$'$) and (II$'$) and Lemma \ref{lem:2}, With probability at least $1-\delta_0 - \delta_0' - \delta_1 -\delta_2$, we have

\begin{align*}
P_1(\widehat G) - P_1(G^*)\leq &\quad  c' M_0 \left((\Delta R_{0,\mathcal{C}} / M_1)^{1/\uderbar{\gamma}} + 2T\right)^{1+\bar{\gamma}} \\
&+ c'' M_0 \left((\Delta R_{0, \mathcal{C}}/M_0)^{1/\bar{\gamma}} \wedge (\Delta R_{0, \mathcal{C}}/M_1)^{1/\uderbar{\gamma}} + 2T\right)^{1+\bar\gamma} \\
&+ 2(1 + C_{\alpha}) \exp\{- (n_0 \wedge n_1)^{1/2}\} + C_{\alpha}\cdot \xi_{\alpha, \delta_0,n_0'}(\delta_0') \\
\leq &\quad   \bar{c}_1' \Delta R_{0, \mathcal{C}}^{(1+\bar{\gamma})/\uderbar{\gamma}} + \bar{c}_2' T^{1 + \bar\gamma} + \bar{c}_3' \exp\{-(n_0 \wedge n_1)^{1/2} \}+ C_{\alpha}\cdot \xi_{\alpha, \delta_0,n_0'}(\delta_0') \,,
\end{align*}
for some positive constants $\bar{c}_1'$, $\bar{c}_2'$ and $\bar{c}_3'$. In the last inequality of the above chain, we used $\uderbar{\gamma}\geq \bar{\gamma}$.  
Note that on the event $\mathcal{E}_2 \cap \mathcal{E}_3$,  Lemma \ref{lem:4} guarantees  $\Delta R_{0, \mathcal{C}}\leq 2[\xi_{\alpha, \delta_0,n_0'}(\delta_0') + \exp\{- (n_0 \wedge n_1)^{1/2}\}]$. Lemma \ref{lem:large_prob_set} guarantees that $T \leq 4c_1' \varphi \lambda s (n_0\wedge n_1)^{1/4}$.  Therefore, 
\begin{align*}
P_1(\widehat G) - P_0(G^*)
\leq &\quad \bar{c}_1'' \xi_{\alpha, \delta_0,n_0'}(\delta_0')^{(1+\bar{\gamma})/\uderbar{\gamma}\wedge 1} + \bar{c}_2'' [4c_1' \varphi \lambda s (n_0\wedge n_1)^{1/4}]^{1 + \bar\gamma} \\
&+ \bar{c}_3'' [\exp\{-(n_0 \wedge n_1)^{\frac{1}{2}}\}]^{(1+\bar{\gamma})/\uderbar{\gamma}} \,.	
\end{align*}

Lemma \ref{prop::R0} guarantees that $
\xi_{\alpha, \delta_0,n_0'}(\delta_0') \leq  ({5}/{2}){(n_0')^{-1/4}}.
$
Then the excess type II error is bounded by
\begin{eqnarray*}
P_1(\widehat G) - P_0(G^*)
&\leq& \bar{c}_1 (n_0')^{- (\frac{1}{4}\wedge \frac{1+\bar\gamma}{4\uderbar{\gamma}})} + \bar{c}_2  (\lambda s)^{1 + \bar\gamma} (n_0\wedge n_1)^{\frac{1 + \bar\gamma}{4}} \\
&+& \bar{c}_3 \exp\left\{-(n_0\wedge n_1)^{\frac{1}{2}}(\frac{1+\bar\gamma}{\uderbar{\gamma}}\wedge 1)\right\}.	
\end{eqnarray*}
\end{proof}
%\textcolor{red}{Need to fix typos and format in Bibliography}
%\bibliographystyle{plainnat}
\bibliography{reference}

\end{document}